\documentclass[a4paper,onecolumn, superscriptaddress,10pt,accepted=2022-11-07,issue=5, volume=5, shorttitle=papers/compositionality-5-5]{compositionalityarticle}
\pdfoutput=1

\usepackage[utf8]{inputenc}
\usepackage[english]{babel}
\usepackage{amsthm,amsmath}
\usepackage{stmaryrd}
\usepackage{graphicx}
\usepackage{keycommand}
\usepackage{hyperref}
\usepackage[all]{hypcap}
\usepackage{multicol}

\usepackage[numbers,sort&compress]{natbib}

\usepackage{threeparttable}

\theoremstyle{definition}
\newtheorem{theorem}{Theorem}[section]
\newtheorem{corollary}[theorem]{Corollary}
\newtheorem{lemma}[theorem]{Lemma}
\newtheorem{proposition}[theorem]{Proposition}

\newtheorem{definition}[theorem]{Definition}

\newtheorem{example}[theorem]{Example}

\newtheorem{example*}[theorem]{Example*}
\newtheorem{examples*}[theorem]{Examples*}
\newtheorem{remark}[theorem]{Remark}
\newtheorem{remark*}[theorem]{Remark*}

\newtheorem*{theorem*}{Theorem}
\newtheorem*{corollary*}{Corollary}
\newtheorem*{lemma*}{Lemma}
\newtheorem*{proposition*}{Proposition}

\usepackage{tikzit}
\input{zh.tikzdefs}

\tikzstyle{dot}=[inner sep=0.3mm, minimum width=2mm, minimum height=2mm, draw, shape=circle, font={\footnotesize}, tikzit fill=magenta]
\tikzstyle{white dot}=[dot, fill=white, text depth=-0.2mm, tikzit category=ZH-pf]
\tikzstyle{gray dot}=[dot, fill={rgb,255: red,180; green,180; blue,180}, text depth=-0.2mm, tikzit category=ZH-pf]
\tikzstyle{gray phase dot}=[gray dot, fill={rgb,255: red,180; green,180; blue,180}, tikzit fill=magenta]
\tikzstyle{hadamard}=[fill=white, draw, inner sep=0.6mm, minimum height=1.5mm, minimum width=1.5mm, shape=rectangle, tikzit shape=rectangle, tikzit category=ZH-pf]
\tikzstyle{small hadamard}=[fill=white, draw, inner sep=0.6mm, minimum height=1.5mm, minimum width=1.5mm, tikzit shape=rectangle]
\tikzstyle{halfscalar}=[star, fill=black, draw=black, minimum size=8pt, inner sep=0pt]
\tikzstyle{box}=[shape=rectangle, text height=1.5ex, text depth=0.25ex, yshift=0.2mm, fill=white, draw=black, minimum height=3mm, minimum width=5mm, font={\small}]
\tikzstyle{Z dot}=[inner sep=0mm, minimum size=2mm, shape=circle, draw=black, fill=zx_green]
\tikzstyle{Z phase dot}=[minimum size=5mm, font={\footnotesize\boldmath}, shape=rectangle, rounded corners=2mm, inner sep=0.2mm, outer sep=-2mm, scale=0.8, tikzit shape=circle, draw=black, fill=zx_green, tikzit draw=blue]
\tikzstyle{X dot}=[Z dot, shape=circle, draw=black, fill=zx_red]
\tikzstyle{X phase dot}=[Z phase dot, tikzit shape=circle, tikzit draw=blue, fill=zx_red, font={\footnotesize\color{black}\boldmath}]
\tikzstyle{zxnode}=[inner sep=0.3mm, minimum width=2mm, minimum height=2mm, draw, shape=circle, font={\footnotesize}]

\tikzstyle{gn}=[zxnode, fill=zx_green]
\tikzstyle{rn}=[zxnode, fill=zx_red]
\tikzstyle{H box}=[rectangle, fill=yellow, draw=black, xscale=1, yscale=1, font={\small}, inner sep=0.75pt, minimum width=0.15cm, minimum height=0.15cm, tikzit shape=rectangle]
\tikzstyle{ug}=[regular polygon, regular polygon sides=3, fill=zx_red, draw=black, inner sep=0pt, minimum width=1em, tikzit draw=blue]
\tikzstyle{st}=[star, star points=5, fill=white, draw=black, inner sep=1.2pt, line width=1.2pt, tikzit fill=blue, tikzit draw=red, tikzit category=ZH-pf]
\tikzstyle{triangle}=[regular polygon, regular polygon sides=3, fill=white, draw=black, inner sep=0pt, minimum width=1em, tikzit draw=blue, tikzit category=ZH-pf]
\tikzstyle{not}=[fill={rgb,255: red,180; green,180; blue,180}, draw=black, shape=circle, font={$\neg$}, dot]

\tikzstyle{bbindex}=[font={\color{blue}\footnotesize}]
\tikzstyle{wide point}=[fill=white,draw,shape=isosceles triangle,shape border rotate=-90,isosceles triangle stretches=true,inner sep=0pt,minimum width=1.5cm,minimum height=6.12mm,yshift=-0.0mm]
\tikzstyle{medium gray box}=[semilarge box,fill={rgb,255: red,180; green,180; blue,180}]
\tikzstyle{small box}=[rectangle,inline text,fill=white,draw,minimum height=5mm,yshift=-0.5mm,minimum width=5mm,font=\small]
\tikzstyle{small gray box}=[small box,fill={rgb,255: red,180; green,180; blue,180}]
\tikzstyle{medium box}=[rectangle,inline text,fill=white,draw,minimum height=5mm,yshift=-0.5mm,minimum width=8mm,font=\small]
\tikzstyle{wire label}=[font=\footnotesize, tikzit fill=blue, anchor=west, shape=rectangle, inner sep=1pt, xshift=-1mm]

\tikzstyle{gray}=[-, draw={blue!60!white}, tikzit draw=blue]
\tikzstyle{blue}=[-, draw={blue!60!white}, tikzit draw=blue]
\tikzstyle{brace edge}=[-, tikzit draw=blue, decorate, decoration={brace,amplitude=1mm,raise=-1mm}]
\tikzstyle{diredge}=[->]
\tikzstyle{not edge}=[-, dashed, dash pattern=on 2pt off 1.5pt, thick, draw={rgb,255: red,255; green,68; blue,68}]

\tikzstyle{Z}=[white dot]
\tikzstyle{wire}=[none]
\tikzstyle{X}=[gray dot]
\tikzstyle{bbox}=[]
\tikzstyle{simple}=[]
\tikzstyle{quanto}=[]

\makeatletter
\newcommand\etc{etc\@ifnextchar.{}{.\@}\xspace}

\makeatother

\newcommand{\intf}[1]{\left\llbracket #1 \right\rrbracket} 

\usepackage{bm}

\newcommand{\bra}[1]{\ensuremath{\left\langle #1 \right|}\xspace}
\newcommand{\ket}[1]{\ensuremath{\left|  #1 \right\rangle}\xspace}
\newcommand{\braket}[2]{\ensuremath{\langle#1|#2\rangle}\xspace}
\newcommand{\ketbra}[2]{\ensuremath{\ket{#1}\!\bra{#2}}\xspace}
\newcommand{\abs}[1]{\ensuremath{|#1|}\xspace}
\newcommand{\C}{\mathbb{C}}

\newcommand{\N}{\mathbb{N}}

\newcommand{\zh}{\text{ZH}\xspace}
\newcommand{\ZHR}{\ensuremath{\text{ZH}_R}\xspace}
\newcommand{\ZHRstar}{\ensuremath{\text{ZH}_R^{\dstar}}\xspace}
\newcommand{\ZHRstars}{\ZHRstar}
\newcommand{\RING}[1]{\ensuremath{\textsc{ring}_#1}}
\newcommand{\ZHRhalf}{\ensuremath{\text{ZH}_{R[\half]}}\xspace}
\newcommand{\ZH}{\ensuremath{\text{ZH}}\xspace}
\newcommand{\ZHC}{\ensuremath{\text{ZH}_\mathbb{C}}\xspace}
\newcommand{\half}{\ensuremath{\frac{1}{2}}\xspace}

\newcommand{\tensor}{\ensuremath{\otimes}\xspace}

\tikzstyle{dotpic}=[] 
\tikzstyle{semilarge box}=[rectangle,inline text,fill=white,draw,minimum height=5mm,yshift=-0.5mm,minimum width=12.5mm,font=\small]
\tikzstyle{inline text}=[text height=1.5ex, text depth=0.25ex,yshift=0.5mm]
\tikzstyle{label}=[font=\footnotesize,text height=1.5ex, text depth=0.25ex]
\tikzstyle{white phase dot}=[white dot]

\newcommand{\dotonly}[1]{%
\,\begin{tikzpicture}[dotpic]
\node [#1] (a) at (0,0) {};
\end{tikzpicture}\,}

\newcommand{\dotcounit}[1]{%
\,\begin{tikzpicture}[dotpic,yshift=-1mm]
\node [#1] (a) at (0,0.35) {}; 
\draw (0,-0.3)--(a);
\end{tikzpicture}\,\xspace}
\newcommand{\dotunit}[1]{%
\,\begin{tikzpicture}[dotpic,yshift=1.5mm]
\node [#1] (a) at (0,-0.35) {}; 
\draw (a)--(0,0.3);
\end{tikzpicture}\,\xspace}

\newcommand{\dotmult}[1]{%
\,\begin{tikzpicture}[dotpic]
    \node [#1] (a) {};
    \draw (a) -- (90:0.55);
    \draw (a) (-45:0.6) -- (a);
    \draw (a) (-135:0.6) -- (a);
\end{tikzpicture}\,\xspace}

\newcommand{\spider}[2]{
    \begin{tikzpicture}
        \begin{pgfonlayer}{nodelayer}
            \node [style=none] (0) at (-0.75, 1) {};
            \node [style=none] (1) at (0, 0.75) {$\ldots$};
            \node [style=none] (2) at (0.75, 1) {};
            \node [style=#1] (3) at (0, -0) {$#2$};
            \node [style=none] (4) at (-0.75, -1) {};
            \node [style=none] (5) at (0, -0.75) {$\ldots$};
            \node [style=none] (6) at (0.75, -1) {};
            \node [style=none] (7) at (-1, 1.25) {};
            \node [style=none] (8) at (1, 1.25) {};
            \node [style=none] (9) at (-1, -1.25) {};
            \node [style=none] (10) at (1, -1.25) {};
        \end{pgfonlayer}
        \begin{pgfonlayer}{edgelayer}
            \draw [bend left=15, looseness=1.00] (3) to (6.center);
            \draw [bend right=15, looseness=1.00] (3) to (4.center);
            \draw [bend right=15, looseness=1.00] (0.center) to (3);
            \draw [bend right=15, looseness=1.00] (3) to (2.center);
        \end{pgfonlayer}
    \end{tikzpicture}
}

\newcommand{\bracedSpider}[2]{
    \begin{tikzpicture}[decoration=brace]
        \begin{pgfonlayer}{nodelayer}
            \node [style=none] (0) at (-0.75, 1) {};
            \node [style=none] (1) at (0, 0.75) {$\ldots$};
            \node [style=none] (2) at (0.75, 1) {};
            \node [style=#1] (3) at (0, -0) {$#2$};
            \node [style=none] (4) at (-0.75, -1) {};
            \node [style=none] (5) at (0, -0.75) {$\ldots$};
            \node [style=none] (6) at (0.75, -1) {};
            \node [style=none] (7) at (-1, 1.25) {};
            \node [style=none] (8) at (1, 1.25) {};
            \node [style=none] (9) at (-1, -1.25) {};
            \node [style=none] (10) at (1, -1.25) {};
            \node [style=none] (11) at (0, 1.75) {$n$};
            \node [style=none] (12) at (0, -1.75) {$m$};
        \end{pgfonlayer}
        \begin{pgfonlayer}{edgelayer}
            \draw [bend left=15, looseness=1.00] (3) to (6.center);
            \draw [bend right=15, looseness=1.00] (3) to (4.center);
            \draw [bend right=15, looseness=1.00] (0.center) to (3);
            \draw [bend right=15, looseness=1.00] (3) to (2.center);
            \draw [decorate] (7.center) to (8.center);
            \draw [decorate] (10.center) to (9.center);
        \end{pgfonlayer}
    \end{tikzpicture}
}

\newcommand{\scalar}[2]{
    \,
\begin{tikzpicture}
    \begin{pgfonlayer}{nodelayer}
        \node [style=#1] (2) at (0, -0) {$#2$}; 
    \end{pgfonlayer}
\end{tikzpicture}
\,
}

\newkeycommand{\phase}[style=white phase dot][1]{\,\begin{tikzpicture}
    \begin{pgfonlayer}{nodelayer}
        \node [style=none] (0) at (0, 0.6) {};
        \node [style=\commandkey{style}] (2) at (0, -0) {$#1$}; 
        \node [style=none] (3) at (0, -0.6) {};
    \end{pgfonlayer}
    \begin{pgfonlayer}{edgelayer}
        \draw (2) to (0.center);
        \draw (3.center) to (2);
    \end{pgfonlayer}
\end{tikzpicture}\,}

\newcommand{\gendiagram}[1]{
\begin{tikzpicture}
    \begin{pgfonlayer}{nodelayer}
        \node [style=none] (0) at (-0.75, 1) {};
        \node [style=none] (1) at (0.75, 1) {};
        \node [style=none] (2) at (-0.75, -1) {};
        \node [style=none] (3) at (0.75, -1) {};
        \node [style=none] (4) at (0, 0.75) {$\ldots$};
        \node [style=semilarge box] (5) at (0, -0) { #1 };
        \node [style=none] (6) at (0, -0.75) {$\ldots$};
    \end{pgfonlayer}
    \begin{pgfonlayer}{edgelayer}
        \draw (0.center) to (2.center);
        \draw (1.center) to (3.center);
    \end{pgfonlayer}
\end{tikzpicture}
}

\newcommand{\namedeq}[1]{\,\begin{tikzpicture}
  \begin{pgfonlayer}{nodelayer}
    \node [style=none] (0) at (0, 1) {#1};
    \node [style=none] (1) at (0, 0) {$=$};
  \end{pgfonlayer}
\end{tikzpicture}\,}

\newcommand{\graymult}{\dotmult{gray dot}}

\newcommand{\grayphase}[1]{\phase[style=gray phase dot]{#1}}
\newcommand{\whiteunit}{\dotunit{white dot}}
\newcommand{\whitecounit}{\dotcounit{white dot}}
\newcommand{\whitemult}{\dotmult{white dot}}
\newcommand{\whitedot}{\dotonly{white dot}\xspace}
\newcommand{\whitephase}[1]{\phase[style=white dot]{#1}}
\newcommand{\dstar}{\scalar{st}{}}

\begin{document}
\title{Completeness of the ZH-calculus}
\date{}

\author{Miriam Backens}
\email{m.backens@bham.ac.uk}
\affiliation{School of Computer Science, University of Birmingham, Edgbaston, Birmingham B15 2TT, UK} 

\author{Aleks Kissinger}
\email{aleks0@gmail.com}
\affiliation{Department of Computer Science, University of Oxford, Wolfson Building, Parks Road, Oxford OX1 3QD, UK} 

\author{Hector Miller-Bakewell}
\email{hmillerbakewell@gmail.com}
\homepage{https://hjmb.co.uk/}
\affiliation{Department of Computer Science, University of Oxford, Wolfson Building, Parks Road, Oxford OX1 3QD, UK} 

\author{John van de Wetering}
\email{john@vdwetering.name}
\homepage{http://vdwetering.name}
\affiliation{Institute for Computing and Information Sciences, Radboud Universiteit, Toernooiveld 212, 6525 EC Nijmegen, NL} 
\affiliation{Department of Computer Science, University of Oxford, Wolfson Building, Parks Road, Oxford OX1 3QD, UK} 

\author{Sal Wolffs}
\email{sal.wolffs@gmail.com}
\affiliation{Institute for Computing and Information Sciences, Radboud Universiteit, Toernooiveld 212, 6525 EC Nijmegen, NL} 

\begin{abstract}
	There are various gate sets used for describing quantum computation. A particularly popular one consists of Clifford gates and arbitrary single-qubit phase gates. Computations in this gate set can be elegantly described by the \emph{ZX-calculus}, a graphical language for a class of string diagrams describing linear maps between qubits. The ZX-calculus has proven useful in a variety of areas of quantum information, but is less suitable for reasoning about operations outside its natural gate set such as multi-linear Boolean operations like the Toffoli gate.
	In this paper we study the \emph{ZH-calculus}, an alternative graphical language of string diagrams that does allow straightforward encoding of Toffoli gates and other more complicated Boolean logic circuits.
 	We find a set of simple rewrite rules for this calculus and show it is complete with respect to matrices over $\mathbb Z[\frac12]$, which correspond to the approximately universal Toffoli+Hadamard gateset.
 	Furthermore, we construct an extended version of the ZH-calculus that is complete with respect to matrices over any ring $R$ where $1+1$ is not a zero-divisor.
\end{abstract} 

\tableofcontents

\maketitle
\section{Introduction}

Graphical calculi give us a compact way to express and reason about complex, interacting processes using \textit{string diagrams}. String diagrams represent processes that can compose in sequence or parallel using a notation consisting of boxes or nodes connected by wires. Notably, wires can be left `open' at one or both ends to indicate inputs and outputs of the composed process. For example, depicting sequential composition as $\circ$ and parallel composition as $\otimes$, we can translate expressions such as the following one into a string diagram:
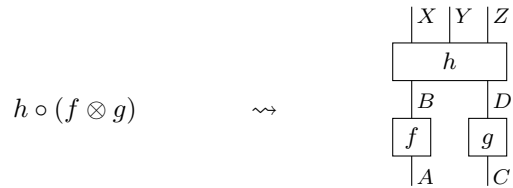
\[
  h \circ (f \otimes g) \qquad \qquad \leadsto \qquad \qquad
  \input{string-diag-example.tikz}
\]
Such notation has proven convenient for expressing compositions and tensor products of linear maps (as used extensively in quantum theory), and more generally for expressing morphisms in a generic symmetric monoidal category.

A \textit{graphical calculus} is both a collection of graphical building-blocks and a collection of equations between string diagrams that we can use to transform one diagram into another, typically equivalent, one. 
One of the best-studied graphical calculi is the \textit{ZX-calculus}~\cite{CD1,CD2,CKbook,vandewetering2020zxcalculus}, which has been applied extensively in the study of quantum circuits and related structures in quantum computation (e.g.\ measurement-based quantum computing~\cite{DP2,Backens2020extraction,kissinger2017MBQC}, fault-tolerant quantum computations~\cite{horsman2017surgery,hanks2019effective,magicFactories}, error-correcting codes~\cite{chancellor2016graphical}, circuit optimisation~\cite{FaganDuncan,cliff-simp,optimisation-paper}, and classical simulation~\cite{kissinger2021simulating,kissinger2022classical}).
ZX-diagrams consist of \emph{spiders}, a type of well-behaved linear map that is depicted by coloured dots:
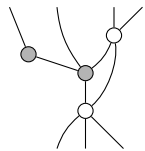
\ctikzfig{example-ZX-diagram}

Quantum circuits lie at the heart of the \textit{circuit model} of quantum computation, which represents the quantum part of the computation as a large unitary operator arising from compositions and tensor products of basic gates~\cite{NielsenChuang}. There are many choices of basic gates which are \textit{universal}, in the sense that they can construct or approximate arbitrary $2^n \times 2^n$ unitary matrices. A popular choice is the \textit{Clifford+phase} set of gates:
\[
  \text{CNOT} \ :=\ \begin{pmatrix}
    1 & 0 & 0 & 0 \\
    0 & 1 & 0 & 0 \\
    0 & 0 & 0 & 1 \\
    0 & 0 & 1 & 0
  \end{pmatrix}
  \qquad\qquad
  H \ := \ \frac{1}{\sqrt 2} \begin{pmatrix}
    1 & 1 \\ 1 & -1
  \end{pmatrix}
  \qquad\qquad
  Z_{\alpha} \ := \ \begin{pmatrix}
    1 & 0 \\ 0 & e^{i \alpha}
  \end{pmatrix}
\]
This set of gates is universal, here in the sense that it can exactly express any unitary matrix. A common finite restriction of this infinite set is the \textit{Clifford+T} gate set, which replaces $Z_\alpha$ with $T := Z_{\pi/4}$. This set of gates is able to approximate any unitary up to arbitrary finite precision. A further restriction to the \textit{Clifford} gate set $(\text{CNOT}, H, S := T^2)$ is no longer universal, and in fact yields only
circuits that can be efficiently simulated on a classical computer~\cite{aaronsongottesman2004}. Nevertheless, Clifford circuits play an important role in quantum computation~\cite{NielsenChuang}, quantum error correction~\cite{fowler2012surface,gottesman2010introduction,BB84}, and many quantum communication protocols.

Notably, these three important families of circuits map straightforwardly to three fragments of the ZX-calculus: where parameters are restricted to integer multiples of $\frac\pi 2$ (Clifford), restricted to integer multiples of $\frac \pi 4$ (Clifford+T), or are unrestricted (Clifford+phase). Furthermore, each of these fragments of the ZX-calculus has been proven \textit{complete} in the sense that two ZX-diagrams which evaluate to the same linear operator are provably equal using just diagram rewrite rules~\cite{Backens1,SimonCompleteness,ng2017universal}.

However, the ZX-calculus remains closely tied to the structure of the Clifford+phase family of circuits. The further one gets from this family, the more awkward it becomes to work with computations using the ZX-calculus. For instance, the CNOT gate (`controlled-NOT') can be expressed quite simply in the ZX-calculus:
\[
  \left\llbracket
  \ \input{cnot-intro.tikz}\ 
  \right\rrbracket \ =\ \textit{CNOT}
\]
In contrast, the `controlled-controlled-NOT', commonly referred to as the Toffoli gate, is considerably more complicated, with typical presentations requiring many more generators which cannot be easily manipulated using the ZX-calculus rules~\cite[Ex.~12.10]{CKbook}.

The problem with the Toffoli gate is that it acts on basis states in a manner that is not $\mathbb Z_2$-affine. 
For example, the Boolean CNOT is the map $\text{CNOT}:\{0,1\}^2\rightarrow \{0,1\}^2$ given by $\text{CNOT}(x,y) = (x,x\oplus y)$, where $\oplus$ denotes XOR.
Such a Boolean function is called \emph{affine}, in analogy to the theory of polynomials, since it only contains terms of multiplicative degree at most 1 that are then added together (XOR acts as addition on the field $\mathbb Z_2$).
Similarly, we say the NOT gate $NOT(x) = x\oplus 1$ is affine, since again it is a sum of terms of degree at most 1.
In the ZX-calculus there are two main types of generators: Z-spiders and X-spiders.   
Both of these are closely related to affine Boolean functions: the Z-spider is based on the COPY Boolean map $\text{COPY}(x) = (x,x)$, while the X-spider is based on XOR. 
As being affine is baked into the generators, some tricks are required to represent non-affine Boolean functions in the ZX-calculus,
such as the Toffoli gate $\text{TOF}(x,y,z) = (x,y,(x\cdot y)\oplus z)$ that contains the degree-2 term $x\cdot y$.

The way we solve this issue in this paper is by introducing a third diagrammatic type of generator beyond the Z and X spiders: the H-box.
An $n$-input $0$-output H-box is defined by $\ket{x_1\cdots x_n} \mapsto (-1)^{x_1\cdot \ldots \cdot x_n}$.
This can then be used in combination with Z-spiders (i.e.\ copy maps) to represent the `controlled-controlled-Z' gate, and thus the Toffoli gate.

Another way to view the H-box generator is as a generalisation of the Hadamard gate to a linear map with an arbitrary number of inputs and outputs. 

With this new generator comes a new set of graphical rewrite rules that allow us to reason more efficiently about quantum computation involving non-affine Boolean functions.
The resulting graphical language is called the \emph{ZH-calculus}, since its main generators are Z-spiders and H-boxes.
The main topic of this paper is to show that various fragments of the ZH-calculus are complete. 
Specifically, we find a set of particularly simple and compelling rewrite rules (see Figure~\ref{fig:phasefree-rules}) that are complete for the approximately universal Toffoli-Hadamard gate set.%
\footnote{The Toffoli-Hadamard gate set is only approximately universal for `real' quantum computing, where all unitaries contain only real numbers. However, such real unitaries can simulate with constant overhead unitaries with complex entries~\cite{shi2003toffoli}. The notion of approximate universality is here hence different than that of Clifford+$T$, and might better be called `computationally universal'. We will however not make this technical distinction in the remainder of the paper.} 
Arguably, this is the smallest complete rule set for an approximately universal fragment of quantum computing (we discuss this claim in more detail in Section~\ref{sec:motivation}).
Building on this work, we show that for any ring where $2:=1+1$ is not a zero divisor,
we can generalise the ZH-calculus in order to express all matrices over this ring,
and we exhibit a complete set of rules for this generalisation as well.

This article is based on a conference paper by Backens and Kissinger~\cite{backens2018zhcalculus} which first introduced the ZH-calculus and proved completeness for the ZH-calculus with complex parameters,
an unpublished preprint of van de Wetering and Wolffs~\cite{PhasefreeZH2019} which proved completeness for the phase-free fragment of the calculus, and the DPhil thesis of Miller-Bakewell~\cite{millerbakewell2020thesis} which showed how the ZH-calculus and its completeness could be extended to arbitrary rings.
The formal developments in those papers have been expanded and unified, and the technical content has been extended significantly in three ways: we find a simpler complete rule set for the phase-free fragment, we give a completely self-contained proof of phase-free completeness based on an encoding of arithmetic into ZH, and we extend the ZH-calculus to any ring where $2 := 1 + 1$ is not a zero divisor.

\subsection{Related work}
In the earliest complete version of the ZX-calculus, a new generator, the `triangle', was introduced in order to represent non-affine Boolean functions~\cite{ng2017universal}. 
This has later also been used to prove completeness of the fragment corresponding to the Toffoli+Hadamard gate set that we also study~\cite{vilmart2018zxtriangle}, as well as to find an axiomatisation of the ZX-calculus over arbitrary rings~\cite{wang2020algebraic}.
In~\cite{ZXand},
a complete ruleset for classical circuits, i.e.\ natural number matrices, was found 
which uses almost the same ruleset as we do in the phase-free fragment. The difference is
that ~\cite{ZXand} treats 
the AND gate as an atomic non-symmetric generator, while we decompose it into H-boxes.\footnote{In fact, his rules were inspired by the rules of the phase-free ZH-calculus presented in an earlier preprint~\cite{PhasefreeZH2019} (personal communication, Comfort).}
The seminal complete graphical language for the related fragment of integer matrices is the \emph{ZW-calculus}~\cite{hadzihasanovic2015diagrammatic} which has generators based on the two different types of entanglement that are possible in tripartite qubit systems~\cite{CK}.
The related language \RING{R}\ \cite{millerbakewell2020thesis} was designed to
unify the study of qubit graphical calculi parameterised by phase rings,
and indeed both parameterised ZH (Section~\ref{sec:zh-ring})
and the interpretation of ZH as operations on Boolean functions (Section~\ref{sec:motivation}) fit into this scheme.
Therefore there are natural maps from \RING{\mathbb{B}}\ to ZH and
from \RING{R}\ to \ZHR, the second of which is exhibited in
\cite{millerbakewell2020thesis}. 
Note that
\cite{carette2020recipe} 
shows that, under certain assumptions, there are only three basic graphical calculi for qubits: ZX, ZW, and ZH.
The results presented in the current paper prove that the youngest of these calculi, ZH, is complete, concluding the program of finding complete graphical calculi for qubits.

Using the ZH-calculus, an efficient description of \emph{hypergraph states} can be given~\cite{Lemonnier2020hypergraph}, which allows for a description of ZH-diagrams in terms of \emph{path-sums}~\cite{AmyVerification,pathsRenaud,Vilmart2022Completeness}. There is also a close connection between ZH-diagrams and quantum multiple-valued decision diagrams~\cite{vilmart2021quantum}. 
The ZH-calculus has been used to give a graphical description of \emph{AKLT states}, a particularly canonical type of condensed matter system~\cite{east2020akltstates}, and of spin-networks~\cite{d.p.east2021spinnetworks}.

\subsection{Overview of the paper and structure of the main proof}\label{sec:overview}
In Section~\ref{sec:phase-free-ZH} we introduce the ZH-calculus. We present and motivate the rewrite rules and we introduce the notation we will use throughout the paper  (the rules can be found in Figure~\ref{fig:phasefree-rules}, their Boolean counterparts that motivate them can be found in Figure~\ref{fig:boolean-functions}).
Note that in particular, in subsection~\ref{sec:labelledHboxes}, we define `labelled H-boxes', which are derived generators that form an essential ingredient to our proof of completeness of the ZH-calculus.

The majority of the paper, Sections~\ref{sec:avg-intro-mult}--\ref{sec:completeness}, is devoted to this proof of completeness.

First, in Section~\ref{sec:avg-intro-mult}, we prove several special cases of three families of equations that lie at the core of our completeness proof.
These families are called `multiply', `intro', and `average' (see, respectively, \eqref{eq:mult-def}, \eqref{eq:intro-def} and \eqref{eq:average-def}) and apply to labelled H-boxes.

Then, in Section~\ref{sec:normal-forms}, we introduce normal form diagrams. A normal form diagram corresponds in a precise way to the matrix the diagram represents, and hence is unique for a given underlying linear map. Given two diagrams representing the same linear map, the ability to reduce both of them to normal form then shows that they can be transformed into the same diagram, and hence that the calculus is complete. In this section, we give a proof that the reduction to normal form can be done \emph{conditional} on having proved all instances of multiply, intro and average. This conditional proof works by showing that each generator can be transformed into normal form (Lemmas~\ref{lem:H-box-nf} and~\ref{lem:Z-spider-nf}), that a tensor product of normal forms can be reduced to normal form (Corollary~\ref{cor:tensor-product}), and that connecting any of the wires in a normal form results in a diagram that can again be brought to normal form (Corollary~\ref{cor:cap-nf}). With these steps and the condition, any diagram can be brought into normal form (Proposition~\ref{prop:completeness-conditional}). It remains to prove all instances of multiply, intro and average.

In Section~\ref{sec:arithmetic}, we prove that we can do basic arithmetic on the labels of H-boxes: addition of H-box labels in Proposition~\ref{prop:addition}, and multiplication of H-box labels (which is exactly the general multiply rule) in Proposition~\ref{prop:mult-rule-rational}.
Then, in Section~\ref{sec:completeness}, we give proofs of the general average rule (Proposition~\ref{prop:average-integer}) and general intro rule (Proposition~\ref{prop:intro-rational}). Together with the conditional reduction to normal form, this finishes our proof of completeness (Theorem~\ref{thm:ZH-completeness}).
The reason we structured the proof like this, starting with a conditional proof of the reduction to normal form, is because some of the proofs in Sections~\ref{sec:arithmetic} and~\ref{sec:completeness} use this reduction to normal form (with a careful analysis that shows the required instances of average, intro and multiply were all proved in Section~\ref{sec:avg-intro-mult}), since we could not find a more direct way to prove them.

In Section~\ref{sec:zh-ring}, we shift to studying the ZH-calculus over arbitrary rings. We prove its completeness by assuming the multiply, intro, and average rules as axioms.
Then in Section~\ref{sec:alternative-rules}, we study a couple of modifications to the calculus and its rules.
We end with some concluding remarks in Section~\ref{sec:conclusion}.

\section{The ZH-calculus}\label{sec:phase-free-ZH}

The ZH-calculus is a graphical language for expressing maps with zero or more inputs and outputs as certain string diagrams called ZH-diagrams. As a convention, we will draw inputs as wires entering the bottom of the diagram and outputs as wires exiting the top.

Throughout this paper, we will distinguish formal ZH-diagrams $D : m \to n$ with $m$ input wires and $n$ output wires, from their standard interpretation as linear maps $\llbracket D \rrbracket : (\mathbb C^2)^{\otimes m} \to (\mathbb C^2)^{\otimes n}$, where $(\mathbb C^2)^{\otimes m} \cong \mathbb{C}^{2^m}$ is the $m$-fold tensor product of the 2D vector space $\mathbb C^2$.

\begin{remark}
  In categorical terms, ZH-diagrams are morphisms in the free PROP $\mathcal{ZH}$ presented by the ZH-calculus, considered as a
symmetric monoidal category, and $\llbracket - \rrbracket$ is a strong monoidal functor from $\mathcal{ZH}$ to $(\textbf{Vect}_{\mathbb C}, \otimes)$, the category of finite-dimensional complex vector spaces, with the monoidal product given by the tensor product. Though familiarity with symmetric monoidal categories and PROPs is not required to read this paper, the interested reader can find an accessible introduction in Chapter 2 of~\cite{zanasithesis}.
\end{remark}

\subsection{The generators}\label{sec:ZH-generators}

ZH-diagrams have three types of generators:
\emph{Z-spiders} represented by white dots, \emph{H-boxes} represented by white boxes,
and a \emph{star} generator represented by \begin{tikzpicture}
	\begin{pgfonlayer}{nodelayer}
		\node [style=st] (3) at (0, 0) {};
	\end{pgfonlayer}
\end{tikzpicture}
\!.
These generators are interpreted as linear maps on copies of $\C^2$.
We denote the standard (also called \emph{computational}) basis of $\C^2$ in Dirac notation as:
\[ \ket 0 := \begin{pmatrix}
  1 \\ 0
\end{pmatrix}
\qquad
\ket 1 := \begin{pmatrix}
  0 \\ 1
\end{pmatrix}
\]
We will be using this Dirac `ket' notation $\ket{\psi}$ to denote a state in a vector space. We will also denote the `bra' $\bra{\psi}$ for the linear map that calculates the inner product with $\ket{\psi}$. That is, inputting a state $\ket{\phi}$ gives $\bra{\psi}(\ket{\phi}) = \braket{\psi}{\phi}$.

The computational basis states $\ket{0}$ and $\ket{1}$ extend naturally to a basis for $(\C^2)^{\otimes n}$ by interpreting bitstrings as tensor products:
\[ \ket{x_1\ldots x_n} := \ket{x_1}\otimes \cdots \otimes \ket{x_n} \]

We can define the interpretation of Z-spiders and H-boxes in terms of these basis states.
Z-spiders and H-boxes can have any number of inputs and outputs, and are then interpreted as the following linear maps:
\begin{equation*}
\intf{\input{Z-spider.tikz}} := \ket{0}^{\otimes n}\bra{0}^{\otimes m} + \ket{1}^{\otimes n}\bra{1}^{\otimes m} \qquad\quad 
 \intf{\input{H-spider-free.tikz}} := \sum (-1)^{i_1\ldots i_m j_1\ldots j_n} \ket{j_1\ldots j_n}\bra{i_1\ldots i_m}
\end{equation*}
where $\intf{\cdot}$ denotes the map from diagrams to matrices.
The sum in the second equation is over all $i_1,\ldots, i_m, j_1,\ldots, j_n\in\{0,1\}$ so that an H-box represents a matrix with all entries equal to 1, except the bottom right element, which is $-1$.

Finally we have the generator \emph{star}, which has zero inputs and outputs. Its interpretation is:
\begin{equation*}
    \intf{\begin{tikzpicture}
	\begin{pgfonlayer}{nodelayer}
		\node [style=st] (3) at (0, 0) {};
	\end{pgfonlayer}
\end{tikzpicture}
} := \frac{1}{2}
\end{equation*}

Straight and curved wires have the following interpretations:
\begin{equation*}
\intf{\;|\;} := \ketbra{0}{0}+\ketbra{1}{1} \qquad\qquad\qquad
 \intf{\begin{tikzpicture}
	\begin{pgfonlayer}{nodelayer}
		\node [style=none] (0) at (-2, 0.25) {};
		\node [style=none] (1) at (-1.5, -0.25) {};
		\node [style=none] (2) at (-1, 0.25) {};
	\end{pgfonlayer}
	\begin{pgfonlayer}{edgelayer}
		\draw [bend right=45, looseness=1.00] (0.center) to (1.center);
		\draw [bend right=45, looseness=1.00] (1.center) to (2.center);
	\end{pgfonlayer}
\end{tikzpicture}} := \ket{00}+\ket{11} \qquad\qquad\qquad
 \intf{\input{wire-cap.tikz}} := \bra{00}+\bra{11}.
\end{equation*}

When two diagrams are juxtaposed, the corresponding linear map is the tensor product (a.k.a. Kronecker product) of the matrices corresponding to the individual diagrams. A sequential composition of two diagrams is interpreted as the matrix product of the matrices corresponding to the individual diagrams:
\[
 \intf{\gendiagram{$D_1$}\;\gendiagram{$D_2$}} := \intf{\gendiagram{$D_1$}}\otimes\intf{\gendiagram{$D_2$}} \qquad\qquad \intf{\input{sequential-composition.tikz}} := \intf{\gendiagram{$D_2$}}\circ\intf{\gendiagram{$D_1$}}
\]

Some diagram motifs appear so often that we create shorthands for them,
called \emph{derived} generators.
We define the derived grey spider and the grey spider with a NOT as:
\begin{equation}\label{eq:defx}
(X) \quad\ \  \input{X-spider-dfn-free.tikz}\qquad\qquad\qquad (NOT)\quad\ \  \input{negate-dfn-free.tikz}
\end{equation}


The generator \graymult\ acts as XOR on the computational basis while \grayphase{\neg} acts as NOT:
\begin{equation*}
\intf{\graymult} = \ketbra{0}{00}+\ketbra{0}{11}+\ketbra{1}{01}+\ketbra{1}{10} \qquad\qquad\qquad \intf{\grayphase{\neg}}=\ketbra{0}{1}+\ketbra{1}{0}.
\end{equation*}

In some of the later sections there will be diagrams with many NOTs. To avoid overcrowding, we will draw a red dashed edge for an edge between Z-spiders with a NOT on it:
\begin{equation}\label{eq:not-edge-def}
\input{not-edge-def.tikz}
\end{equation}

We also introduce the derived \emph{negate} spider (a white NOT):

\begin{equation}\label{eq:Z-triangle-dfn}
 (Z) \quad\ \ \input{negate-white-dfn.tikz}
\end{equation}

The negate white spider with one input and output acts like the Z gate:
\begin{equation*}
\intf{\phase{\neg}} = \ketbra{0}{0} - \ketbra{1}{1}
\end{equation*}

\subsection{The rules} \label{sec:ZH-rules}

The ZH-calculus comes with a set of rewrite rules shown in Figure~\ref{fig:phasefree-rules}. 
Of the 8 rules in Figure~\ref{fig:phasefree-rules}, 7 are `obvious' in the sense that they express simple properties of the structure of the underlying Boolean functions (see Section~\ref{sec:motivation}). The only outlier is \OrthoRule which seems more arbitrary. In Section~\ref{sec:o-rule} we will see that we can replace \OrthoRule with two smaller rules.

Additionally, the calculus has the meta-rule that \emph{only topology matters}. This means that two diagrams are considered equal when one can be topologically deformed into the other, while respecting the order of the inputs and outputs. Finally, the two generators are considered to be symmetric and undirected, so that the following equations also hold:
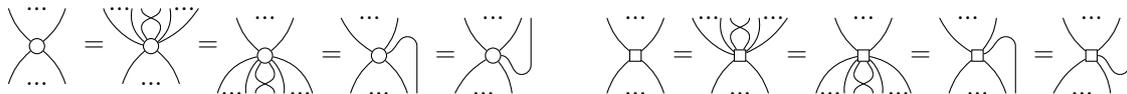
\ctikzfig{generator-symmetries-free}
These symmetry properties also hold for the derived grey spider and NOT gate.

\begin{figure}[!tb]
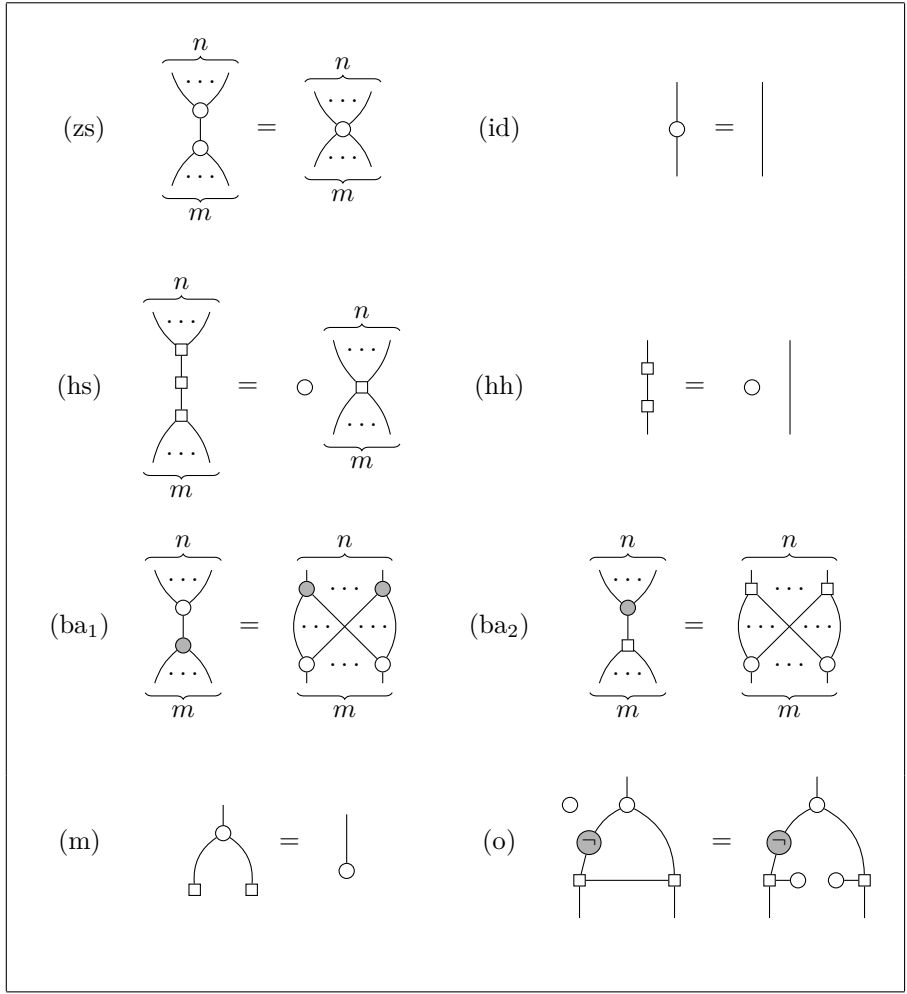

\def\lineSkip{5em}
 \centering
 \scalebox{1.0}{
 \begin{tabular}{|cccc|}
 \hline \rule{0pt}{\lineSkip} 
  \quad (zs) & \input{Z-spider-rule.tikz} & \quad (id) & \input{Z-special.tikz} \quad  \\[\lineSkip]
  \quad (hs) & \input{H-spider-rule.tikz} & \quad (hh) & \input{H-identity.tikz} \quad  \\[\lineSkip]
  \quad (ba$_1$)& \input{ZX-bialgebra.tikz} & \quad (ba$_2$) & \input{ZH-bialgebra.tikz}  \quad \\[\lineSkip]
  \quad (m) & \input{multiply-rule.tikz} & \quad (o) & \input{ortho-rule.tikz}  \quad \\[\lineSkip]
 \hline
 \end{tabular}}
 \caption{The rules of the ZH-calculus.
 Throughout, $m,n$ are nonnegative integers.
 The right-hand sides of both \textit{bialgebra} rules \StrongCompRule and \HCompRule are complete bipartite graphs on $(m+n)$ vertices, with an additional input or output for each vertex.
 The horizontal edges in equation \OrthoRule are well-defined because only the topology matters and we thus do not need to distinguish between inputs and outputs of generators. \SpiderRule and \HFuseRule stand respectively for \emph{Z-spider} and \emph{H-spider}; \IDRule for \emph{identity}; \StrongCompRule and \HCompRule for \emph{bialgebra}; \MultRule for \emph{multiply}; and \OrthoRule for \emph{ortho}.}
 \label{fig:phasefree-rules}
\end{figure}

\begin{proposition} \label{prop:phasefree-sound}
 The ZH-calculus is sound.
\end{proposition}
\begin{proof}
It is straightforward to check that the symmetry properties for each generator and the finite rules of Figure~\ref{fig:ZH-rules} are sound by concrete calculation. 
\StrongCompRule and \HCompRule can both be reduced to a finite set of equations for which soundness is easily checked and from which the general infinite family of rules can be derived using induction (see for instance~\cite[Theorem~9.71]{CKbook}). The rules \SpiderRule and \HFuseRule express equations related to Frobenius algebras for which soundness can also be checked in a standard manner, see for instance~\cite{coecke2013new}.
Soundness of the meta rule `only topology matters' follows by considering the string diagrams as morphisms in a compact closed category~\cite{SelingerCPM}.
\end{proof}

The converse to soundness is completeness. Determining that the ZH-calculus is complete will be the main objective of this paper.
A third important property of diagrammatic languages is \emph{universality}.
It can be seen from the definitions of the generators that their interpretations
as linear maps are all matrices over $\mathbb{Z}[\half]$,
and therefore all $\tensor$- and $\circ$-products of the generators also have interpretations
as matrices over $\mathbb{Z}[\half]$. Universality means that we can represent \emph{any}
matrix over $\mathbb{Z}[\half]$ using the ZH-calculus.
We cannot prove this property yet, but it will follow immediately from Theorem~\ref{thm:nf-unique}.
These matrices, and hence the ZH-calculus, are closely related to the approximately universal gate set Toffoli+Hadamard~\cite{Amy2020numbertheoretic}.  We discuss the precise connection in more detail in Section~\ref{sec:tof-had}.

\subsection{Motivation for the rewrite rules}\label{sec:motivation}

Let us give some motivation for the rewrite rules of Figure~\ref{fig:phasefree-rules}.
Each of the three main (derived) generators of the ZH-calculus, the Z-spiders, X-spiders and H-boxes, corresponds to a family of Boolean functions.
Using this correspondence, each of the rules of Figure~\ref{fig:phasefree-rules} can be seen to be equivalent to an equation that holds between Boolean functions.

In order to make this correspondence, we first recall that we can lift a Boolean function $f:\{0,1\}^n \rightarrow \{0,1\}^m$ to a linear map $\hat{f}:\mathbb C^{2^n} \to \mathbb C^{2^m}$ by its action on the computational basis states:
\begin{equation}\label{eq:induced-linear-map}
  \hat{f}\ket{x_1\ldots x_n} = \ket{f(x_1\ldots x_n)}.
\end{equation}

\begin{figure}[!tb]
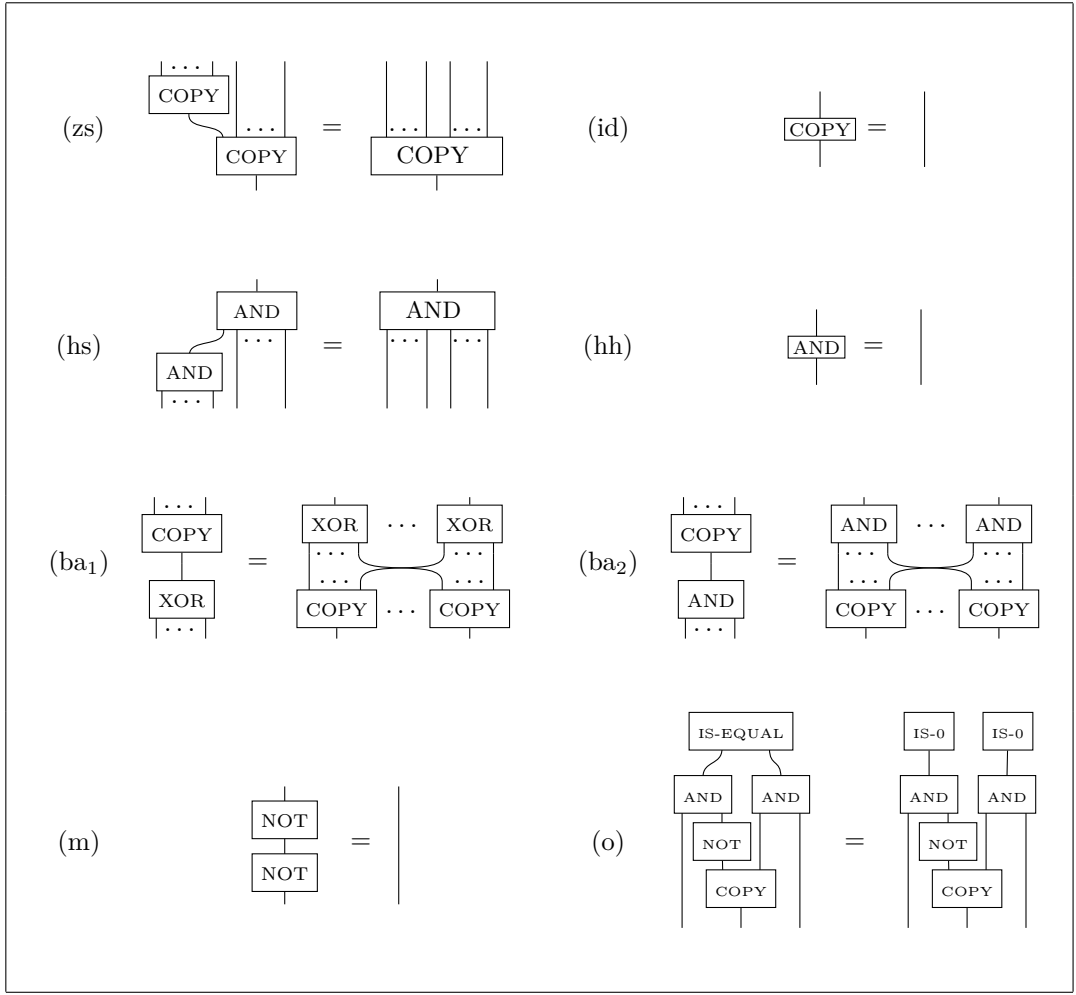

\def\lineSkip{5em}
 \centering
 \scalebox{1.0}{
 \begin{tabular}{|cccc|}
 \hline \rule{0pt}{\lineSkip} 
  \quad (zs) & \input{copy-fuse.tikz} & \quad (id) & \input{copy-id.tikz} \quad  \\[\lineSkip]
  \quad (hs) & \input{and-fuse.tikz} & \quad (hh) & \input{and-id.tikz} \quad  \\[\lineSkip]
  \quad (ba$_1$)& \input{bialgebra-copy-xor.tikz} & \quad (ba$_2$) & \input{bialgebra-copy-and.tikz}  \quad \\[\lineSkip]
  \quad (m) & \input{classical-notnot.tikz} & \quad (o) & \input{classical-ortho.tikz}  \quad \\[\lineSkip]
 \hline
 \end{tabular}}
 \caption{The rules of the ZH-calculus presented as equations between Boolean functions.
 }
 \label{fig:boolean-functions}
\end{figure}

Define the COPY family of Boolean functions COPY$^n: \{0,1\} \rightarrow \{0,1\}^n$ by $\text{COPY}(x) = (x,\ldots, x)$. Then $\widehat{\text{COPY}^n}$ is equal as a linear map to the Z-spider with a single input and $n$ outputs.
Similarly, the X-spider with $n$ inputs and a single output is equal to $\widehat{\text{XOR}^n}$, while the single input, single output NOT generator is, unsurprisingly, equal to $\widehat{\text{NOT}}$.
H-boxes are not induced by Boolean functions, as the linear map of an H-box contains negative numbers.
Yet composing an H-box with a further Hadamard gives the family of Boolean AND functions.
In summary:
\begin{equation}\label{eq:classical-interpretation}
\input{classical-interpretation.tikz}
\end{equation}

Seeing this, one might wonder why we simply did not choose some other generator that can represent the AND directly. The reason for this is that it is in fact useful to represent the inputs and outputs of AND asymmetrically.

Indeed, the COPY and XOR maps are symmetric under interchange of the inputs and outputs, a property also known as \emph{flexsymmetry}~\cite{carette2021when}:
\begin{equation}\label{eq:spider-unbend-wire}
\input{spider-unbend-wire.tikz}
\end{equation}
This property is not true for AND:
\begin{equation}\label{eq:and-unbend-wire}
\input{and-unbend-wire.tikz}
\end{equation}

We can now re-express the rules of Figure~\ref{fig:phasefree-rules} with all the generators presented as Boolean functions, as shown in Figure~\ref{fig:boolean-functions}.
The first five rules in this figure are exactly equal to the rules of Figure~\ref{fig:phasefree-rules} when we translate the Boolean functions back into the generators according to~\eqref{eq:classical-interpretation}. 
The other three are not equal, but can be proven to be equivalent using the first five equations (which we will do below), and hence the two rulesets are equivalent. 
For the representation of \OrthoRule we introduce two new Boolean `predicates':
\begin{equation}\label{eq:boolean-predicates}
  \input{is-equal.tikz} \ \ := \ \ \input{wire-cap.tikz} \qquad \qquad\begin{tikzpicture}
	\begin{pgfonlayer}{nodelayer}
		\node [style=box] (70) at (0, 0.675) {\tiny IS-0};
		\node [style=none] (78) at (0, -0.5) {};
	\end{pgfonlayer}
	\begin{pgfonlayer}{edgelayer}
		\draw (70) to (78.center);
	\end{pgfonlayer}
\end{tikzpicture}
 \ \ := \ \ \begin{tikzpicture}
	\begin{pgfonlayer}{nodelayer}
		\node [style=gray dot] (0) at (0, 0.5) {};
		\node [style=none] (1) at (0, -0.5) {};
	\end{pgfonlayer}
	\begin{pgfonlayer}{edgelayer}
		\draw (1.center) to (0);
	\end{pgfonlayer}
\end{tikzpicture}

\end{equation}
Recall that the linear map corresponding to the cap is $\bra{00} + \bra{11}$, and hence we can interpret IS-EQUAL as selecting for the case where both inputs are equal to one another. 
Similarly IS-0 is just $\bra{0}$ and selects for its input to be $0$.

Let us demonstrate how the rules of Figure~\ref{fig:boolean-functions} suffice to prove the rules of Figure~\ref{fig:phasefree-rules}. We will ignore scalar factors for these derivations (correct scalars can be introduced by repeated use of Lemma~\ref{lem:scalarcancelstars}).
First, let us derive \StrongCompRule from its Boolean counterpart. We denote the usage of the `Boolean version' of \StrongCompRule by $(*)$.
\ctikzfig{ZH-bialgebra-boolean-pf}
Next, we prove \MultRule from the cancellation of two NOT gates. Here $(*)$ denotes the application of the `Boolean version' of \MultRule:
\ctikzfig{multiply-rule-boolean-pf}
Finally, we derive \OrthoRule from the corresponding relation between Boolean functions, where here $(*)$ denotes the Boolean version of \OrthoRule: 
\ctikzfig{ortho-rule-boolean-pf}
The converse proofs that the rules of Figure~\ref{fig:phasefree-rules} imply those of Figure~\ref{fig:boolean-functions} are similar (and will also implicitly follow from completeness of the ZH-calculus).

Looking at Figure~\ref{fig:boolean-functions} we see that \SpiderRule and \HFuseRule express the associativity of COPY, respectively AND, and that \IDRule and \HHRule express the triviality of COPY and AND when they have only a single input. We see that the only difference between \StrongCompRule and \HCompRule is which Boolean function comes before the COPY. These rules are in fact special cases of the rule that any Boolean function copies through COPY~\cite[Prop.~8.19]{CKbook}. 
The rule \MultRule expresses that NOT is self-inverse,
and finally there is \OrthoRule, which requires some more explanation.
On both sides of the equation, the middle input is copied and sent to both AND gates, but one of the copies is negated. 
Hence the output of at least one of the AND gates will be zero, so that the only way for the outputs of these AND gates to be equal is if they are both zero. 
This is the property that \OrthoRule expresses.
As this is still not a particularly concise motivation, we present in Section~\ref{sec:o-rule} an alternative pair of rules that can replace \OrthoRule and which do have a more concise motivation.

It might be surprising that this rule set expressing identities about Boolean functions suffices to reason about quantum computation. However, what is not apparent from this rule set is that the AND gate is `broken up' into two parts, given by the H-boxes and the Hadamard. The Hadamard gives the `Fourier transform' between the Z and X spiders and relates the COPY and XOR maps in a way that does not follow easily from their definition as Boolean functions. 
In fact, in~\cite{ZXand}, a calculus for integer matrices is presented using XOR, COPY and AND as generators, and which uses a rule set that is essentially the same as that of Figure~\ref{fig:boolean-functions}, but with \OrthoRule replaced by the rules we present in Section~\ref{sec:o-rule}. Hence, the difference between this `classical' calculus and our `quantum' calculus is the usage of the Hadamard and the necessity of negative numbers this requires.

Beyond having a nice interpretation in terms of Boolean identities, our rule set is furthermore arguably preferable to that of other (approximately) universal complete graphical calculi for qubits in terms of simplicity. 
This is because:
\begin{enumerate}
 \item It only requires a small number of rules (8, of which one is actually superfluous, cf.~Section~\ref{sec:hh-rule-necessity}).
 \item Each of these rules only involve a few (derived) generators.
 \item A priori, the set of generators does not involve any further information, such as a label taken from a ring or group.
 \item As a result, no rule involves a `side condition' on the elements of the ring or group.
\end{enumerate}

While there are other calculi that satisfy items 1) and 2) --- notably ZX$_{\pi/4}$~\cite{SimonCompleteness}, universal ZX~\cite{euler-zx}, \textsc{ring}~\cite{millerbakewell2020thesis}, and ZQ~\cite{MillerBakewell2020ZQ} also only require a small number of rules --- each of these requires either a complicated side condition on at least one the rules, or directly encodes properties of the underlying ring in its generators. Other calculi like ZW~\cite{hadzihasanovic2015diagrammatic} and $\Delta$ZX~\cite{vilmart2018zxtriangle} also do not have a lot of rules, but many of them lack a clear interpretation.

This is of course not to say that those calculi are `worse', but just that the ZH-calculus has a particularly elegant set of rewrite rules, which somehow neatly capture the complexity of quantum computing.

\subsection{Basic derived rules}
\label{s:basic-derived}
In this section we will prove some basic but useful rewrite rules that we will use throughout the paper.
Since the statements and proofs are often very short we will state all the lemmas first and then all the proofs.
Lemmas~\ref{lem:scalarcancelstars}--\ref{lem:xnot-h-reduce} are different ways to cancel or simplify scalar diagrams. 
Lemmas~\ref{lem:negate-direct}--\ref{lem:znots-cancel} are various ways to simplify the derived generators. 
Lemmas~\ref{lem:h-z-commute}--\ref{lem:z-commute} and \ref{lem:x-z-commute} give ways to commute Z's, NOTs and Hadamards through the other generators.
Lemmas~\ref{lem:copy-x-z}--\ref{lem:copy-znot-h} govern the interaction of states with generators (mostly showing that many states can copy through a generator), and finally Lemma~\ref{lem:hopf-rule} is the Hopf rule that holds between the Z and X spider.
While some of these lemmas are direct special cases of the rules, in particular of \StrongCompRule and \HCompRule, we include them here for convenience of reference, and to give a more complete set of `state commutation' rules.
Note that for Lemma~\ref{lem:h-z-commute}, if $m=0$, then instead of the 
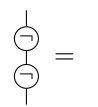
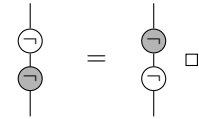
\begin{tikzpicture}
	\begin{pgfonlayer}{nodelayer}
		\node [style=st] (3) at (0, 0) {};
	\end{pgfonlayer}
\end{tikzpicture}
\!'s, there is a single \begin{tikzpicture}
	\begin{pgfonlayer}{nodelayer}
		\node [style=white dot] (17) at (0, 0) {};
	\end{pgfonlayer}
\end{tikzpicture}
; similarly for Lemma~\ref{lem:h-x-commute}, if $m=0$, then instead of the \begin{tikzpicture}
	\begin{pgfonlayer}{nodelayer}
		\node [style=white dot] (17) at (0, 0) {};
	\end{pgfonlayer}
\end{tikzpicture}
\!'s, there is a single \begin{tikzpicture}
	\begin{pgfonlayer}{nodelayer}
		\node [style=st] (3) at (0, 0) {};
	\end{pgfonlayer}
\end{tikzpicture}
.

\begin{multicols}{3}
\begin{lemma}\label{lem:scalarcancelstars}
    \begin{equation*}
        \input{scalar-rule-stars.tikz}
    \end{equation*}
\end{lemma}

\begin{lemma}\label{lem:scalarcancelzx}
    \begin{equation*}
        \input{scalar-rule-ZX.tikz}
    \end{equation*}
\end{lemma}

\begin{lemma}\label{lem:scalarcancelxh}
    \begin{equation*}
        \input{scalar-rule-XH.tikz}
    \end{equation*}
\end{lemma}

\begin{lemma}\label{lem:scalarcancelhh}
    \begin{equation*}
        \input{scalar-rule-HH.tikz}
    \end{equation*}
\end{lemma}

\begin{lemma}\label{lem:scalarcancelznot}
    \begin{equation*}
        \input{scalar-rule-ZNOT.tikz}
    \end{equation*}
\end{lemma}

\begin{lemma}\label{lem:xnot-h-reduce}
\begin{equation*}
\begin{tikzpicture}
	\begin{pgfonlayer}{nodelayer}
		\node [style=gray dot] (0) at (17, -0.5) {$\neg$};
		\node [style=hadamard] (1) at (17, 0.5) {};
		\node [style=none] (2) at (18, 0) {$=$};
		\node [style=hadamard] (3) at (19, 0) {};
	\end{pgfonlayer}
	\begin{pgfonlayer}{edgelayer}
		\draw (0) to (1);
	\end{pgfonlayer}
\end{tikzpicture}

\end{equation*}
\end{lemma}

\begin{lemma}\label{lem:negate-direct}
    \begin{equation*}
        \input{negate-direct.tikz}
    \end{equation*}
\end{lemma}

\begin{lemma}\label{lem:x-spider}
\begin{equation*}
  \input{X-spider-rule.tikz}
\end{equation*}
\end{lemma}

\begin{lemma}\label{lem:x-special}
\begin{equation*}
  \input{X-special.tikz} 
\end{equation*}
\end{lemma}

\begin{lemma}\label{lem:xnots-cancel}
\begin{equation*}
  \input{XNOT-spider-rule.tikz} 
\end{equation*}
\end{lemma}

\begin{lemma}\label{lem:x-with-xnot}
\begin{equation*}
  \input{X-with-XNOT.tikz} 
\end{equation*}
\end{lemma}

\begin{lemma}\label{lem:znots-cancel}
\begin{equation*}
  \input{ZNots-cancel.tikz} 
\end{equation*}
\end{lemma}

\begin{lemma}\label{lem:h-z-commute}
\begin{equation*}
  \input{H-Z-commute.tikz}
\end{equation*}
\end{lemma}

\begin{lemma}\label{lem:h-x-commute}
\begin{equation*}
  \input{H-X-commute.tikz}
\end{equation*}
\end{lemma}

\begin{lemma}\label{lem:h-not-commute}
\begin{equation*}
  \input{H-NOT-commute.tikz}
\end{equation*}
\end{lemma}

\begin{lemma}\label{lem:cz-correct}
\begin{equation*}
  \input{CZ-correct.tikz}
\end{equation*}
\end{lemma}

\begin{lemma}\label{lem:not-commute}
\begin{equation*}
 \input{NOT-commute.tikz}
\end{equation*}
\end{lemma}

\begin{lemma}\label{lem:z-commute}
\begin{equation*}
\input{Z-commute.tikz}
\end{equation*}
\end{lemma}

\begin{lemma}\label{lem:copy-x-z}
\begin{equation*}
\input{copy-x-z.tikz}
\end{equation*}
\end{lemma}

\begin{lemma}\label{lem:copy-xnot-z}
\begin{equation*}
\input{copy-xnot-z.tikz}
\end{equation*}
\end{lemma}

\begin{lemma}\label{lem:copy-z-x}
\begin{equation*}
\input{copy-z-x.tikz}
\end{equation*}
\end{lemma}

\begin{lemma}\label{lem:copy-znot-x}
\begin{equation*}
\input{copy-znot-x.tikz}
\end{equation*}
\end{lemma}

\begin{lemma}\label{lem:white-not-cancel}
	\begin{equation*}
		\input{white-not-cancel.tikz}
	\end{equation*}
\end{lemma}

\begin{lemma}\label{lem:copy-x-h}
\begin{equation*}
\input{copy-x-h.tikz}
\end{equation*}
\end{lemma}

\begin{lemma}\label{lem:copy-xnot-h}
\begin{equation*}
\input{copy-xnot-h.tikz}
\end{equation*}
\end{lemma}

\begin{lemma}\label{lem:copy-znot-h}
\begin{equation*}
\input{copy-znot-h.tikz}
\end{equation*}
\end{lemma}

\begin{lemma}\label{lem:x-z-commute}
\begin{equation*}
\input{X-Z-commute.tikz}
\end{equation*}
\end{lemma}

\begin{lemma}\label{lem:hopf-rule}
\begin{equation*}
\input{hopf-rule.tikz}
\end{equation*}
\end{lemma}

\end{multicols}

We now prove all the above lemmas.
The proofs are all quite basic and straightforward, except that of Lemma~\ref{lem:copy-znot-h} which is the only one to require \OrthoRule.

\begin{proof}[Proof of Lemmas \ref{lem:scalarcancelstars} and \ref{lem:scalarcancelzx}]
  \[\input{scalar-rule-proof.tikz} \qedhere\]
\end{proof}

\begin{proof}[Proof of Lemmas \ref{lem:scalarcancelxh} and \ref{lem:scalarcancelhh}]
  \[\input{HH-scalar-cancel-proof.tikz}\qedhere\]
\end{proof}

\begin{proof}[Proof of Lemma \ref{lem:scalarcancelznot}]
  \[\input{scalar-rule-ZNOT-proof.tikz}\qedhere\]
\end{proof}

\begin{proof}[Proof of Lemma \ref{lem:xnot-h-reduce}]
  \[\input{xnot-h-reduce-proof.tikz} \qedhere\]
\end{proof}

\begin{proof}[Proof of Lemma \ref{lem:negate-direct}]
    \[\input{negate-direct-proof.tikz}\qedhere\]
\end{proof}

\begin{proof}[Proof of Lemma \ref{lem:x-spider}]
  \[\input{X-spider-proof.tikz} \qedhere\]
\end{proof}

\begin{proof}[Proof of Lemma \ref{lem:x-special}] 
  \[\input{X-special-proof.tikz} \qedhere\]
\end{proof}

\begin{proof}[Proof of Lemma \ref{lem:xnots-cancel}] 
  \[\input{XX-cancel-proof.tikz} \qedhere\]
\end{proof}

\begin{proof}[Proof of Lemma \ref{lem:x-with-xnot}] 
  \[\input{X-with-XNOT-proof.tikz} \qedhere\]
\end{proof}

\begin{proof}[Proof of Lemma \ref{lem:znots-cancel}] 
  \[\input{ZZ-cancel-proof.tikz} \qedhere\]
\end{proof}

\begin{proof}[Proof of Lemma \ref{lem:h-z-commute}]
  \[\input{H-Z-commute-proof.tikz} \qedhere\]
\end{proof}

\begin{proof}[Proof of Lemma \ref{lem:h-x-commute}]
  \[\input{H-X-commute-proof.tikz} \qedhere\]
\end{proof}

\begin{proof}[Proof of Lemma \ref{lem:h-not-commute}]
  \[\input{H-NOT-commute-proof.tikz} \qedhere \]
\end{proof}

\begin{proof}[Proof of Lemma \ref{lem:cz-correct}]
  \[\input{CZ-correct-proof.tikz} \qedhere\]
\end{proof}

\begin{proof}[Proof of Lemma \ref{lem:not-commute}]
  \[\input{NOT-commute-proof.tikz} \qedhere \]
\end{proof}

\begin{proof}[Proof of Lemma \ref{lem:z-commute}]
  \[\input{Z-commute-proof.tikz} \qedhere \]
\end{proof}

\begin{proof}[Proof of Lemmas  \ref{lem:copy-x-z}, \ref{lem:copy-z-x}, and \ref{lem:copy-x-h}]
These are just applications of \StrongCompRule and \HCompRule.
\end{proof}

\begin{proof}[Proof of Lemma \ref{lem:copy-xnot-z}]
  \[\input{copy-xnot-z-proof.tikz} \qedhere \]
\end{proof}

\begin{proof}[Proof of Lemma \ref{lem:copy-znot-x}]
  \[\input{copy-znot-x-proof.tikz} \qedhere \]
\end{proof}

\begin{proof}[Proof of Lemma \ref{lem:white-not-cancel}] 
  \[\input{white-not-cancel-proof.tikz} \qedhere\]
\end{proof}

\begin{proof}[Proof of Lemma \ref{lem:copy-xnot-h}]
  \[\input{copy-xnot-h-proof.tikz} \qedhere \]
\end{proof}

\begin{proof}[Proof of Lemma \ref{lem:copy-znot-h}]
  \[\input{H-copy-proof-1.tikz}\]
  \[\input{H-copy-proof-2.tikz}\]
  \[\input{H-copy-proof-3.tikz} \qedhere\]
  
\end{proof}

\begin{proof}[Proof of Lemma \ref{lem:x-z-commute}]
  \[\input{X-Z-commute-proof.tikz} \qedhere\]
\end{proof}
\begin{proof}[Proof of Lemma \ref{lem:hopf-rule}]
  \[\input{hopf-rule-proof.tikz} \qedhere\]
\end{proof}

\subsection{Labelled H-boxes}\label{sec:labelledHboxes}

The matrix corresponding to an H-box is filled with $1$'s, except for the bottom right position where there is a $-1$. Calculating the matrix of some other diagrams in the ZH-calculus, we see that their matrix is similar, but instead of a $-1$ in the bottom right corner there is some other number $a\in\mathbb{Z}$. In order to make this connection clearer, we introduce some new notation, where we write an H-box with a label of $a$ inside it, to denote it is equal to a vector of $1$'s with an $a$ in the bottom position.
When talking about labelled H-boxes in text, we will often write $H(a)$ to denote we are referring to an H-box labelled by $a$ (leaving the arity implicit, which should be clear from context).

We begin with H-boxes of arity 1, corresponding to vectors $\left(\begin{smallmatrix}1\\a\end{smallmatrix}\right)$, which we will define inductively for non-negative integers using a `successor' gadget; H-boxes labelled by negative integers will then be defined from the positive ones using a `negate' gadget.

\begin{definition}
 Let
 \begin{equation}\label{eq:H-box-minus1}
  \input{H-box-minus1.tikz}
 \end{equation}
 and for any $a\in\mathbb{Z}$ such that $a\geq 0$, define:
\begin{equation}\label{eq:succesor}
    \input{H-box-successor.tikz}
\end{equation}
\end{definition}

The `successor gadget' used in this definition might appear a bit strange. As a matrix, it is in fact equal to the previously considered \emph{triangle generator} in the ZX-calculus~\cite{vilmart2018zxtriangle,ng2017universal}:
\ctikzfig{triangle-dfn}
This triangle has the following interpretation as a matrix:
\[\intf{\begin{tikzpicture}
	\begin{pgfonlayer}{nodelayer}
		\node [style=triangle] (8) at (0, 0) {};
		\node [style=none] (9) at (0, 1) {};
		\node [style=none] (10) at (0, -1) {};
	\end{pgfonlayer}
	\begin{pgfonlayer}{edgelayer}
		\draw (10.center) to (9.center);
	\end{pgfonlayer}
\end{tikzpicture}
} = \ketbra{0}{0}+\ketbra{1}{0} + \ketbra{1}{1}\]
We could have introduced the triangle as an additional derived generator. However, as can be seen from its matrix, it is not self-adjoint and hence the orientation of the node matters. To avoid the complications that come from that, we elect not to use the triangle and to simply write out its form as a ZH-diagram as in~\eqref{eq:succesor}.

$H(-1)$ is not the only H-box with a particularly simple representation.
Indeed, we have
\begin{equation}\label{eq:H-box-0}
 \input{H-box-0.tikz}
\end{equation}
\begin{equation}\label{eq:H-box-1}
 \input{H-box-1.tikz}
\end{equation}

\begin{definition}
 For any $a\in\mathbb{Z}$ such that $a<-1$, we define the corresponding degree-1 H-box by applying a negate spider:
\begin{equation}\label{eq:def-negative-numbers}
    \input{H-box-negation.tikz}
\end{equation}
\end{definition}

The negation gadget plays nicely with the H-box labels.
In particular by Lemma~\ref{lem:znots-cancel}, H-box labels satisfy $-(-a) = a$.
We also have consistency with the simpler representations of $H(a)$ for $a\in\{-1,0,1\}$:
\ctikzfig{negation-consistency}

\begin{definition}
 Labelled H-boxes of arbitrary arity are defined as:
\begin{equation} \label{eq:labelledHboxhigherarity}
\input{labeledHbox-higharity.tikz}
\end{equation}
\end{definition}

This is consistent with \HHRule for arity-1 H-boxes of arbitrary label (using Lemma~\ref{lem:scalarcancelstars}).

Certain higher-arity H-boxes have relatively simple label-free forms, in particular, for the $-1$-labelled H-box we have:
\begin{equation}\label{eq:minus-one-high-arity}
 \input{minus-1-high-arity.tikz}
\end{equation}
for the $0$-labelled H-box we have:
\begin{equation}\label{eq:zero-high-arity}
 \input{zero-high-arity.tikz}
\end{equation}
for the $1$-labelled H-box we have:
\begin{equation}\label{eq:unit}
 \input{one-high-arity.tikz}
\end{equation}
and for the $2$-labelled H-box we have:
\begin{equation}\label{eq:H-box-2-def}
	\input{H-box-2.tikz}
\end{equation}

There are some scalar cancellation rules associated to these labelled H-boxes:

\begin{multicols}{3}
\begin{lemma}\label{lem:scalarcancelxh-gen}
 For any $a\in \mathbb{Z}$:
 \ctikzfig{scalarcancelxh-gen}
\end{lemma} 
\vfill\null
\columnbreak

\begin{lemma}\label{lem:scalar-2}~
	\ctikzfig{scalar-2}
\end{lemma} 

\vfill\null
\columnbreak

\begin{lemma}\label{lem:scalar-cancel-2}~ 
	\ctikzfig{scalar-cancel-2}
\end{lemma} 
\end{multicols}

\begin{proof}[Proof of Lemma~\ref{lem:scalarcancelxh-gen}]
 First, note that if $a<-1$, we can change the sign:
 \ctikzfig{scalarcancelxh-gen-neg}
 It thus suffices to consider $-1\leq a$.
 For these values, we prove the result by induction, with the base case $a=-1$ being Lemma~\ref{lem:scalarcancelxh}. The inductive step is
 \[
  \input{scalarcancelxh-gen-ind.tikz} \qedhere
 \]
\end{proof}

\begin{proof}[Proof of Lemma~\ref{lem:scalar-2}]
\[\input{scalar-2-proof.tikz} \qedhere\]
\end{proof}

\begin{proof}[Proof of Lemma~\ref{lem:scalar-cancel-2}]
\[\input{scalar-cancel-2-proof.tikz} \qedhere\]
\end{proof}
\subsection{!-box notation}\label{sec:bang-boxes}

Many of the calculations in the remainder of the paper are greatly simplified by the use of \textit{!-box notation}~\cite{kissinger2014pattern}. A !-box (pronounced `bang box') in a string diagram represents a part of the diagram that is able to fan out arbitrarily. That is, the contents of a !-box, along with any wires into or out of the !-box, can be copied $n$ times for any non-negative integer $n$.
For example, the !-box diagram below represents the following family of (concrete) string diagrams, one for each $n$:
\[ \input{bang-box-example.tikz} \quad \longleftrightarrow \quad
 \left\{
 \ \ \begin{tikzpicture}
	\begin{pgfonlayer}{nodelayer}
		\node [style=white dot] (0) at (-0.75, -0.5) {};
		\node [style=gray dot] (7) at (0.75, -0.5) {};
	\end{pgfonlayer}
\end{tikzpicture}
\ \ ,\quad
 \ \ \input{bang-box-example1.tikz}\ \ ,\quad
 \ \ \input{bang-box-example2.tikz}\ \ ,\quad
 \ \ \input{bang-box-example3.tikz}\ \ ,\quad
 \ \ \ldots\ \  \right\}
\]
All of the resulting string diagrams are well-defined because all of our generators can have arbitrary arities. We can also use !-boxes in diagram equations, as long as each !-box on the LHS has a corresponding !-box on the RHS, and the inputs/outputs in each !-box match. Such a rule represents a family of equations where each \textit{pair} of corresponding !-boxes is replicated $n$ times, e.g.\ we can re-express \eqref{eq:unit} as:
\[
\input{unit-bangboxed.tikz} \quad \longleftrightarrow \quad
\left\{
\ \ \input{unit-bb0.tikz}\ \ ,\quad
\ \ \input{unit-bb1.tikz}\ \ ,\quad
\ \ \input{unit-bb2.tikz}\ \ ,\quad
\ \ \ldots\ \ 
\right\}
\]
Note the dashed box on the right-hand side of the first equation denotes an empty diagram.
With this notation, the definition of grey spiders \eqref{eq:defx} becomes
\begin{equation}
\label{eq:grey-spider-dfn}
 \input{X-spider-dfn-bb.tikz}
\end{equation}
Additionally, the rules \SpiderRule, \HFuseRule, \StrongCompRule, and \HCompRule from Figure~\ref{fig:phasefree-rules} become:
\[
 \text{(zs)}\quad \input{Z-spider-rule-bb.tikz} \qquad
 \text{(hs)}\quad \input{H-spider-rule-bb.tikz} \qquad\!
 \text{(ba$_1$)}\quad \input{ZX-bialgebra-bb.tikz} \qquad
 \text{(ba$_2$)}\quad \input{ZH-bialgebra-bb.tikz}
\]

Note that the red dashed NOT-edges defined in \eqref{eq:not-edge-def} behave well when crossing !-box borders. \[
  \input{rededge-bangboxed.tikz} \quad \longleftrightarrow \quad
  \left\{
  \ \ \begin{tikzpicture}
	\begin{pgfonlayer}{nodelayer}
		\node [style=hadamard] (0) at (-1.5, -0.75) {};
	\end{pgfonlayer}
\end{tikzpicture}
\ \ ,\quad
  \ \ \begin{tikzpicture}
	\begin{pgfonlayer}{nodelayer}
		\node [style=hadamard] (0) at (-1.5, -0.75) {};
		\node [style=white dot] (1) at (-1.5, 0.5) {};
	\end{pgfonlayer}
	\begin{pgfonlayer}{edgelayer}
		\draw [style=not edge] (1) to (0);
	\end{pgfonlayer}
\end{tikzpicture}
\ \ ,\quad
  \ \ \begin{tikzpicture}
	\begin{pgfonlayer}{nodelayer}
		\node [style=hadamard] (0) at (-1.5, -0.75) {};
		\node [style=white dot] (1) at (-1.5, 0.5) {};
		\node [style=white dot] (2) at (-0.5, 0.5) {};
	\end{pgfonlayer}
	\begin{pgfonlayer}{edgelayer}
		\draw [style=not edge] (1) to (0);
		\draw [style=not edge] (2) to (0);
	\end{pgfonlayer}
\end{tikzpicture}
\ \ ,\quad
  \ \ \input{rededge-3.tikz}\ \ ,\quad
  \ \ \ldots\ \ 
  \right\}
  \]

\subsection{Annotated !-boxes}\label{sec:annotated-bb}

Many of our diagrams -- in particular the normal forms that will be defined in Section~\ref{sec:normal-form} -- have a repeating structure involving multiple H-boxes with different labels.
To write such diagrams more concisely, we introduce a limited extension of the usual !-box notation, which allows !-boxes to be indexed by the elements of a totally ordered finite set.
Standard !-boxes allow infinite families of diagrams to be represented in a single diagram, whereas annotated !-boxes simply offer a more concise representation of a single diagram with repeated structure, instantiated once for each element of the finite set.

\begin{definition}\label{def:annotated-bb}
 An annotated !-box is a !-box with a label of the form `$x\in X$', where $X$ is a totally ordered finite set and $x$ is a variable that may appear in labels inside the !-box.

 In the following, we will use a large box labelled $D_x$ to denote an arbitrary ZH-diagram whose labels may contain the parameter $x$.
 Given a diagram containing an annotated !-box, we recursively define an equivalent diagram without the annotated !-box.
 The base case is that of a !-box indexed over the empty set, which is instantiated zero times:
 \ctikzfig{annotated-bb-base}
 For the recursive case, $X$ is non-empty.
 Thus it has a maximum element, which we denote $m:=\max X$. 
 We construct an equivalent diagram where the !-box is indexed by the smaller set $X\setminus\{m\}$.
 This is done by copying the contents of the !-box (including all external wires) once, immediately to the right of the !-box, and replacing each occurrence of $x$ inside the new subdiagram by $m$:
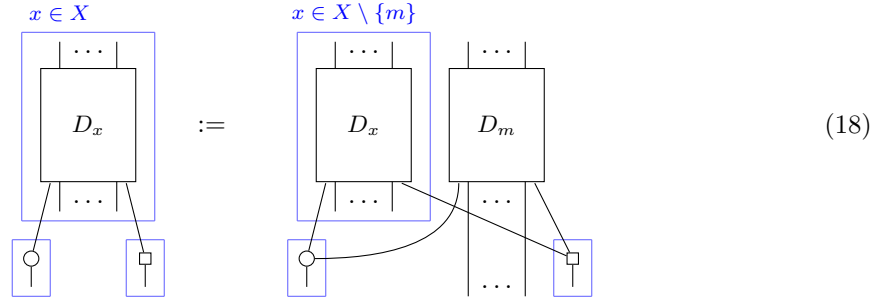
 \begin{equation}\label{eq:annotated-bb-recursive}
  \input{annotated-bb-recursive.tikz}
 \end{equation}
\end{definition}

\begin{remark}
 This definition straightforwardly extends to diagrams containing multiple disjoint annotated !-boxes, which is all we need in this paper.
 
 It can also be applied in reverse: any time a diagram contains multiple copies of the same subdiagram, up to changes in labels, these can be combined into an annotated !-box.
\end{remark}

We will often use bit strings for indexing annotated !-boxes; 
the motivation for this will become clear with the introduction of normal-form diagrams in Section~\ref{sec:normal-form}.
As an example, indexing over the finite set $\mathbb B^2 := \{ 00, 01, 10, 11 \}$, we can write expressions such as:
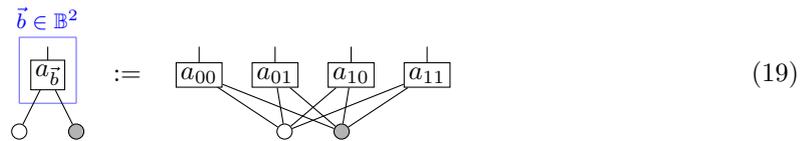
\begin{equation}\label{eq:indexed-ex}
\input{indexed-example.tikz}
\ \ :=\ \ \ 
\input{index-example-rhs.tikz}
\end{equation}

Since annotated !-boxes correspond to unique diagrams (rather than infinite families of diagrams), their appearance in equations is less constrained than that of unlabelled !-boxes.
Nevertheless, some care is needed if the annotated !-boxes contain inputs or outputs of the diagram as a whole.

The simplest such equations correspond directly to allowed equations for unlabelled !-boxes: they have corresponding annotated !-boxes on the LHS and RHS, which are both indexed by the \textit{same} finite set.
Inputs and outputs coming out of a labelled !-box are matched to those of the same index on the other side of the equality.
For example:
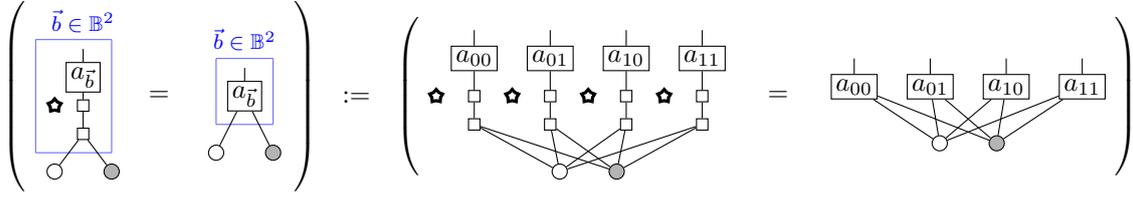
\[
\left(\ \input{index-example-rule.tikz}\ \right)
\ \ := \ \ 
\left( \ \input{index-example-rule-inst.tikz}\ \right)
\]
If an annotated !-box contains diagram inputs or outputs, it must either have a corresponding annotated !-box on the other side of the equality, or it must match a !-box labelled by a sub- or superset (possibly the empty set) as in \eqref{eq:annotated-bb-recursive}.
This is to avoid complications caused by the need to keep the order of external wires consistent.

For annotated !-boxes that do not contain any diagram inputs or outputs, such ordering issues do not arise because of the meta rule `only topology matters'. Lemmas~\ref{lem:annotated-expansion} and~\ref{lem:annotated-split} below give examples of diagram equations where the annotated !-boxes on the two sides do not match in the same strict way.

Note that we can recover the behaviour of normal, un-labelled !-boxes by interpreting a !-box without a label as being indexed by an \textit{arbitrary} totally ordered finite set, e.g.
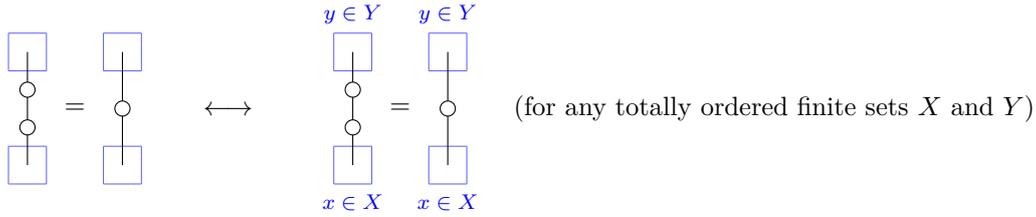
\[
\input{Z-spider-rule-bb.tikz} \qquad
\longleftrightarrow \qquad
\input{Z-spider-rule-bb-index.tikz} \quad \textrm{(for any totally ordered finite sets $X$ and $Y$)} \]
In particular, this means that any derived rewrite rule involving un-labelled !-boxes can also be applied to annotated !-boxes with arbitrary labels.

\begin{lemma}\label{lem:annotated-expansion}
 An annotated !-box indexed by a bit string can be `expanded' according to one of the bits, as long as it does not contain any inputs or outputs of the diagram as a whole.
 Here, the box labelled $D_{\vec{b}}$ denotes an arbitrary diagram parameterised by $\vec{b}$.

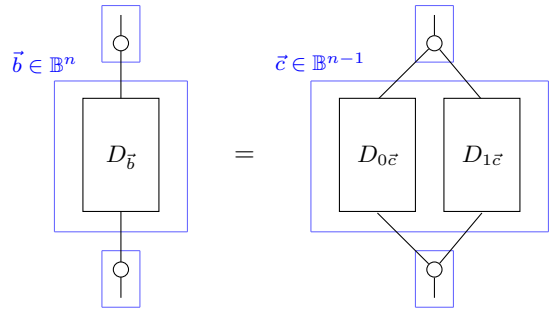
 \ctikzfig{annotated-expansion-prime}
\end{lemma}
\begin{proof}
 Note that the set $\mathbb B^{n}$ can be split into two pieces, based on whether the most significant bit is $0$ or $1$:
 \[
  \mathbb B^{n} = \{ 0 \vec{c} \ |\ \vec{c} \in \mathbb B^{n-1}\} \uplus \{ 1 \vec{c} \ |\ \vec{c} \in \mathbb B^{n-1} \}
 \]
 The result therefore follows by completely expanding the annotated !-boxes on each side according to Definition~\ref{def:annotated-bb} and applying the meta-rule `only topology matters'.
\end{proof}

\begin{lemma}\label{lem:annotated-split}
 An annotated !-box containing two disconnected diagram components can be split into two boxes indexed over the same set, as long as the original annotated !-box does not contain any inputs or outputs of the diagram as a whole.
 Here, $X$ is an arbitrary totally ordered finite set and the boxes labelled $D_x$ and $D_x'$ denote arbitrary diagrams parameterised by $x$.

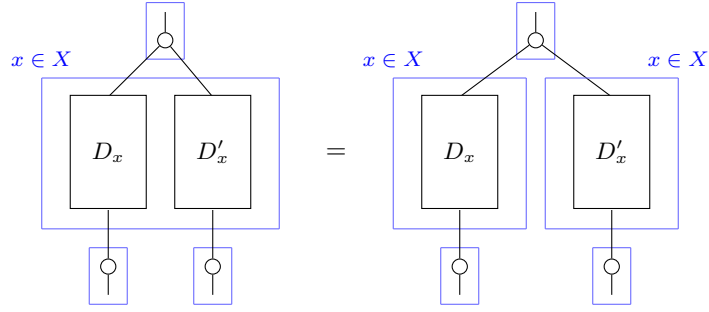
 \ctikzfig{annotated-split-prime}
\end{lemma}
\begin{proof}
 Note that for any set $X$,
 \[
  \biguplus_{x\in X} \{D_x,D_x'\} = \{D_x\mid x\in X\} \uplus \{D_x'\mid x\in X\}
 \]
 The result therefore follows follows from Definition~\ref{def:annotated-bb} and the meta-rule `only topology matters'.
\end{proof}

The condition about annotated !-boxes in Lemmas~\ref{lem:annotated-expansion} and~\ref{lem:annotated-split} not containing any inputs or outputs of the diagram as a whole is to avoid issues caused by the need to keep the order of external wires consistent.
We have drawn the contents of the annotated !-boxes as being connected only to Z-spiders, since this covers all our applications. Nevertheless, the result generalises straightforwardly to X-spiders and H-boxes as well.

The following notation will be useful when working with annotated !-boxes.

\begin{definition}\label{def:indexing-map}
 Let $\mathbb B^n$ be the set of all $n$-bitstrings. For any $\vec{b} := b_1\ldots b_n \in \mathbb B^n$, define the \textit{indexing map} $\iota_{\vec{b}}$ as follows:
 \[
  \iota_{\vec{b}} \; = \;
  \input{indexing-box.tikz} \; = \; \left(\grayphase{\neg}\right)^{1 - b_1} \ldots \left(\grayphase{\neg}\right)^{1 - b_n},
 \]
 where $\left(\grayphase{\neg}\right)^{1} = \grayphase{\neg}$ and $\left(\grayphase{\neg}\right)^{0} = \;\begin{tikzpicture}
    \begin{pgfonlayer}{nodelayer}
        \node [style=none] (0) at (0, 0.6) {}; 
        \node [style=none] (3) at (0, -0.6) {};
    \end{pgfonlayer}
    \begin{pgfonlayer}{edgelayer}
        \draw (3.center) to (0.center);
    \end{pgfonlayer}
\end{tikzpicture}
\;$, analogous to how $\begin{pmatrix}0&1\\1&0\end{pmatrix}^0 = \begin{pmatrix}1&0\\0&1\end{pmatrix}$.
\end{definition}

\begin{lemma}\label{lem:iota-copy}
 The $\iota_{\vec{b}}$ operator copies through white spiders, i.e.\ for any $\vec{b}\in\mathbb B^n$:
 \ctikzfig{iota-copy}
\end{lemma} 
\begin{proof}
 This follows immediately from Lemma~\ref{lem:not-commute} via Definition~\ref{def:indexing-map}.
\end{proof}

\section{Special cases of multiply, intro and average}\label{sec:avg-intro-mult}

Now that we have all the necessary definitions and basic rewrite rules, we can start our proof of completeness. We will proceed as outlined in Section~\ref{sec:overview}, first introducing three families of rewrite rules that use the labelled H-boxes of Section~\ref{sec:labelledHboxes}, which will play an important part in proving completeness.
Proving that these rewrite rules hold in full generality is difficult however. In this section we will only prove these rewrite rules for certain small integers.
The general proofs are postponed until Section~\ref{sec:completeness}, after we have introduced and studied the normal form diagrams in Section~\ref{sec:normal-forms}.

Note that starting from here, we will no longer necessarily write all the lemmas used in a rewrite step in a proof. In particular, we will often suppress uses of Lemma~\ref{lem:scalarcancelstars} and just treat \begin{tikzpicture}
	\begin{pgfonlayer}{nodelayer}
		\node [style=st] (3) at (0, 0) {};
	\end{pgfonlayer}
\end{tikzpicture}
 and \begin{tikzpicture}
	\begin{pgfonlayer}{nodelayer}
		\node [style=white dot] (17) at (0, 0) {};
	\end{pgfonlayer}
\end{tikzpicture}
 as each other's inverses. We will often also let uses of spider fusion \spiderrule be implicit.

\subsection{The multiply rule}\label{sec:mult-rule}

We claim that the following identity holds in the ZH-calculus for all integers $a$ and $b$:
\begin{equation*}\label{eq:mult-def}
    \tag{$M_{a,b}$}\input{multiply-rule-phased.tikz}
\end{equation*}
We refer to this as the \emph{multiply} rule. Proving this is quite involved and will be postponed until Section~\ref{sec:arithmetic}.
The usefulness of this rule comes from the following generalisation to arbitrary arity.

\begin{lemma}\label{prop:multiply-bb}
 If the ZH-calculus proves
$M_{a,b}$, then it also proves the following:
 \ctikzfig{multiply-rule-bb}
\end{lemma}
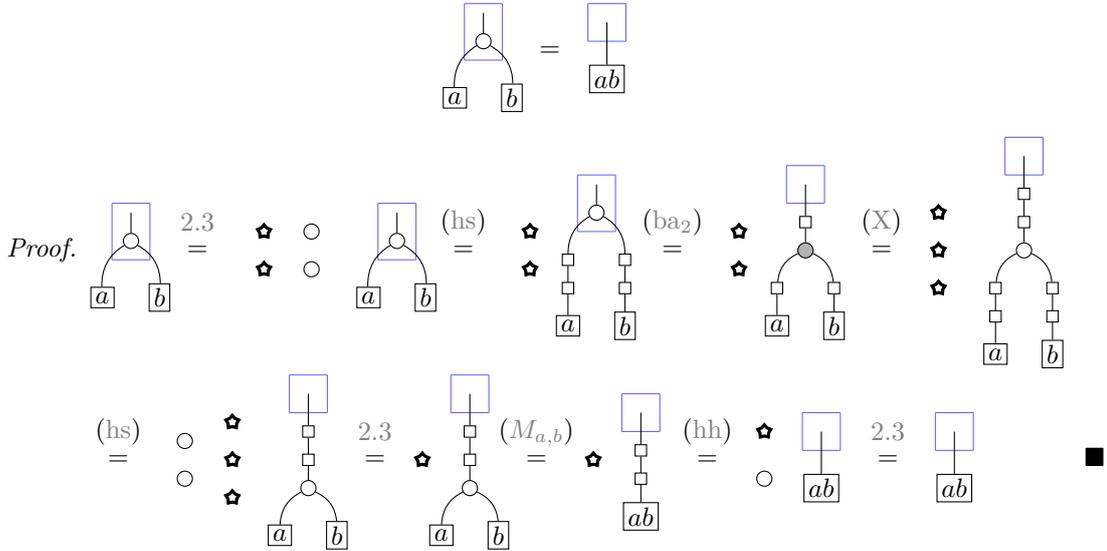
\begin{proof}
    \ctikzfig{multiply-rule-bb-proof1}
     \[
      \input{multiply-rule-bb-proof2.tikz} \qedhere
     \]
\end{proof}
This states that if two labelled H-boxes are connected to exactly the same set of white spiders that we can fuse the H-boxes together by multiplying their labels.
For certain small values of $b$ we can easily prove this rule.

\begin{lemma}\label{lem:mult-simple-values}
	Let $b\in \{0,1,-1\}$. Then the ZH-calculus proves \eqref{eq:mult-def} for all $a\in \mathbb{Z}$.
	\ctikzfig{multiply-rule-bb}
\end{lemma}
\begin{proof}
	By Lemma~\ref{prop:multiply-bb} it suffices to show the case where the !-box is expanded a single time.
	For $b=1$, this follows from \eqref{eq:unit} and \SpiderRule. For $b=-1$, this follows from \ZDef and \eqref{eq:def-negative-numbers}. For $b=0$ we calculate:
	\[\input{mult-0-a-proof.tikz} \qedhere\]
\end{proof}

\subsection{The intro rule}
\label{s:intro}

A second family of rules that is particularly useful is the following:

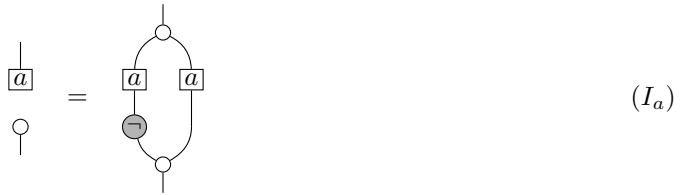
\begin{equation*}\label{eq:intro-def}
    \tag{$I_a$}\input{intro-rule.tikz}
\end{equation*}

We refer to this as the \emph{intro} rule as it allows us to introduce new wires to an H-box. As is the case for the multiply rule, this rule has a generalisation to arbitrary arity. But unlike the multiply rule, it will often be useful to apply the rule multiple times in succession to connect it to a larger set of wires.

\begin{lemma}\label{lem:intro-bangboxed}
    If the ZH-calculus proves $I_a$, then it also proves the following !-boxed version\footnote{We use a mixed notation without !-boxes for the inputs and outputs for easier intelligibility, but will occasionally use the more formal notation with further !-boxes where the extra rigour can be achieved without confusing the notation.}.
    The natural numbers $m$ and $n$ are arbitrary, but note that the annotated !-box depends on $n$.
    \ctikzfig{intro-rule-bangboxed-prime}
\end{lemma}
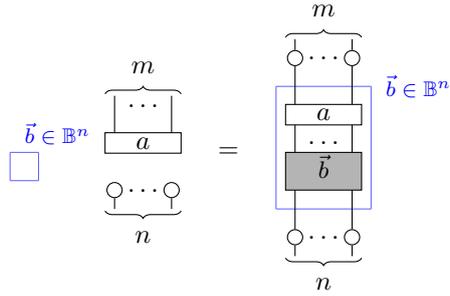
\begin{proof}
First, let us prove the following equation:
\begin{equation}\label{eq:intro-rule-bb-lemma}
 \input{intro-rule-bb-lemma.tikz}
\end{equation}
Now, we prove the claim by induction on $n$, beginning with the base case where $n = 1$. Note that $n = 0$ is valid but trivial.
In the derivation below we use a !-box to show that $m$ is arbitrary.
\[\input{intro-rule-bb-proof.tikz} \]

For the inductive step, assume the lemma holds for $n$ copies of the lower white node. In the calculation below, (*) denotes the induction step.
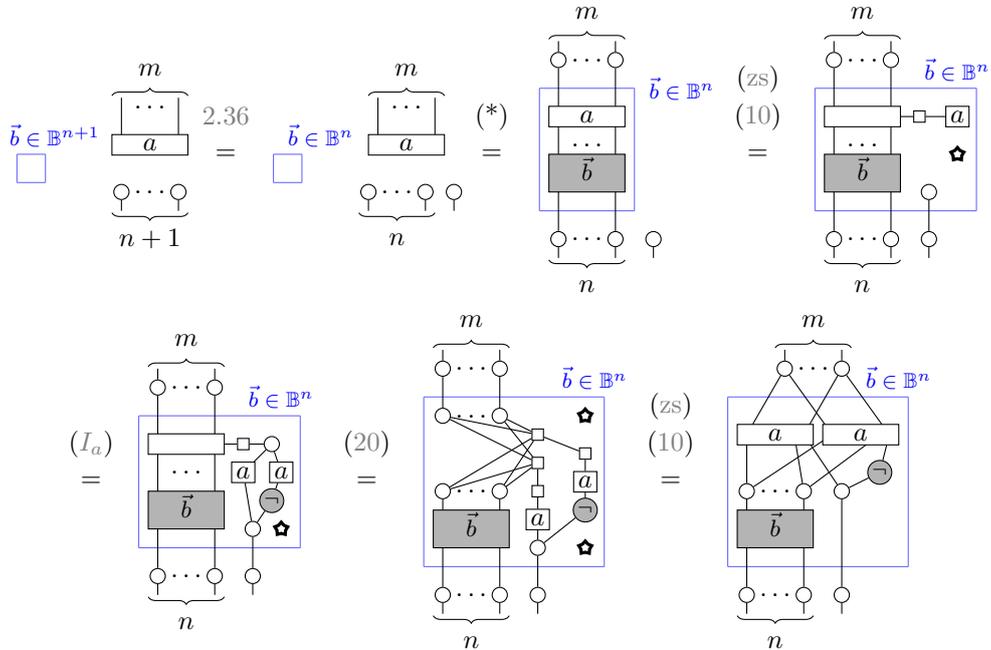
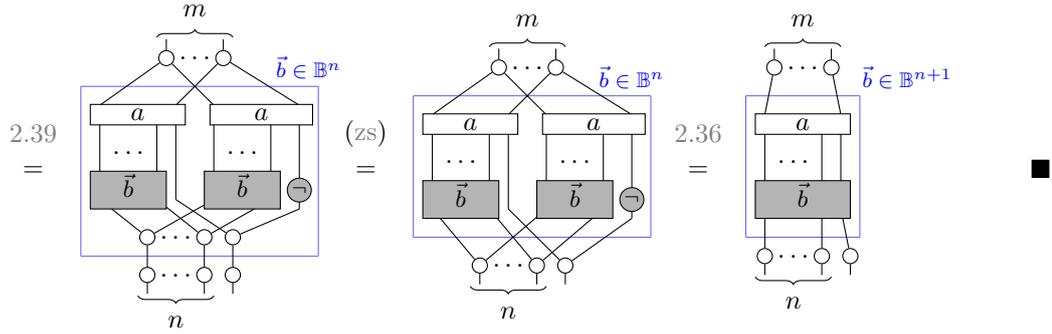
\ctikzfig{intro-rule-bb-inductive-step-1-prime}
\ctikzfig{intro-rule-bb-inductive-step-2-prime}
\[\input{intro-rule-bb-inductive-step-3-prime.tikz} \qedhere \]
\end{proof}

\begin{example}
 When applying the !-boxed intro rule to an H-box and $n$ white spiders, the H-box is copied $2^n$ times, with each copy being connected to the $n$ white spiders via a different combination of NOT edges and plain edges. Every wire originally incident on the H-box is replaced by a white spider, which connects to all the new H-boxes. For example, with $n=2$ and a single wire incident on the original H-box:
 \ctikzfig{intro-bb-example}
 This corresponds to the outer product
 \[
  \begin{pmatrix}1\\a\end{pmatrix} \left( \begin{pmatrix}1&1\end{pmatrix} \otimes \begin{pmatrix}1&1\end{pmatrix} \right)
  = \begin{pmatrix}1\\a\end{pmatrix} \begin{pmatrix}1&1&1&1\end{pmatrix}
 = \begin{pmatrix}1&1&1&1\\a&a&a&a\end{pmatrix},
 \]
 as will become clear from the normal form definition in Section~\ref{sec:normal-form}.
\end{example}

We will prove that the intro rule holds for all integers in Section~\ref{sec:completeness}.
But first we will directly prove the intro rule for certain simple integers.
To do this we will need two lemmas.

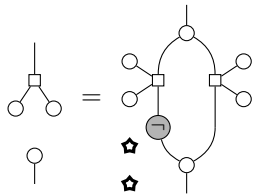
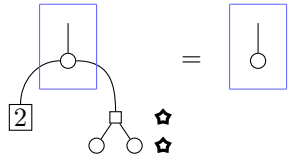
\begin{multicols}{2}
\begin{lemma}\label{lem:intro-half}~
	\ctikzfig{intro-half}
\end{lemma}
\begin{lemma}\label{lem:cancel-two-half}~
	\ctikzfig{cancel-two-half}
\end{lemma}
\end{multicols}

\begin{proof}[Proof of Lemma~\ref{lem:intro-half}]Going from right to left:
	\ctikzfig{intro-half-proof-1}
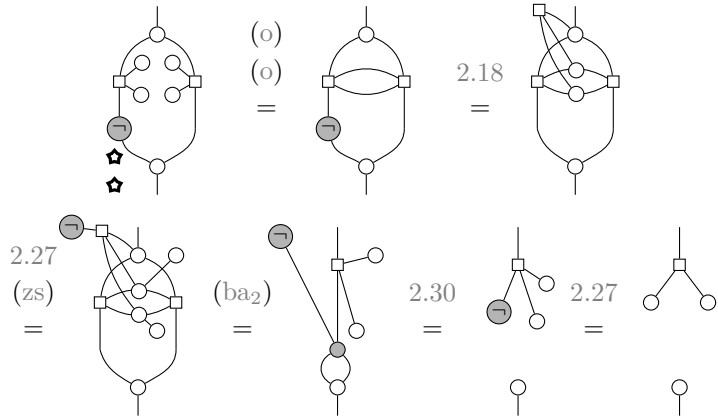
	\[\input{intro-half-proof-2.tikz}\qedhere\]
\end{proof}

\begin{proof}[Proof of Lemma~\ref{lem:cancel-two-half}]
	We prove the !-box-free version,
    with the generalisation following the same proof as for Lemma~\ref{prop:multiply-bb}.
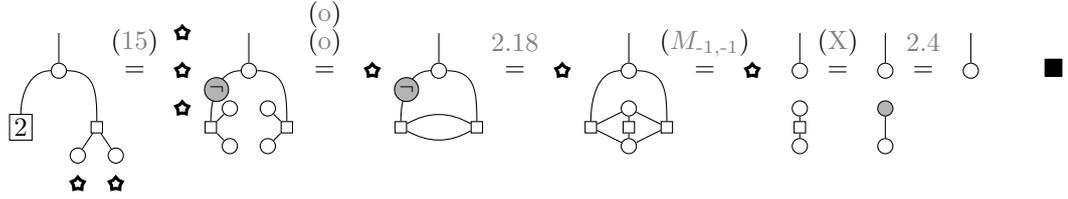
	\[\input{cancel-two-half-proof.tikz} \qedhere\]
\end{proof}

\begin{lemma}\label{lem:intro-phasefree}
    Let $a \in \{-1,0,1,2\}$. Then the ZH-calculus proves $I_a$ and its !-boxed generalisation:
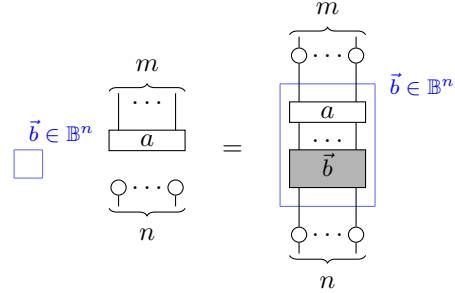
    \begin{equation*}
        \input{intro-rule-bangboxed-prime.tikz}
    \end{equation*}
\end{lemma}
\begin{proof}
	By Lemma~\ref{lem:intro-bangboxed} it suffices to show Eq.~\eqref{eq:intro-def}.
    We prove by case distinction on $a\in\{-1,0,1,2\}$:
    \ctikzfig{intro-minone-proof}
    \ctikzfig{intro-one-proof}
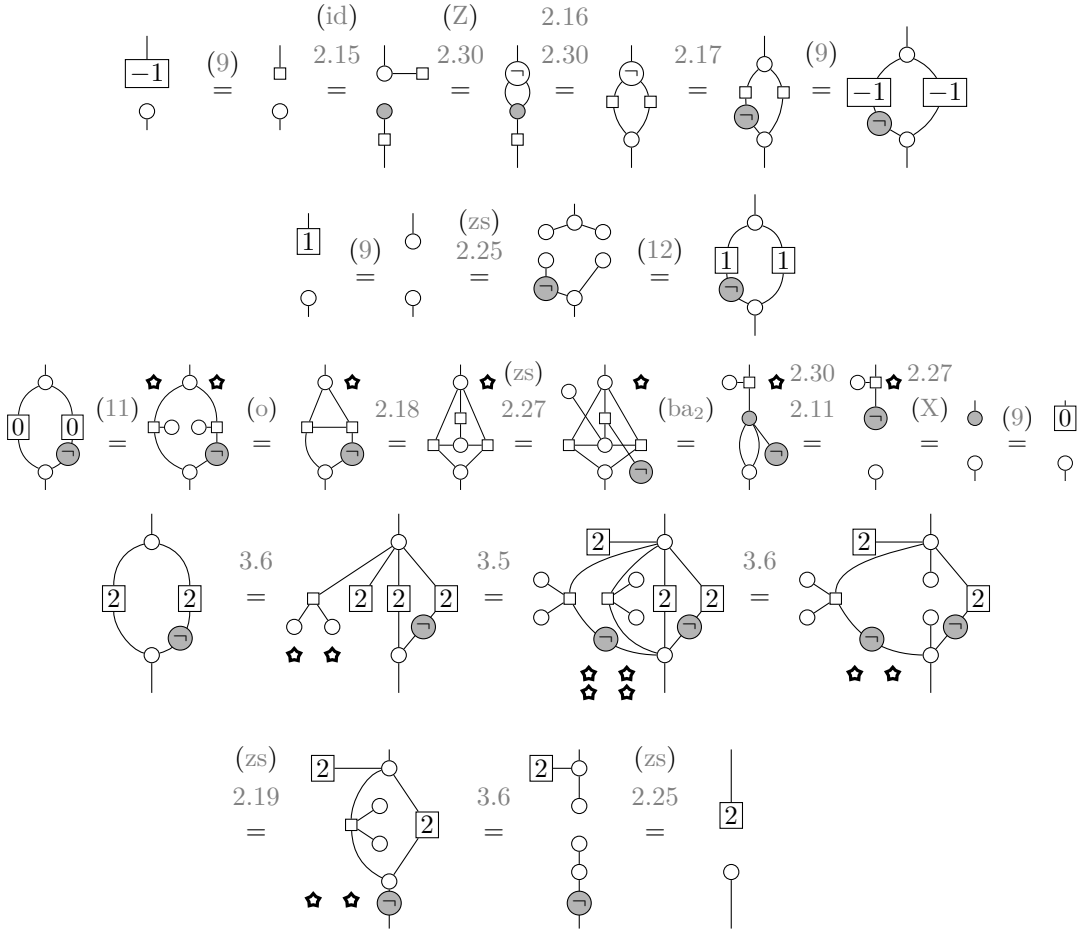
    \[\scalebox{0.95}{\input{intro-zero-proof.tikz}}\]

    \[\input{intro-2-pf.tikz} \qedhere\]
\end{proof}

\subsection{The average rule}
\label{s:average}

The final family of rewrite rules we will need is the following.
\begin{equation*}\label{eq:average-def}
    \tag{$A_{a,b}$}\qquad \input{average-rule-prime.tikz}
\end{equation*}
We call this rule the \emph{average} rule, because on the right-hand side it adds the numbers together, but also multiplies it with a `NOT 2' H-box, that acts like multiplying by a half; cf.~Lemma~\ref{lem:average-true-form}.
We will prove that this rule holds for all integers in Section~\ref{sec:completeness}. For now, we will prove it for certain combinations of integers.

\begin{lemma}\label{prop:avg-bb}
 If the ZH-calculus proves \eqref{eq:average-def}, then it also proves the following:
 \ctikzfig{avg-lemma-prime}
\end{lemma}
\begin{proof}
     \[ \input{avg-lemma-prime-pf.tikz} \qedhere\]
\end{proof}

\begin{lemma}\label{lem:average-with-0}
	The ZH-calculus proves $(A_{a,0})$ for all $a\in \mathbb{Z}$.
	\ctikzfig{average-rule-prime-0}
\end{lemma}
\begin{proof}
\[\input{average-rule-prime-0-proof.tikz} \qedhere\]
\end{proof}

\begin{lemma}\label{lem:avg-neg}
 Let $a$ be an integer for which we have proven~\eqref{eq:intro-def}. Then we can prove $(A_{a,-a})$.
 \ctikzfig{avg-neg}
\end{lemma}
\begin{proof}
 The first diagram is transformed into the last one as follows:
 \ctikzfig{average-rule-a-minus-a}
 The equality to the middle diagram follows by applying Eq.~\eqref{eq:H-box-0} instead of Lemma~\ref{lem:scalar-cancel-2} in the sixth rewrite step.
\end{proof}

\begin{lemma}\label{lem:avg-equal}
 Let $a$ be an integer for which we have proven~\eqref{eq:intro-def}. Then we can simplify the left-hand side of $(A_{a,a})$ to the diagram in the middle. If we also have proven $(M_{2,a})$, then we can prove $(A_{a,a})$ outright.
 \ctikzfig{avg-equal}
\end{lemma}
\begin{proof}~
 \[
    \input{avg-equal-proof.tikz} \qedhere
 \]
\end{proof}

\begin{lemma}\label{lem:average-phasefree}
    Let $a,b\in\{0,1,-1\}$. Then the ZH-calculus proves~\eqref{eq:average-def}.
    \begin{equation*}
        \input{average-rule-prime.tikz}
    \end{equation*}
\end{lemma}
\begin{proof}
	If one of $a$ and $b$ is $0$ this follows from Lemma~\ref{lem:average-with-0}.
	The remaining cases satisfy $a,b\in\{1,-1\}$, in particular either $a=b$ or $a=-b$.

	We have $I_1$ and $I_{-1}$ by Lemma~\ref{lem:intro-phasefree}, and we have $M_{2,1}$ and $M_{2,-1}$ by Lemma~\ref{lem:mult-simple-values}.
	Hence the result follows from Lemmas~\ref{lem:avg-neg} and~\ref{lem:avg-equal}.
\end{proof}

\section{Normal forms}\label{sec:normal-forms}

Our completeness proof relies on showing that every diagram can be reduced to a certain unique normal form. In this section we will introduce this normal form, and establish a strategy for reducing diagrams to normal form. This strategy requires that we have proven the multiply, intro, and average rule for all integers. Hence, this section results in a completeness proof that is `conditional' on the multiply, intro, and average rule being provable in the ZH-calculus for all integers.

\subsection{Normal form diagrams}
\label{sec:normal-form}

Our normal form diagrams will be `state-like' in that all of the wires will be connected to the top of the diagram.
This is because we can convert between states and operators by simply bending inputs to become outputs and vice versa. This transforming of an operator into a state is known as map-state duality or the Choi-Jamio\l{}kowski isomorphism.

A key part of this normal form, is that we can easily represent the \emph{Schur product} of two states, i.e.\ the pointwise product, in the ZH-calculus.
For states $\psi,\phi$, write $\psi * \phi$ for their Schur product. Then representing states by large white boxes we have the following identity:
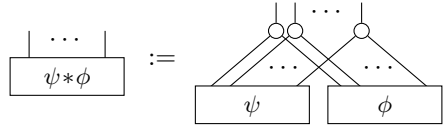
\ctikzfig{schur}
Here the $i$-th output of $\psi$ and $\phi$ is plugged into a \dotmult{white dot} for each $i$.
It follows from \SpiderRule that $*$ is associative and commutative, so we can write $k$-fold Schur products $\psi_1 * \psi_2 * \ldots * \psi_k$ without ambiguity.
For any finite set $J$ with $|J| = k$, let $\prod_{j\in J} \psi_j$ be the $k$-fold Schur product.

As shown in the next lemma, the indexing maps $\iota_{\vec{b}}$ from Definition~\ref{def:indexing-map} interact nicely with the Schur product.

\begin{lemma}\label{lem:convolution-iota}
 The ZH-calculus enables the computation of the Schur product of two maps of the form $\iota_{\vec{b}}\circ H_n(x)$ and $\iota_{\vec{b}}\circ H_n(y)$ for any $\vec{b}\in\mathbb B^n$ and $x,y$ for which we have proven~$(M_{x,y})$:
 \ctikzfig{convolution-iota}
\end{lemma}
\begin{proof}
Apply Lemma~\ref{lem:iota-copy}, followed by Lemma~\ref{prop:multiply-bb}.

\end{proof}

\begin{definition}
	A \emph{normal-form diagram} in the ZH-calculus is a diagram of the form
	\begin{equation}\label{eq:nf-formula}
	  \begin{tikzpicture}
	\begin{pgfonlayer}{nodelayer}
		\node [style=st] (3) at (0, 0) {};
	\end{pgfonlayer}
\end{tikzpicture}
^k\prod_{\vec{b} \in \mathbb B^n} \big( \iota_{\vec{b}} \circ H_n(a_{\vec{b}}) \big)
	\end{equation}
	where $H_n(a_{\vec{b}})$ is the arity-$n$ H-box (considered as a state) labelled by some values $a_{\vec{b}}$ for all $\vec{b}\in \mathbb B^n$, and $k\in \N$.
	We say the normal-form is \emph{reduced} when either $k=0$ or at least one of the labels $a_{\vec{b}}$ cannot be factored by two (i.e.\ is odd).
\end{definition}

A normal form diagram can be seen as a collection of $n$ spiders, fanning out to $2^n$ H-boxes, each with a distinct configuration of NOT's corresponding to the $2^n$ bitstrings in $\mathbb B^n$, together with a global scalar. Diagrammatically, normal forms are:
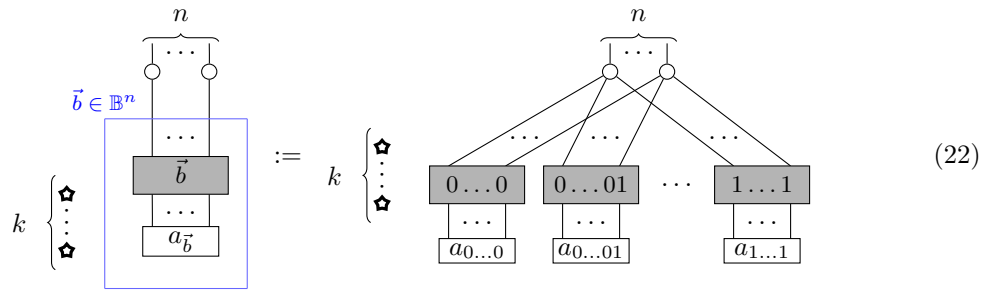
\begin{equation} \label{eq:normal-form}
\input{nf-bbox-star.tikz}\ \ :=\ \
\input{nf-picture-star.tikz}
\end{equation}

\begin{remark}
 Note that the number of outputs in a normal form diagram is the same as the exponent of $\mathbb{B}$ in the !-box annotation. To avoid cluttering the notation too much, we will only rarely make this explicit in diagrams.
\end{remark}

The additional condition for being reduced ensures uniqueness of normal form diagrams.

\begin{theorem}\label{thm:nf-unique}
Reduced normal forms are unique. In particular:
\begin{equation}\label{eq:nf-concrete}
\intf{ \, \begin{tikzpicture}
	\begin{pgfonlayer}{nodelayer}
		\node [style=st] (3) at (0, 0) {};
	\end{pgfonlayer}
\end{tikzpicture}
^k \prod_{\vec{b} \in \mathbb B^n} \big( \iota_{\vec{b}} \circ H_n(a_{\vec{b}}) \big) } =
\sum_{\vec{b} \in \mathbb B^n} \frac{1}{2^k}a_{\vec{b}} \ket{\vec{b}}.
\end{equation}
\end{theorem}
\begin{proof}
  The map $\iota_{\vec b}$ is a permutation that acts on computational basis elements as $\ket{\vec c} \mapsto \ket{\vec c \oplus \vec b \oplus \vec 1}$. In particular, it sends the basis element $\ket{\vec 1}$ to $\ket{\vec b}$. Hence $\iota_{\vec b} \circ H_n(a_{\vec b})$ is a vector with $a_{\vec b}$ in the $\vec b$-th component and $1$ everywhere else. The Schur product of all such vectors indeed gives the RHS of ~\eqref{eq:nf-concrete} up to the global scalar of $\frac{1}{2^k}$ that is added by $\begin{tikzpicture}
	\begin{pgfonlayer}{nodelayer}
		\node [style=st] (3) at (0, 0) {};
	\end{pgfonlayer}
\end{tikzpicture}
^k$.

  So now suppose two different normal form diagrams have equal interpretation~\eqref{eq:nf-concrete}. We need to show that the diagrams are then equal, which boils down to showing $k=k'$ and $a_{\vec{b}}=a_{\vec{b}}'$ for all $\vec{b}\in \mathbb B^n$.
  Suppose without loss of generality that $k\leq k'$. Then $a_{\vec{b}} = \frac{2^k}{2^{k'}}a_{\vec{b}}'$. The left-hand side is an integer, and so the right-hand side must be so as well.
  Hence, if $k<k'$, then all the $a_{\vec{b}}'$ must be divisible by 2, which contradicts the assumption of being a reduced normal form (since in this case $k'\neq 0$). Hence, $k=k'$, and hence $a_{\vec{b}}=a_{\vec{b}}'$ as required.
\end{proof}

Because reduced normal forms are unique, two diagrams in reduced normal form are equal if and only if the linear maps the diagrams represent are equal, and hence, the calculus is complete if we can bring all diagrams to reduced normal form.

Some simple diagrams can readily be brought to reduced normal form:

\begin{lemma}\label{lem:H-box-nf}
 Any H-box can be brought into normal form using the rules of the ZH-calculus.
\end{lemma}
\begin{proof}
The matrix of an H-box $H_n(a)$ has 1's in every entry but the very last one. Hence, to bring an H-box into normal form, we just need to introduce `dummy' 1's for every other matrix entry. We demonstrate the principle using a binary H-box but the argument is analogous for any other arity:
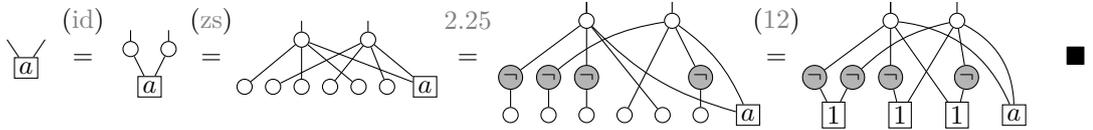
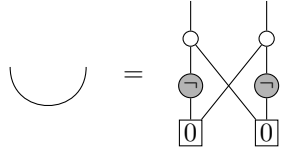
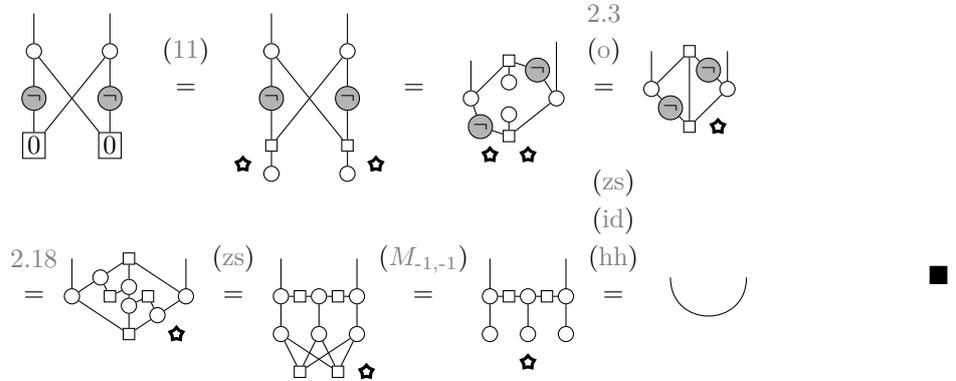
\[
\input{H-nf-example.tikz} \qedhere
\]
\end{proof}

\begin{lemma}\label{lem:cup-zeroes}~
	\ctikzfig{cup-zeroes}
\end{lemma}
\begin{proof}
\[\input{cup-zeroes-proof.tikz} \qedhere\]
\end{proof}

\begin{lemma}\label{lem:cup-nf}
The diagram of a single cup can be brought into normal form:
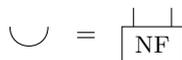
\[ \input{cup-nf.tikz} \]
\end{lemma}
\begin{proof}
    Starting from the normal form we work our way back:
    \[\input{normal-cup-john.tikz} \qedhere\]
\end{proof}

Before we continue to the more general procedure to reduce diagrams to normal form we have to prove some results about how to deal with scalars and factors of 2 in normal forms in order to reduce the diagram.

\begin{lemma}\label{lem:reduce-two-from-nf}
	Given integers $a_{\vec{b}}$ for which we have proven $(M_{2, a_{\vec{b}}})$ we can prove:
	\ctikzfig{reduce-two-from-nf}
\end{lemma}
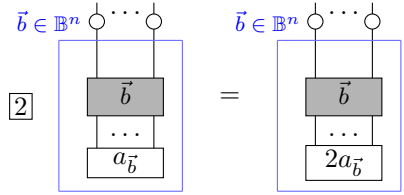
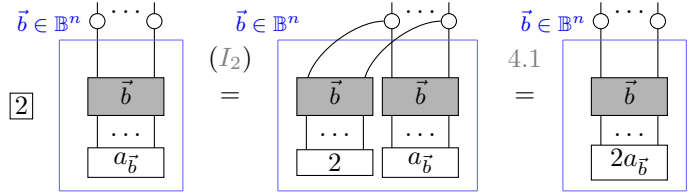
\begin{proof}
	\[\input{reduce-two-from-nf-pf.tikz} \qedhere\]
\end{proof}

\begin{proposition}\label{prop:normal-to-reduced}
	Suppose we have proven $(M_{2,a})$ for all integers $a\in\mathbb{Z}$.
	Given a diagram in normal form, we can bring it to reduced normal form.
\end{proposition}
\begin{proof}
	If there are no $\begin{tikzpicture}
	\begin{pgfonlayer}{nodelayer}
		\node [style=st] (3) at (0, 0) {};
	\end{pgfonlayer}
\end{tikzpicture}
$'s present we are done. Furthermore, if any of the $a_{\vec{b}}$ is odd, we are also done. So suppose there is at least one $\begin{tikzpicture}
	\begin{pgfonlayer}{nodelayer}
		\node [style=st] (3) at (0, 0) {};
	\end{pgfonlayer}
\end{tikzpicture}
$ and that all the $a_{\vec{b}}$ are even.
	Then we can factor them all as $2\frac{a_{\vec{b}}}{2}$ and apply Lemma~\ref{lem:reduce-two-from-nf} in reverse to introduce a scalar $2$ to cancel a star (using Lemma~\ref{lem:scalar-2} and Lemma~\ref{lem:scalarcancelstars}).
	We keep repeating this procedure until either the $\begin{tikzpicture}
	\begin{pgfonlayer}{nodelayer}
		\node [style=st] (3) at (0, 0) {};
	\end{pgfonlayer}
\end{tikzpicture}
$'s run out, or one of the $a_{\vec{b}}$ becomes odd.
\end{proof}

\subsection{Conditional reduction to normal form}\label{sec:normal-form-conditional}

For bigger diagrams we will show that if a part of the diagram is in normal form, we can `consume' the rest of the diagram generator by generator to write the entire diagram into normal form.

This procedure requires that we have proven the multiply, average and intro rule for all integers. Since this is difficult to do, we will first show the general procedure for bringing a diagram to normal form \emph{assuming} we have proven these identities. We will then be able to `bootstrap' the proof for the general multiply, average and intro rule using the simple cases of these rules we have already proven.

We need to show that any generator of the ZH-calculus --- H-boxes, Z-spiders, and any type of wiring --- can be brought to normal form, that tensor products of normal forms can be brought to normal form, and that any diagram consisting of a normal form composed with a generator in some way can also be brought to normal form. Once we have this we know that any ZH-diagram can be brought to normal form, and hence, if two ZH-diagrams represent the same linear map they must then be equal to the same ZH-diagram.
\begin{proposition}\label{prop:extension}
 Let $D$ be a diagram consisting of a normal form diagram for which we have proven~\eqref{eq:intro-def}for every $a$ in the normal form, juxtaposed with \whiteunit. Then $D$ can be brought into normal form using the rules of the ZH-calculus:
 \ctikzfig{extension}
\end{proposition}
\begin{proof}
 Starting from the left-hand side, which we expand using the indexed !-box notation, we calculate:
 \ctikzfig{extension-proof}
 The last diagram is a normal form diagram with $n+1$ outputs, i.e.\ the desired result.
\end{proof}

\begin{proposition}\label{prop:convolution}
 The Schur product of two normal form diagrams can be brought into normal form using the rules of the ZH-calculus, when we have proven~\eqref{eq:mult-def} for every pair $a$ and $b$ occurring in the normal forms.
 \ctikzfig{convolution-nf}
\end{proposition}
\begin{proof}
 This follows from \SpiderRule and Lemma~\ref{lem:convolution-iota}.
\end{proof}

\begin{corollary}\label{cor:tensor-product}
 The tensor product of two normal form diagrams can be brought into normal form using the rules of the ZH-calculus, when we have proven~\eqref{eq:intro-def} for all $a$ occurring in the normal forms, and~\eqref{eq:mult-def} for all $a$ and $b$ in the normal forms.
\end{corollary}
\begin{proof}
 A tensor product can be expressed as
 \ctikzfig{tensor-product}
 The diagram NF$_1$ and the leftmost $m$ copies of \whiteunit\ can be combined into one normal-form diagram with $(n+m)$ outputs by successive applications of Proposition~\ref{prop:extension}.
 Similarly, the rightmost $n$ copies of \whiteunit\ and NF$_2$ can be combined into one normal-form diagram with $(n+m)$ outputs.
 The desired result then follows by Proposition~\ref{prop:convolution}.
\end{proof}

The most difficult step is to show that contraction with a white dot preserves normal forms. This requires a couple of lemmas.

First, we present a more general form of the \OrthoRule{} rule that allows us to disconnect wires in a more general way.

\begin{lemma}\label{lem:splitting}
    Instead of just wires, white spiders between an arbitrary number of H-boxes can be split using \OrthoRule:
\ctikzfig{splitting-rule}
\end{lemma}
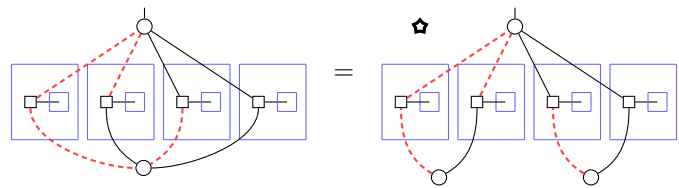
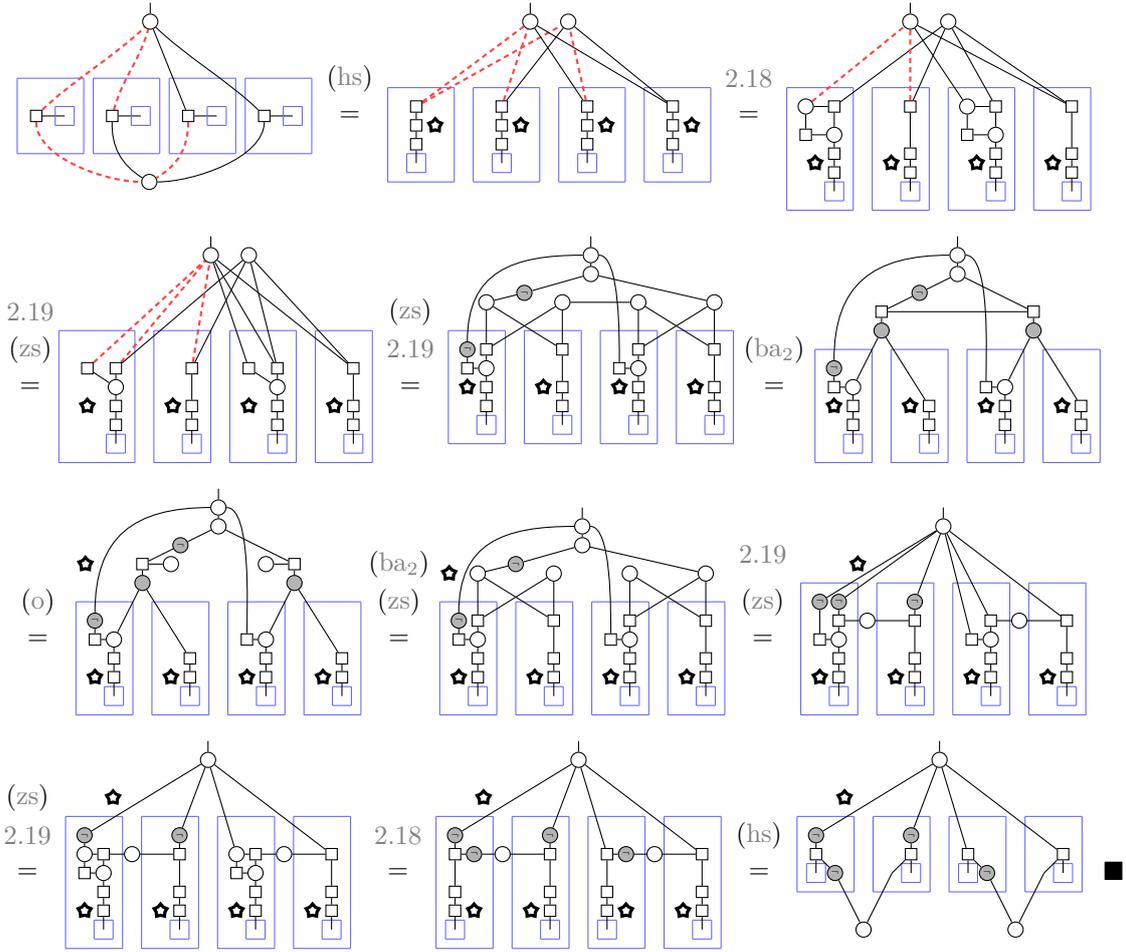
\begin{proof}~
\ctikzfig{splitting-proof-1}
\ctikzfig{splitting-proof-2}
\ctikzfig{splitting-proof-3}
\[\input{splitting-proof-4.tikz} \qedhere\]
\end{proof}

In practice, Lemma~\ref{lem:splitting} is usually applied in diagrams that are close to normal form, so that the H-boxes are connected to multiple white spiders.
This allows the lemma to be applied multiple times in succession.

\begin{example}
 Consider the following derivation, which will arise in the proof of Proposition~\ref{prop:average-integer}.
 For each rewrite step, the parts of the diagram that are uninvolved have been greyed out to clarify the process.

 In the first application of Lemma~\ref{lem:splitting}, the top leftmost spider is the important one.
 The rewrite step splits the bottom spider to separate the first H-box from the others, as the first H-box is the only one that is connected to the leftmost spider via an X node:
  \ctikzfig{splitting-ex1}
 Now, applying the lemma to the middle spider would have no effect as all H-boxes are connected to the middle spider via the same kind of wire (namely, simple ones).
 Thus, it remains to apply Lemma~\ref{lem:splitting} according to the top rightmost spider.
 This step splits the bottom spider to separate the second H-box from the others because the second H-box is connected to the rightmost spider via an X node while the others are connected by simple wires:
  \ctikzfig{splitting-ex2}
\end{example}

\begin{lemma}\label{lem:big-disconnect}
    Disconnect lemma:\footnote{A similar but incorrectly scaled result is given as Lemma~B.3 in \cite{backens2018zhcalculus}. Since any application of the lemma in that paper is always combined with an also incorrectly-scaled application of the average rule, the scalars work out there overall anyway. We give the correct scalar factors here.}
  \begin{equation}\label{eq:contraction-sep}
    \input{contraction-sep.tikz} \quad\ =\ \input{contraction-sep-rhs.tikz}
  \end{equation}
\end{lemma}
\begin{proof}
If $n=1$, the !-boxes on both sides are indexed over the one-element set, so the result follows straightforwardly from the definition of annotated !-boxes in Section~\ref{sec:annotated-bb} and removal of the white spider via \IDRule. Note that there are zero stars in this case.

For $n>1$, by Lemma~\ref{lem:annotated-expansion}, we can equivalently express the LHS as
\[
\scalebox{0.83}{\input{contraction-msb-index1.tikz}}
\ \namedeq{\eqref{eq:labelledHboxhigherarity}}
\scalebox{0.83}{\input{contraction-msb-index2.tikz}}
\]
We can then apply Lemma~\ref{lem:splitting} to split the bottom spider in two:
\[
\cdots \ \ \namedeq{\ref{lem:splitting}}\
\scalebox{0.8}{\input{contraction-msb-index3.tikz}}
\ \namedeq{\eqref{eq:labelledHboxhigherarity}}\
\scalebox{0.8}{\input{contraction-msb-index3p.tikz}}
\]
By Lemma~\ref{lem:annotated-split}, we can then split the !-box into two parts, indexed over the same set to obtain:
\[
\scalebox{0.9}{\input{contraction-msb-index4.tikz}}
\]
We now have two copies of a graph which is very similar to the LHS of \eqref{eq:contraction-sep}, but for one fewer bit. We can thus repeat the process above to split each of the two spiders using the second bit of the bitstring, then the third, and so on, until $\whitedot$-spiders only connect pairs of $H$-spiders that disagree on the least significant bit.

Each application of Lemma~\ref{lem:splitting} introduces a star, so we need to count how many times that lemma is used in total.
For the first bit of $\mathbb{B}^{n-1}$, the lemma is applied once.
Afterwards, the diagram has separated into two indexed !-boxes and Lemma~\ref{lem:splitting} needs to be applied to each part.
Continuing in this way, the number of !-boxes which need to be split doubles in each step.
Hence we introduce $1+2+\ldots+2^{n-2} = 2^{n-1}-1$ stars.

Note that at the end of this process, each !-box will be indexed over a one-element set, so according to the definition of annotated !-boxes we can simply drop them.
By replacing 2-legged $\whitedot$-spiders with cups  using \IDRule{}, and applying the definition of annotated !-boxes to re-introduce indexing over $\mathbb{B}^{n-1}$, we obtain the RHS of \eqref{eq:contraction-sep}.
\end{proof}

\begin{lemma}\label{lem:nf-avg}
 For any $a_{\vec{b}}$ we have
 \ctikzfig{nf-avg-lemma}
\end{lemma}
\begin{proof}
First, we introduce a scalar $H(2)$ and get this in the !-box using Lemma~\ref{lem:intro-bangboxed} (the !-boxed form of the Intro rule):
\[\scalebox{0.9}{\input{nf-avg-lemma-proof-1.tikz}}\]
And then we prepare the diagram to apply Lemma~\ref{lem:cancel-two-half}:
\[\input{nf-avg-lemma-proof-2.tikz}\]
Note that, in the last step, \eqref{eq:labelledHboxhigherarity} introduces a scalar white dot in each !-box, resulting in $2^n$ white dots, one of which is cancelled by the $\begin{tikzpicture}
	\begin{pgfonlayer}{nodelayer}
		\node [style=st] (3) at (0, 0) {};
	\end{pgfonlayer}
\end{tikzpicture}
$ outside the !-box.
\end{proof}

\begin{proposition}\label{prop:contraction}
 The diagram resulting from applying \whitecounit\ to an output of a normal form diagram can be brought into normal form, when we have proven~\eqref{eq:average-def} for all $a$ and $b$ in the normal form:
 \ctikzfig{whitecounit-nf}
\end{proposition}
\begin{proof}
  Starting from an arbitrary normal form, with a \whitecounit plugged into the right most output, we first expand the annotated !-box:
  \[ \input{contraction-thm-pf-prime.tikz} \]
  Now the diagram is ready for application of Lemma~\ref{lem:big-disconnect}, followed by the average rule:
  \[\scalebox{0.9}{\input{contraction-thm-pf2-prime.tikz}}\]
  This diagram can now be brought to normal form by application of Lemma~\ref{lem:nf-avg}.
  Note the stars exactly cancel the white dots introduced by that lemma.
\end{proof}

Our strategy will now be to show that any diagram can be decomposed into H-boxes, combined via the operations of extension, convolution, and contraction. This will give our completeness proof.

To simplify the decomposition of diagrams into H-boxes, we prove a few corollaries.

\begin{corollary}\label{cor:whitemult-nf}
 The diagram resulting from applying \whitemult\ to a pair of outputs of a normal form diagram can be brought into normal form, when we have proven~\eqref{eq:intro-def} and~\eqref{eq:average-def} for all $a$ and $b$ in the normal form.
 \begin{equation}\label{eq:whitemult-nf}
    \input{whitemult-nf.tikz}
 \end{equation}
\end{corollary}
\begin{proof}
Applying a \whitemult\ to a pair of outputs has the same result as convolving with a cup, then contracting one of the outputs. That is, we can decompose \eqref{eq:whitemult-nf} as follows:
\ctikzfig{whitemult-decomp}
then apply Lemma~\ref{lem:cup-nf} and Propositions \ref{prop:extension}, \ref{prop:convolution}, and \ref{prop:contraction}.
\end{proof}

\begin{corollary}\label{cor:cap-nf}
A diagram consisting of a cap applied to a normal form diagram can be transformed into a normal form diagram, when we have proven~\eqref{eq:intro-def}, \eqref{eq:average-def} and $(A_{a', b'})$ for all $a$ and $b$ in the normal form and for every $a'$ and $b'$ that are sums of labels in the normal form:
\ctikzfig{cap-nf}
\end{corollary}
\begin{proof}
Since the cap can be decomposed as $\whitecounit \circ \whitemult$, the result follows immediately from Corollary~\ref{cor:whitemult-nf} and Proposition~\ref{prop:contraction}.
\end{proof}

Thanks to Corollaries~\ref{cor:tensor-product} and \ref{cor:cap-nf}, we are able to turn any diagram of normal forms into a normal form. It only remains to show that the generators of the ZH-calculus can themselves be made into normal forms. We have already shown the result for H-boxes (Lemma~\ref{lem:H-box-nf}), so the following will be sufficient:

\begin{lemma}\label{lem:Z-spider-nf}
Any Z-spider can be brought into normal form using the rules of the ZH-calculus.
\end{lemma}
\begin{proof}
As shown in Eq.~\eqref{eq:H-box-1}, \whiteunit{} is already a labelled H-box and thus is in normal form.
By \spiderrule, $\whitedot = \input{dot-nf.tikz}$, and hence it can be brought into normal form using Corollaries~\ref{cor:tensor-product} and \ref{cor:cap-nf}.
This covers the cases of Z-spiders with 0 or 1 incident wires.

We can decompose any Z-spider with $n\geq 2$ incident wires as a tensor product of $(n-1)$ cups, with each cup \whitemult-ed with its neighbours:
\ctikzfig{n-ary-Z-decomposition}
If $n=2$, no \whitemult are needed and the equality is by \IDRule instead of \SpiderRule.
In either case, the diagram can be brought into normal form by applying Lemma~\ref{lem:cup-nf} and Corollaries~\ref{cor:tensor-product} and \ref{cor:whitemult-nf}. The intermediate normal forms only involve H-boxes labelled by a $0$ and $1$ and hence we have already proven the necessary cases of multiply, intro and average to prove this.
\end{proof}

\begin{proposition}\label{prop:completeness-conditional}
If the ZH-calculus proves \eqref{eq:mult-def}, \eqref{eq:intro-def}, and \eqref{eq:average-def} for all integers $a$ and $b$, then the ZH-calculus is complete, i.e.\ for any two ZH-diagrams $D_1$ and $D_2$, if $\llbracket D_1 \rrbracket = \llbracket D_2 \rrbracket$ then $D_1$ is transformable into $D_2$ using the rules of the ZH-calculus.
\end{proposition}
\begin{proof}
By Theorem~\ref{thm:nf-unique}, it suffices to show that any ZH diagram can be brought into reduced normal form. Lemmas~\ref{lem:H-box-nf} and~\ref{lem:Z-spider-nf} suffice to turn any generator into normal form. Then using Corollary~\ref{cor:tensor-product} we can turn any tensor product of generators into a normal form and with Corollary~\ref{cor:cap-nf} we can then apply any sort of wiring and reduce it to normal form. Finally, Proposition~\ref{prop:normal-to-reduced} allows us to bring the normal form to reduced normal form.
\end{proof}

To prove completeness, it now remains to prove \eqref{eq:mult-def}, \eqref{eq:intro-def}, and \eqref{eq:average-def} for all integers $a$ and $b$. To assist in this, let us prove a few lemmas that will help reduce diagrams to normal form when we have these rules for certain values of $a$ and $b$.
As we already have many individual cases of those rules --- see Lemma~\ref{lem:mult-simple-values} for~\eqref{eq:mult-def} and Lemma~\ref{lem:intro-phasefree} for~\eqref{eq:intro-def}, as well as Section~\ref{s:average} for~\eqref{eq:average-def} --- these lemmas can already be applied to many diagrams.

\begin{lemma}\label{lem:Hadamard-nf}
 Suppose we have $(I_{a_{b\vec{c}}})$ for all $b\in\mathbb{B}$, $\vec{c}\in\mathbb{B}^{n-1}$.
 Then we can consume a Hadamard gate into the normal form:
 \ctikzfig{hadamard-lemma-prime}
 If we also have $A_{a_{0\vec{c}},a_{1\vec{c}}}$ and $A_{a_{0\vec{c}},-a_{1\vec{c}}}$ for all $\vec{c}\in\mathbb{B}^{n-1}$, then in fact:
 \ctikzfig{hadamard-lemma}
\end{lemma}
\begin{proof}
 We start by manipulating the !-boxes and layout, and then perform the rewriting:
 \ctikzfig{hadamard-lemma-proof1}
 \ctikzfig{hadamard-lemma-proof2}
 \ctikzfig{hadamard-lemma-proof3}
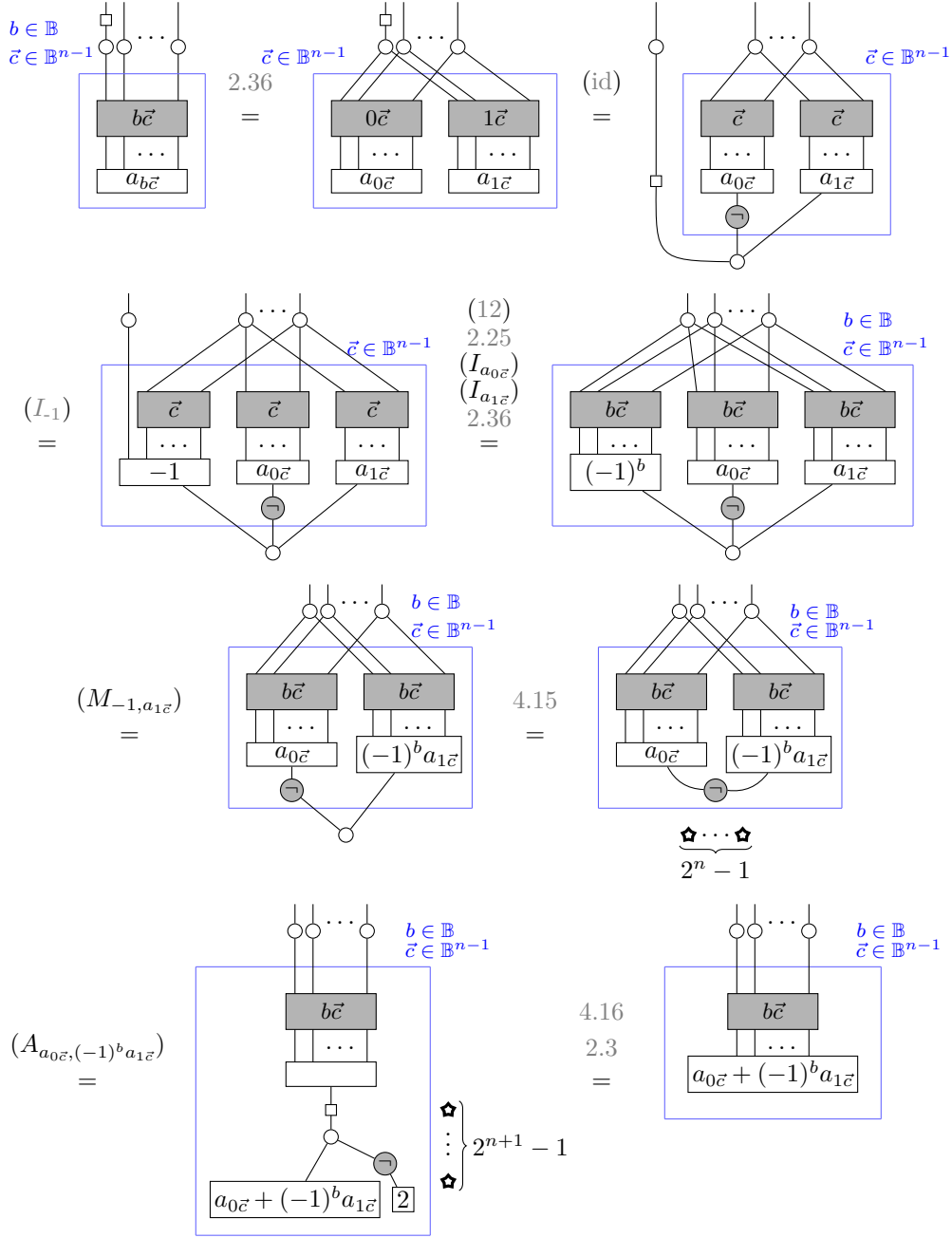
 \[\input{hadamard-lemma-proof4.tikz} \qedhere\]
\end{proof}

The following lemma allows us to (dis)connect integer-labelled H-boxes that are connected to a superset of a given $0$-labelled H-box.
\begin{lemma}\label{lem:AND-1}
    For any integer $a$,
    \ctikzfig{AND-lemma-1-ext}
\end{lemma}
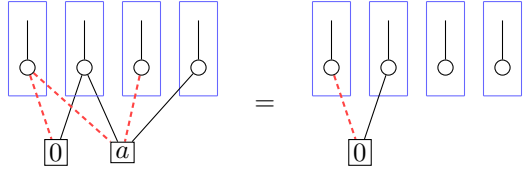
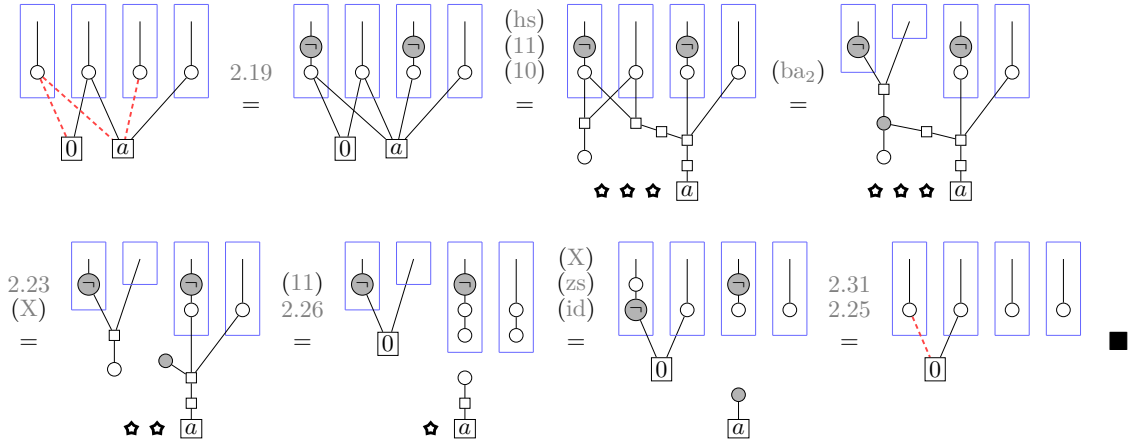
\begin{proof}~
    \[\scalebox{0.9}{\input{AND-lemma-1-ext-proof.tikz}} \qedhere\]
\end{proof}

\begin{lemma}\label{lem:zero-box-nf}
 Suppose we have proven $(I_{a_{b\vec{c}}})$ for all $b\in\mathbb{B}$, $\vec{c}\in\mathbb{B}^{n-1}$ and $(A_{a_{0\vec{c}},a_{1\vec{c}}})$ for all $\vec{c}\in\mathbb{B}^{n-1}$, then we can consume an $H(0)$ into the normal form:
 \ctikzfig{zero-box-nf2}
\end{lemma}
\begin{proof}
 The proof is analogous to that of Lemma~\ref{lem:Hadamard-nf}, using Lemma~\ref{lem:AND-1} instead of $(M_{-1,a_{1\vec{c}}})$.
\end{proof}

\section{Arithmetic}\label{sec:arithmetic}

In this section we show we can do simple arithmetic with labelled H-boxes.
Namely we will show first that the following equation holds for all integers $a$ and $b$.
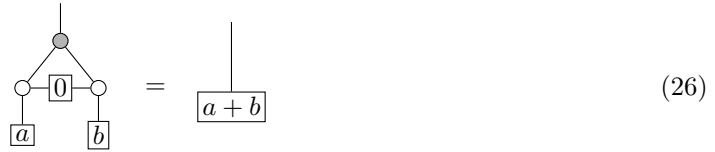
\begin{equation}\label{eq:addition-gadget}
    \input{H-box-addition.tikz}
\end{equation}
Hence, we can add H-box labels together. This will be crucial for when we will prove the general average rule in Section~\ref{sec:completeness}. Secondly, we will prove the multiply rule for all integers.

\subsection{Proving addition for natural numbers}

First we will prove \eqref{eq:addition-gadget} for natural numbers, which requires showing that \eqref{eq:addition-gadget} acts as expected when $b=1$, and that furthermore the `addition gadget' of \eqref{eq:addition-gadget} is associative. For this we need a couple of lemmas.

The next lemma, while looking deceptively simple, has a quite involved proof, requiring repeated use of Lemma~\ref{lem:splitting}. It was found by translating both sides to normal form, and then simplifying the intermediate steps.
\begin{lemma}\label{lem:dedup}~
    \ctikzfig{had-Z-cancel}
\end{lemma}
\begin{proof}
\[\input{dedup-proof-1.tikz}\]
\[\input{dedup-proof-2.tikz}\]
\[\input{dedup-proof-3.tikz}\]
\[\input{dedup-proof-4.tikz}\]
\[\input{dedup-proof-5.tikz} \qedhere\]
\end{proof}

\begin{multicols}{3}
\begin{lemma}\label{lem:zero-double-connection}~
    \ctikzfig{zero-double-connection}
\end{lemma}

\begin{lemma}\label{lem:zero-gray-push}~
    \ctikzfig{zero-gray-push}
\end{lemma}

\begin{lemma}\label{lem:zero-triangle-disconnect}~
    \ctikzfig{zero-triangle-disconnect}
\end{lemma}
\end{multicols}

\begin{proof}[Proof of Lemma~\ref{lem:zero-double-connection}]
    \[\input{zero-double-connection-pf.tikz} \qedhere\]
\end{proof}

\begin{proof}[Proof of Lemma~\ref{lem:zero-gray-push}]
	\[\scalebox{0.9}{\input{zero-gray-push-proof.tikz}}\]
\end{proof}

\begin{proof}[Proof of Lemma~\ref{lem:zero-triangle-disconnect}]
	\[\input{zero-triangle-disconnect-proof.tikz} \qedhere\]
\end{proof}

\begin{lemma}\label{lem:addition-1-is-successor}
    Adding 1 is the same as applying the successor.
    \ctikzfig{addition-1-is-successor}
\end{lemma}
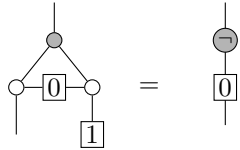
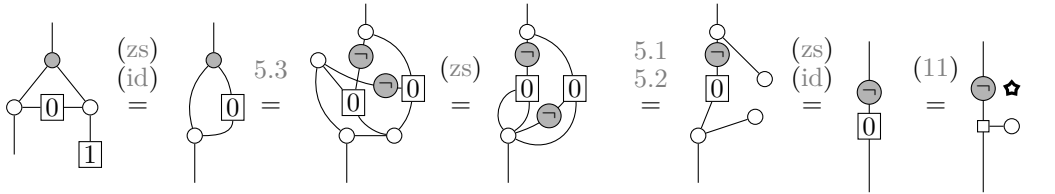
\begin{proof}
    \[\input{addition-1-is-successor-pf.tikz} \qedhere\]
\end{proof}

\begin{lemma}\label{lem:addition-associative}
    Addition is associative.
    \ctikzfig{addition-associative}
\end{lemma}
\begin{proof}
    We can reduce the LHS to a diagram that is symmetric in the three inputs:
    \[\scalebox{0.9}{\input{addition-associative-pf.tikz}}\]
    By symmetry we can also bring the RHS to this last diagram, and hence the LHS and RHS are equal.
\end{proof}

\begin{lemma}\label{lem:addition-natural}
    Let $a,b\in \N$. Then the following holds:
    \ctikzfig{H-box-addition}
\end{lemma}
\begin{proof}
    We prove this by induction on $b$.
    For any $a$, the case $b=0$ is straightforward:
    \ctikzfig{H-box-add-0}
    Furthermore, for any $a$, Lemma~\ref{lem:addition-1-is-successor} shows the case $b=1$.

    Now, suppose there is some $b$ such that the desired result holds for any $a$, and consider $b+1$.
    Then
    \ctikzfig{H-box-add-proof}
    where the step marked (*) is by `only topology matters', and the step marked (**) is the inductive hypothesis.
    Thus, the result follows by commutativity of addition.
\end{proof}

\subsection{Proving addition for all integers}

To generalise addition to arbitrary integers, we need two more lemmas.
\begin{lemma}\label{lem:triangle-Z}~
  \ctikzfig{triangle-Z}
\end{lemma}
\begin{proof}~
  \ctikzfig{triangle-Z-proof-1}
  \[\input{triangle-Z-proof-2.tikz}\qedhere\]
\end{proof}

\begin{lemma}\label{lem:triangle-inverse}The successor has an inverse.
	\ctikzfig{triangle-inverse}
\end{lemma}
\begin{proof}~
\ctikzfig{triangle-inverse-proof-1}
\ctikzfig{triangle-inverse-proof-2}
\ctikzfig{triangle-inverse-proof-3}
\[\input{triangle-inverse-proof-4.tikz} \qedhere \]
\end{proof}
This lemma actually only shows that the successor operation has a one-sided inverse. But by conjugating the top and bottom using a negate spider, we see that this is in fact a two-sided inverse.

\begin{proposition}\label{prop:addition}
	Let $a,b\in \mathbb{Z}$. Then the following holds:
	\ctikzfig{H-box-addition}
\end{proposition}
\begin{proof}
	If $a,b\geq 0$ this is just Lemma~\ref{lem:addition-natural}.
	If both $a,b\leq 0$ we simply do:
	\ctikzfig{H-box-addition-negative}
	Now suppose $a>0$ and $b\leq 0$, i.e.\ $-b\geq 0$. We prove by induction on $-b\in\N$. If $-b=0$ this is trivial. Suppose it holds for $-b$. In the calculation below we denote the induction step by (*):
	\ctikzfig{H-box-addition-negative2}
\end{proof}

\subsection{Proving multiplication}

In this section we will show that the multiply rule introduced in Section~\ref{sec:mult-rule} holds for all integers:
\begin{equation}\label{eq:mult-gadget}
    \input{multiply-rule-phased.tikz}
\end{equation}

It is in any case clear that this operation is commutative, associative and the unit is the $1$-labelled H-box.
To prove that this indeed acts as multiplication it then suffices to prove it distributes over addition, which requires a bit of set-up to prove.

\begin{multicols}{3}
\begin{lemma}\label{lem:zero-projector}~
    \ctikzfig{zero-projector}
\end{lemma}

\begin{lemma}\label{lem:zero-becomes-two}~
    \ctikzfig{zero-becomes-two}
\end{lemma}

\begin{lemma}\label{lem:zero-push-addition}~
    \ctikzfig{zero-push-addition}
\end{lemma}
\end{multicols}

\begin{proof}[Proof of Lemma~\ref{lem:zero-projector}]
 Starting from the RHS, we apply \introzero to the bottom $H(0)$-box.
 The newly-introduced $H(0)$-boxes can then be removed using Lemma~\ref{lem:AND-1}: first with the top binary $H(0)$-box, then with the one on the right, and finally with the one on the left.
 \ctikzfig{zero-projector-proof}
 After replacing the red dashed lines with NOT gates and applying \idrule, this is equal to the LHS.

\end{proof}
\begin{proof}[Proof of Lemma~\ref{lem:zero-becomes-two}]
 First:
 \ctikzfig{zero-becomes-two-new-pf1}
 Note that -- up to one bent wire and the introduction of trivial $H(1)$-boxes -- the resulting diagram consists of an $H(0)$-box on the left output, applied to a normal form diagram.
 The H-box parameters in the normal-form diagram are
 \[
  a_{b\vec{c}} = \begin{cases}0&\text{if } b=0,\,\vec{c}=11 \text{ or } b=1,\,\vec{c}=10 \\ 1 &\text{otherwise.}\end{cases}
 \]
 By Lemma~\ref{lem:intro-phasefree}, we have both $I_0$ and $I_1$, and by Lemma~\ref{lem:average-phasefree}, we have $A_{x,y}$ for any $x,y\in\{0,1\}$.
 Thus we may apply Lemma~\ref{lem:zero-box-nf} (where we have left out the bent wire for simplicity):
 \ctikzfig{zero-becomes-two-new-pf2}
 For the step denoted (*), note that
 \[
  \begin{pmatrix}1&1\\1&0\end{pmatrix} \begin{pmatrix}1&1&1&0\\1&1&0&1\end{pmatrix}
  = \begin{pmatrix}2&2&1&1\\1&1&1&0\end{pmatrix}
 \]
 so
 \[
  a_{0\vec{c}}+(1-b)a_{1\vec{c}} = \begin{cases} 2 &\text{if } b=0 \text{ and } \vec{c}\in\{00,01\} \\ 0 &\text{if } b=1, \vec{c}=11 \\ 1 &\text{otherwise,}\end{cases}
 \]
 and that $H(1)$-boxes can be left out of the diagram by \eqref{eq:unit} and Lemma~\ref{lem:white-not-cancel}.
 Hence, re-introducing the bent wire, we find:
 \[
  \input{zero-becomes-two-new-pf3.tikz} \qedhere
 \]

\end{proof}
\begin{proof}[Proof of Lemma~\ref{lem:zero-push-addition}]
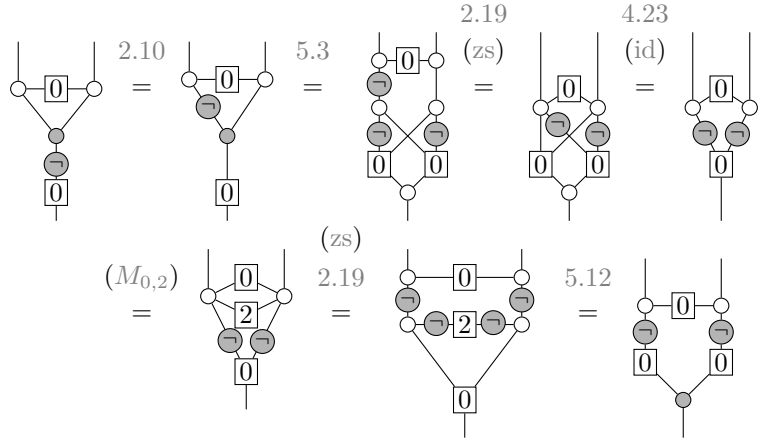
    \[\input{zero-push-addition-pf-1.tikz} \qedhere\]
\end{proof}

\begin{lemma}\label{lem:n-copy}
    For any natural number $n$:
    \ctikzfig{n-copy-lemma}
\end{lemma}
\begin{proof}
    For $n=0$ this is straightforward. We prove by induction, so assume it holds for some $n$. In the calculation below we denote the induction step by (*).
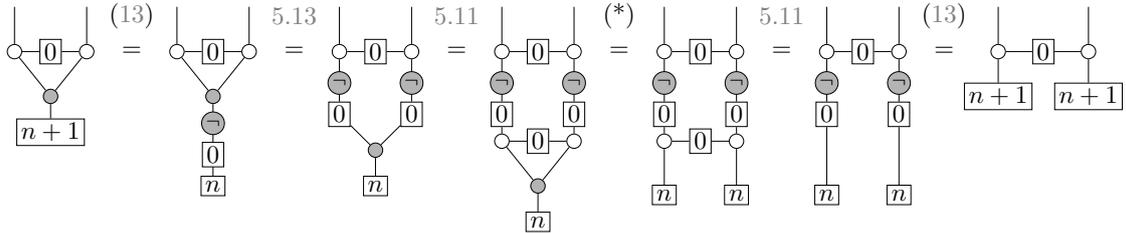
    \[\scalebox{0.95}{\input{n-copy-lemma-proof.tikz}}\]

\end{proof}

\begin{lemma}\label{lem:distributivity-natural}
    The multiplication gadget distributes over addition, i.e.\ for any natural numbers $n\in \N$:
    \ctikzfig{mult-distributes-addition}
\end{lemma}
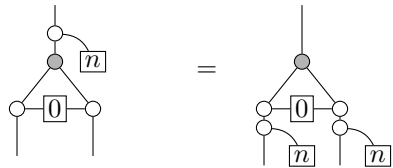
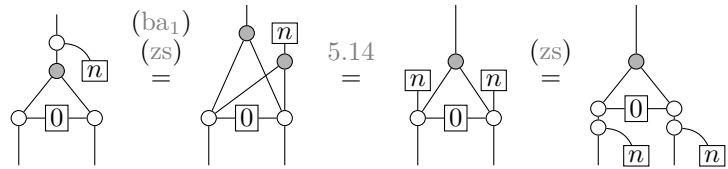
\begin{proof}
	\[\input{mult-distributes-addition-proof.tikz}\qedhere\]
\end{proof}

\begin{proposition}\label{prop:mult-rule-rational}
    The multiply rule holds for any integers $a$ and $b$:
    \ctikzfig{multiply-rule-bb}
\end{proposition}
\begin{proof}
	By Lemma~\ref{prop:multiply-bb} it suffices to show the non--!-boxed version, i.e.\ \eqref{eq:mult-def}.

    First we prove \eqref{eq:mult-def} for $a,b\in \N$.
    The proof is by induction on $b$ for any $a\in\N$.
    The base cases of $b=0$ and $b=1$ have already been shown in Lemma~\ref{lem:mult-simple-values}.

    For the induction step, assume there exists some $b$ such that \eqref{eq:mult-def} holds for any $a$. Then:
    \ctikzfig{mult-rule-inductive}
    where the step labelled (*) uses the base case $b=1$ and the inductive hypothesis.

    Now for negative numbers, we have:
    \ctikzfig{mult-rule-neg}
    and similarly for $b$, so the multiplication reduces to that of non-negative numbers.
    If there are two copies of \whitephase{\neg}, they cancel by Lemma~\ref{lem:znots-cancel}.
\end{proof}

\section{Completeness}\label{sec:completeness}

We write $\zh \vdash D_1 = D_2$ when $D_1$ and $D_2$ can be proven to be equal using the rules of the ZH-calculus. The ZH-calculus is \emph{complete} when $\intf{D_1}=\intf{D_2}$ implies that ${\zh \vdash D_1=D_2}$.

In Proposition~\ref{prop:completeness-conditional}, a `conditional' completeness result was proven, that shows that if the ZH-calculus can prove the average, intro and multiply rule for all integers, then the ZH-calculus is complete. Proposition~\ref{prop:mult-rule-rational} shows that the multiply rule holds for all integers, so that it remains to prove the intro and average rule. That is what we will do in this section.

First, we will establish the average rule.

\begin{lemma}\label{lem:grey-dot-nf}
	The ternary X-spider can be brought to normal form.
\end{lemma}
\begin{proof}
	\begin{equation*}
     \input{grey-dot-nf-proof.tikz}
    \end{equation*}
    where there are three applications of Lemma~\ref{lem:splitting} in the second-to-last step to break up the white spider. This last diagram is in normal form if we introduce the necessary trivial $H(1)$ boxes.
\end{proof}

\begin{proposition}\label{prop:average-integer}
    The ZH-calculus proves \eqref{eq:average-def} for all integers $a$ and $b$.
\end{proposition}
\begin{proof}
    It suffices to show that:
    \begin{equation}\label{eq:avg-as-addition}
    \input{avg-as-addition.tikz}
    \end{equation}
    As then
    \ctikzfig{avg-as-addition-2}
    as required.
    We show this by bringing both sides of Eq.~\eqref{eq:avg-as-addition} to normal form.

    For the RHS of \eqref{eq:avg-as-addition}, first recall that we can bring the X-spider to normal form using Lemma~\ref{lem:grey-dot-nf}, and hence:
    \ctikzfig{avg-proof-rhs}

    For the LHS of \eqref{eq:avg-as-addition}, bending the wires up for simplicity, we first have
    \ctikzfig{avg-proof-lhs1-prime}
    Ignoring the white dot, this is a normal form diagram with $a_{011}=a_{110}=0$, $a_{111}=-1$ and $a_{\vec{b}}=1$ otherwise.
    Lemma~\ref{lem:average-phasefree} proves $(A_{a,b})$ for all combinations of those values, thus we can apply the second part of Lemma~\ref{lem:Hadamard-nf} to show
    \ctikzfig{avg-proof-lhs2-prime}
    Up to a partial transpose (for easier legibility of the matrices), this corresponds to
    \[
     \begin{pmatrix}1&1\\1&-1\end{pmatrix} \begin{pmatrix}1&1&1&0\\1&1&0&-1\end{pmatrix} = \begin{pmatrix}1+1&1+1&1+0&0-1\\1-1&1-1&1+0&0+1\end{pmatrix} = \begin{pmatrix}2&2&1&-1\\0&0&1&1\end{pmatrix}.
    \]
    Again, this is a normal form diagram, this time with $a_{000}=a_{001}=2$, $a_{011}=-1$, $a_{100}=a_{101}=0$ and $a_{010}=a_{110}=a_{111}=1$.
    Now we apply the first part of Lemma~\ref{lem:Hadamard-nf} (note the changes due to the H-box being on the third output):
    \ctikzfig{avg-proof-lhs3-prime}
    We have $a_{\vec{c}0} = \pm a_{\vec{c}1}$ for all $\vec{c}\in\mathbb B^2$, so the two H-boxes in each copy of the !-box contain either equal or opposite integers.
    By Lemma~\ref{lem:intro-phasefree} we have $(I_{a})$ for all $a\in\{-1,0,1,2\}$, so we can use Lemma~\ref{lem:avg-neg}.
    Thus,
    \ctikzfig{avg-proof-lhs4-prime}
    Up to a partial transpose, this sequence of rewriting steps corresponds to
    \[
    \begin{pmatrix}2&2\\1&-1\\0&0\\1&1\end{pmatrix} \begin{pmatrix}1&1\\1&-1\end{pmatrix} = \begin{pmatrix}2+2&2-2\\1-1&1+1\\0&0\\1+1&1-1\end{pmatrix} = 2 \begin{pmatrix}2&0\\0&1\\0&0\\1&0\end{pmatrix}
    \]
    Combining everything and bending back the legs, we find
    \ctikzfig{avg-proof-lhs5-prime}
    This is the same as the RHS.
\end{proof}

It now remains to prove~\eqref{eq:intro-def} for all integers $a$. We can reduce this problem to proving it for natural numbers using the following lemma.

\begin{lemma}\label{lem:intro-mult}
    Suppose we have proven $(I_a)$ and $(I_b)$ for some integers $a$ and $b$. Then we can also prove $(I_{a\cdot b})$. In particular, if we have $(I_a)$, then we also get $(I_{-a})$.
\end{lemma}
\begin{proof}
    We bend all the wires up for a more easy presentation. The proof is then straightforward:
    \ctikzfig{intro-for-mult-pf}
    Because we have $(I_{-1})$ (Lemma~\ref{lem:intro-phasefree}), we get $(I_{-a})$ if we have $(I_a)$.
\end{proof}

\begin{proposition}\label{prop:intro-rational}
    The ZH-calculus proves \eqref{eq:intro-def} for all integers $a$.
\end{proposition}
\begin{proof}
    By the previous lemma it suffices to prove $(I_n)$ for all natural numbers. Note that we have \introzero and \introone so that by induction it suffices to show that we can get $(I_{n+1})$ out of $(I_n)$. We will show that:
    \begin{equation}\label{eq:intro-induction}
    \input{intro-induction-step.tikz}
    \end{equation}
    This is sufficient because then:
    \ctikzfig{intro-induction-step-2}

    We prove Eq.~\eqref{eq:intro-induction} by reducing both sides of the equation to normal form, beginning with the part that is the same on both sides.
    To start, we bring the top part of the diagram into normal form and then apply the first part of Lemma~\ref{lem:Hadamard-nf} to the H-box on the bottom left wire.
    \ctikzfig{intro-ind-nf-proof1}
    For the final step, note that the coefficients $a_{\vec{c}bd}$ of the normal-form diagram are $\pm 1$, indeed $a_{0101}=a_{0111}=a_{1101}=a_{1111}=-1$ and $a_{\vec{c}bd}=1$ otherwise.
    We have $(I_1)$ and $(I_{-1})$ by Lemma~\ref{lem:intro-phasefree}.
    Thus we can use the first equalities from Lemmas~\ref{lem:avg-neg} and~\ref{lem:avg-equal}.
    This transformation corresponds to
    \[
     \begin{pmatrix}1&1&1&1\\1&1&-1&-1\\1&1&1&1\\1&-1&1&-1\end{pmatrix} \begin{pmatrix}1&0&1&0\\0&1&0&1\\1&0&-1&0\\0&1&0&-1\end{pmatrix}
     = \begin{pmatrix}1+1&1+1&1-1&1-1\\1-1&1-1&1+1&1+1\\1+1&1+1&1-1&1-1\\1+1&-1-1&1-1&-1+1\end{pmatrix}
     = 2\begin{pmatrix}1&1&0&0\\0&0&1&1\\1&1&0&0\\1&-1&0&0\end{pmatrix}
    \]
    Let $z_{\vec{c}bd}:=\frac{1}{2}\left(a_{\vec{c}0d}+(-1)^b a_{\vec{c}1d}\right)$, then
    \[
     z_{\vec{c}bd} = \begin{cases} -1 & \text{if } \vec{c}bd=1101 \\ 0 & \text{if } (\vec{c}\in\{00,10,11\}\wedge b=1) \vee (\vec{c}=01\wedge b=0) \\ 1 &\text{otherwise.} \end{cases}
    \]
    For each $\vec{c}$ and $b$, we have $z_{\vec{c}b0} = \pm z_{\vec{c}b1}$, and by Lemma~\ref{lem:intro-phasefree} we have $(I_0)$, $(I_1)$ and $(I_{-1})$.
    Hence we can again apply the first part of Lemma~\ref{lem:Hadamard-nf}, followed by the first equalities of Lemmas~\ref{lem:avg-neg} and \ref{lem:avg-equal}.
    \ctikzfig{intro-ind-nf-proof2}
    This step corresponds to:
    \[
     2 \begin{pmatrix}1&1&0&0\\0&0&1&1\\1&1&0&0\\1&-1&0&0\end{pmatrix} \begin{pmatrix}1&1&0&0\\1&-1&0&0\\0&0&1&1\\0&0&1&-1\end{pmatrix}
     = 2 \begin{pmatrix}1+1&1-1&0+0&0+0\\0+0&0+0&1+1&1-1\\1+1&1-1&0+0&0+0\\1-1&1+1&0+0&0+0\end{pmatrix}
     = 4 \begin{pmatrix}1&0&0&0\\0&0&1&0\\1&0&0&0\\0&1&0&0\end{pmatrix}
    \]
    Combining everything, we have
    \begin{equation}\label{eq:intro-ind-middle}
     \input{intro-ind-nf.tikz}
    \end{equation}
    where
    \[
     w_{\vec{c}bd} := \frac{1}{2}\left( z_{\vec{c}b0} + (-1)^d z_{\vec{c}b1} \right) = \begin{cases} 1 & \text{if } \vec{c}bd \in\{0000, 0110, 1000, 1101\} \\ 0 & \text{otherwise.} \end{cases}
    \]

    Now, for the left-hand side of Eq.~\eqref{eq:intro-induction}, let $y_{\vec{c}bd} := w_{\vec{c}(1-b)(1-d)}$, so that
    \[
     y_{\vec{c}bd} = \begin{cases} 1 & \text{if } \vec{c}bd \in\{0011, 0101, 1011, 1110\} \\ 0 & \text{otherwise.} \end{cases}
    \]
    Note that $y_{\vec{c}0d}y_{\vec{c}1d} = 0$ for all $\vec{c},d$, so we will be able to apply Lemma~\ref{lem:zero-box-nf} because we have already proved $(A_{y_{\vec{c}0d},y_{\vec{c}1d}})$ for all $\vec{c},d$.
    By absorbing the two NOTs on the bottom wires and relabelling, we find:

    \ctikzfig{intro-ind-lhs1}
    This corresponds to
    \[
     4 \begin{pmatrix}0&0&0&1\\0&1&0&0\\0&0&0&1\\0&0&1&0\end{pmatrix}
     \begin{pmatrix}1&0&1&0\\0&1&0&1\\1&0&0&0\\0&1&0&0\end{pmatrix}
     = 4 \begin{pmatrix}0&1&0&0\\0&1&0&1\\0&1&0&0\\1&0&0&0\end{pmatrix}
    \]
    Finally, let $x_{\vec{c}bd} := y_{\vec{c}0d}+(1-b)y_{\vec{c}1d}$, then
    \[
     x_{\vec{c}bd} = \begin{cases} 1 &\text{if } \vec{c}bd \in\{0001,0101,0111,1001,1100\} \\ 0 &\text{otherwise.} \end{cases}
    \]
    Note that $x_{\vec{c}b0}x_{\vec{c}b1}=0$ for all $\vec{c}b$, so we use the same process for the final 0-labelled H-box.
    \ctikzfig{intro-ind-lhs2}
    This rewrite step corresponds to:
    \[
     4 \begin{pmatrix}0&1&0&0\\0&1&0&1\\0&1&0&0\\1&0&0&0\end{pmatrix} \begin{pmatrix}1&1&0&0\\1&0&0&0\\0&0&1&1\\0&0&1&0\end{pmatrix}
     = 4 \begin{pmatrix}1&0&0&0\\1&0&1&0\\1&0&0&0\\1&1&0&0\end{pmatrix}
    \]
    Apart from the two white dots, the left-hand side diagram of Eq.~\eqref{eq:intro-induction} has now been brought into normal form.

    For the right-hand side of Eq.~\eqref{eq:intro-induction}, note that $w_{b0\vec{c}}w_{b1\vec{c}}=0$ for all $b,\vec{c}$, so we can again use Lemma~\ref{lem:zero-box-nf}, before absorbing the NOT and relabelling:
    \ctikzfig{intro-ind-rhs}
    This corresponds to
    \begin{align*}
     4 \begin{pmatrix}0&1&0&0\\1&0&0&0\\0&0&0&1\\0&0&1&0\end{pmatrix} \begin{pmatrix}1&1&0&0\\1&0&0&0\\0&0&1&1\\0&0&1&0\end{pmatrix} \begin{pmatrix}1&0&0&0\\0&0&1&0\\1&0&0&0\\0&1&0&0\end{pmatrix}
     &= 4 \begin{pmatrix}0&1&0&0\\1&0&0&0\\0&0&0&1\\0&0&1&0\end{pmatrix} \begin{pmatrix}1&0&1&0\\1&0&0&0\\1&1&0&0\\1&0&0&0\end{pmatrix} \\
     &= 4 \begin{pmatrix}1&0&0&0\\1&0&1&0\\1&0&0&0\\1&1&0&0\end{pmatrix}
    \end{align*}
    which is equal to the LHS, so \eqref{eq:intro-induction} is proved.
\end{proof}

\begin{theorem}\label{thm:ZH-completeness}
    The ZH-calculus is complete.
\end{theorem}
\begin{proof}
    Propositions~\ref{prop:mult-rule-rational}, \ref{prop:average-integer} and \ref{prop:intro-rational} show that~\eqref{eq:mult-def}, \eqref{eq:average-def} and~\eqref{eq:intro-def} are provable in the ZH-calculus for all integers $a$ and $b$. Hence Proposition~\ref{prop:completeness-conditional} shows that the ZH-calculus is complete.
\end{proof}

\section{The ZH-calculus over arbitrary rings}\label{sec:zh-ring}

The ZH-calculus as defined above represents matrices over $\mathbb{Z}[\frac{1}{2}]$. To handle these efficiently, we introduced integer labels on H-boxes.
In this section, we extend this idea: we allow H-boxes to be labelled \emph{a priori} by elements of some ring $R$.
As a result, diagrams can express matrices in $R$-\textbf{bit},
the full subcategory of $R$-modules where the objects are of the form $(R\oplus R)^{\tensor n}$.
This means our diagrams are interpreted as matrices of shape $2^n\times 2^m$ with entries in $R$.
The conditions on $R$ and the ways in which we show completeness will be the core of this section.
We will dub this new calculus \ZHR, or `ZH over $R$',
with the original phase-free calculus simply called `ZH' without a subscript.
The ruleset of \ZHR\ will be a strict superset of the rules of Figure~\ref{fig:phasefree-rules}, and hence anything we have already proven for ZH will remain true in \ZHR.

The original ZH-calculus as introduced in \cite{backens2018zhcalculus} represented matrices over the complex numbers, and hence in our notation is called \ZHC.
The utility of \ZHC\ is that it can describe evolutions of qubits.
This does not need to be the only useful choice of ring:
The ZW-calculus is similarly parametrised by a commutative ring \cite{hadzihasanovic2017algebra},
and this flexibility is what allowed for the proof of completeness for
Clifford+T ZX via ZW$_{\mathbb{Z}[\half]}$ \cite{SimonCompleteness}\footnote{
    Due to these ideas emerging concurrently the authors of \cite{SimonCompleteness}
    do not explicitly use ZW$_{\mathbb{Z}[\half]}$ but something very similar.
}.

We will show that we can make \ZHR\ sound, complete, and universal into $R$-\textbf{bit} provided that
$R$ is commutative with $1  \neq 0$, and the element $2$ is not a zero-divisor.
Unitarity of the ring is required so that we can preserve the interpretation of the $\whitemult$ generator,
which uses $0$ and $1$.
Commutativity of $R$ is required because we require the following isometry of diagrams to be sound:
\[
    \scalar{hadamard}{a}\ \scalar{hadamard}{b} \simeq \scalar{hadamard}{b}\ \scalar{hadamard}{a}
    \implies \intf{\scalar{hadamard}{a}\ \scalar{hadamard}{b}} = \intf{\scalar{hadamard}{b}\ \scalar{hadamard}{a}}
    \implies a \times b = b \times a
\]

The condition on the element 2 is needed since otherwise we would no longer have a well-defined relationship allowing us to go between white spiders and grey spiders as in \NotDef and \GreyDef.
Likewise we would no longer be able to freely introduce pairs of Hadamard gates via the \HHRule rule.
These relationships are so important that we exclude the case where 2 is a zero-divisor from this paper and leave it to future work.

We will first deal with the case where $\half \in R$
in Section~\ref{sec:ring-with-half},
and afterwards we will adapt the results to work when $\half \not\in R$
in Section~\ref{sec:ring-without-half}.
When $\half\in R$ we will show that $\dstar$ is equal to the $H(\half)$ box with no inputs or outputs, and hence the $\dstar$ generator is superfluous.

The completeness result for when $\half \in R$ uses the rules given in Figure~\ref{fig:ZH-rules},
which are a superset of the rules given in Figure~\ref{fig:phasefree-rules}.
If the ring $R$ does not contain $\half$ but does contain $2$,
which is not a zero-divisor,
then the ring $R[\half]$ is well defined.
We can then use the completeness of $\ZHRhalf$
and the introduction of a meta-rule to eliminate all instances of the $\dstar$
generator to give completeness of \ZHR\ for $\half \notin R$.

Note that the ZH-calculus as introduced in Section~\ref{sec:phase-free-ZH}
is not the same as $\ZH_R$ for any $R$,
since in the ZH-calculus the labels are from $\mathbb{Z}$
but the interpretation is into $\mathbb{Z}[\half]$-\textbf{bit}.

\subsection{ZH-calculus over rings with a half}\label{sec:ring-with-half}

In this section we will take $R$ to be a commutative unital ring that contains $\half$ (and so $2$ is not a zero-divisor).
We generalise the labelled H-boxes of Section~\ref{sec:labelledHboxes} by letting the H-boxes in \ZHR\ be labelled by any element $r\in R$:
\[
 \intf{\bracedSpider{hadamard}{r}} := \sum r^{i_1\ldots i_m j_1\ldots j_n} \ket{j_1\ldots j_n}\bra{i_1\ldots i_m}
\]
As in Section~\ref{sec:ZH-generators}, the sum in the second equation is over all $i_1,\ldots, i_m, j_1,\ldots, j_n\in\{0,1\}$ so that the H-box above represents a matrix with all entries equal to 1, except the bottom right element, which is $r$.
If the label is $-1$ (so that it matches the usual H-box of Section~\ref{sec:ZH-generators}) we will usually leave out the label.
Besides these more general H-boxes, the other generators are identical to those found in Section~\ref{sec:ZH-generators}.
Note that because $\half \in R$, the star generator is redundant:
\[
\intf{\dstar} = \half = \intf{\scalar{hadamard}{\half}}
\]
We will derive the corresponding diagrammatic equality later in Lemma~\ref{lem:star-is-half}.
We could therefore use just the following generators when $\half \in R$:
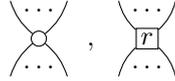
\[
\spider{white dot}{} \ ,\ \spider{hadamard}{r}
\]
But in order to retain compatibility with the previous results, we will continue to write $\!\dstar\!$ to represent the scalar $\half$.

As before, we define the grey spiders and NOT gate of \eqref{eq:defx} as derived generators.
With these generators and derived generators established we can state our theorem of universality.

\begin{theorem}\label{thm:ZHR-universality}
The calculus \ZHR\ is universal over $R$-\textbf{bit},
shown using the unique normal form diagram:
\begin{equation}
\intf{ \, \prod_{\vec{b} \in \mathbb B^n} \big( \iota_{\vec{b}} \circ H_n(a_{\vec{b}}) \big) } =
\sum_{\vec{b} \in \mathbb B^n} a_{\vec{b}} \ket{\vec{b}}.
\end{equation}
\end{theorem}
\begin{proof}  \label{prf:thm:ZHR-universality}
The equation above was demonstrated in the proof of Theorem~\ref{thm:nf-unique}.
The lack of stars in this version comes from noting that every element in $R$ is divisible by $2$ (as $\half \in R$)
and so stars are always subsumed into the ring elements $a_{\vec{b}}$ in the reduced normal form.
\end{proof}

Now let us establish our ruleset for the ZH-calculus over $R$.
The rules for \ZHR\ are given in Figure~\ref{fig:ZH-rules}
along with the same `only topology matters' meta rule and all the same symmetry conditions as before:
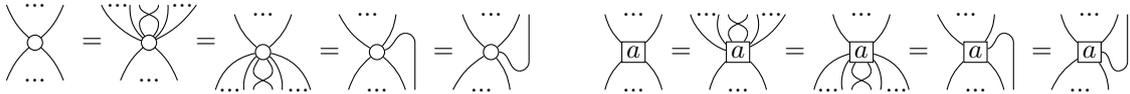
\ctikzfig{generator-symmetries}
Here, $a\in R$ is arbitrary.

\begin{figure}[!ht]
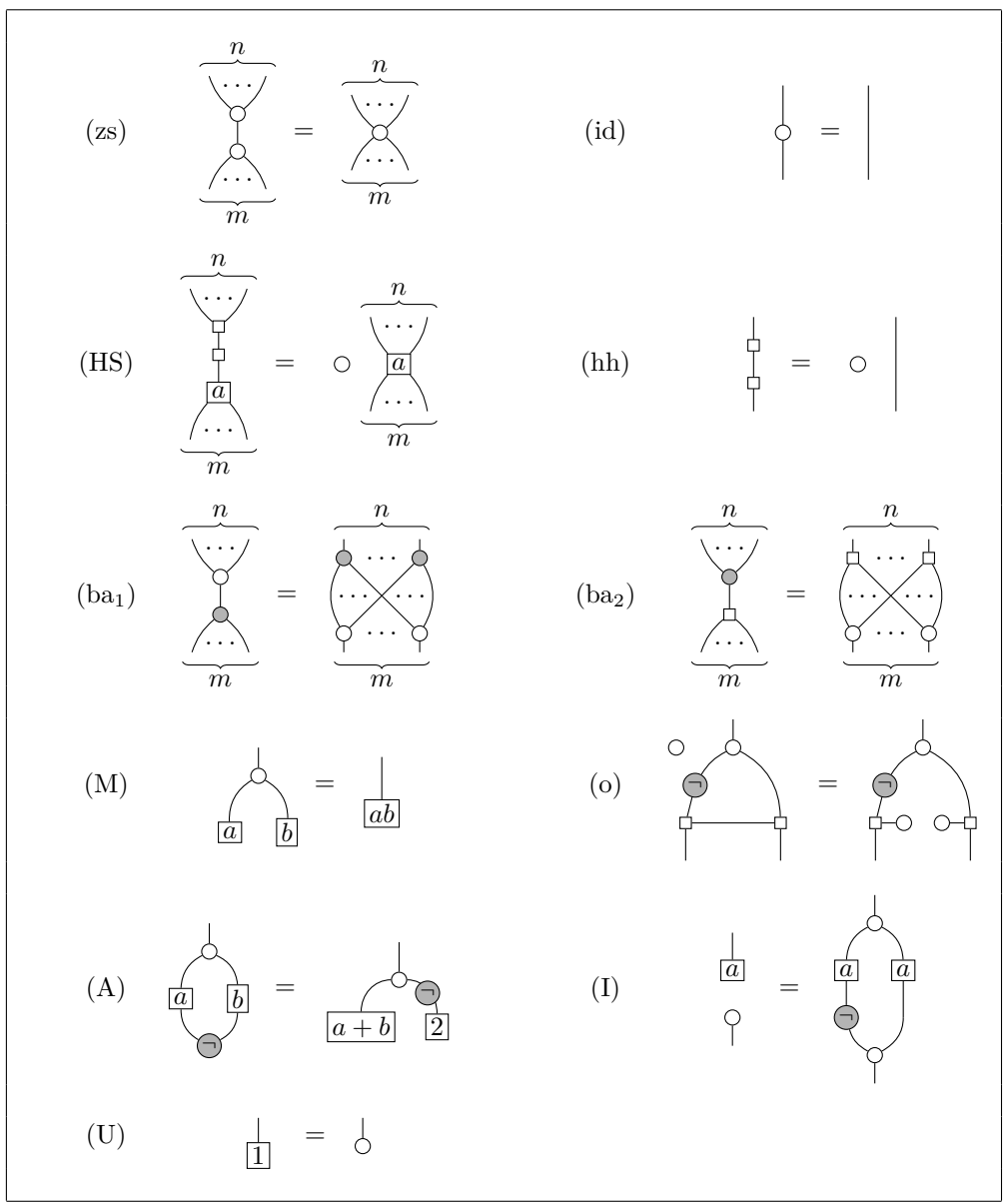

 \centering
 \scalebox{1.0}{
 \begin{tabular}{|ccccc|}
 \hline
 &&&&\\
  \qquad(zs) & \input{Z-spider-rule.tikz} & \qquad\qquad & (id) & \input{Z-special.tikz} \\ &&&& \\
  \qquad(HS) & \input{H-spider-rule-phased.tikz} & & (hh) & \input{H-identity.tikz} \\ &&&& \\
  \qquad(ba$_1$) & \input{ZX-bialgebra.tikz} & & (ba$_2$) & \input{ZH-bialgebra.tikz} \qquad\\ &&&& \\
  \qquad(M) & \input{multiply-rule-phased.tikz} & & (o) & \input{ortho-rule.tikz} \\ &&&& \\
  \qquad(A) & \input{average-rule-prime.tikz} & & (I) & \input{intro-rule.tikz} \qquad\\ &&&& \\
  \qquad(U) & \input{unit-rule.tikz} & & & \\
  &&&&\\
  \hline
 \end{tabular}}
 \caption{The set of rules for the ZH$_R$-calculus.
 Throughout, $m,n$ are nonnegative integers and $a,b$ are arbitrary elements of the commutative unital ring $R$ that contains a half.
 The lowercased rules are the same as those in Figure~\ref{fig:phasefree-rules}.
 The rules (M), (A), (U), and (I) are pronounced \textit{multiply}, \textit{average}, \textit{unit}, and \textit{intro} respectively.
 }
 \label{fig:ZH-rules}
\end{figure}

In comparison to the rules of Figure~\ref{fig:phasefree-rules},
\HPhaseRule has been generalised from \HFuseRule to allow arbitrary labelled H-boxes and
\UnitRule has been added to relate the white dot to an H-box labelled by the $1$ element of the ring.
The rules \MultPhaseRule, \AvgRule and \IntroRule already appeared in Section~\ref{sec:avg-intro-mult} as \eqref{eq:mult-def}, \eqref{eq:intro-def} and \eqref{eq:average-def},
but there they were derived rules that applied to integer labelled H-boxes,
while here they are axioms that relate the ring elements of $R$.
Note that it is not immediately obvious that the $\mathbb{Z}$-labelled H-boxes of Section~\ref{sec:labelledHboxes}
correspond to the $R$-labelled H-boxes of the \ZHR-calculus.
We will prove this imminently, though first we observe the soundness of \ZHR.

\begin{proposition} \label{prop:ZHR-sound}
The rules of Figure~\ref{fig:ZH-rules} are sound.
\end{proposition}

\begin{proof} \label{prf:prop:ZHR-sound}
Because of Proposition~\ref{prop:phasefree-sound}, we simply need to check that the rules not already present in ZH are sound;
these are \HPhaseRule, \MultPhaseRule, \AvgRule, \IntroRule, and \UnitRule.
This is done by straightforward calculation.
\end{proof}

To show how the labelled H-boxes of ZH and \ZHR\ relate,
we first show that $H(1)$ indeed satisfies \eqref{eq:unit}:
\begin{lemma}\label{lem:unit-bb}
 The phase-1 H-box of any arity decomposes:
 \ctikzfig{unit-bangboxed}
\end{lemma}
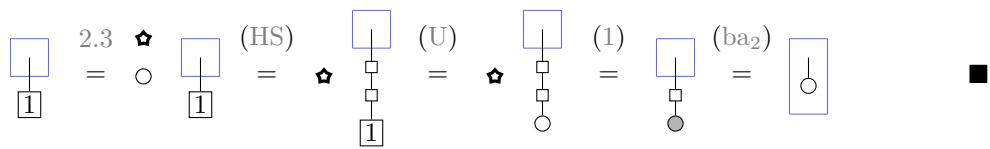
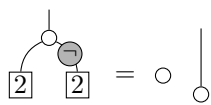
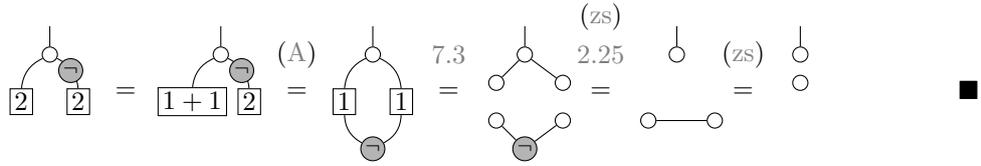
\begin{proof}
 \[
  \input{unit-bb-proof.tikz} \qedhere
 \]
\end{proof}

\begin{lemma}\label{lem:two-cancel}~
	\ctikzfig{two-cancel}
\end{lemma}
\begin{proof}
	\[\input{two-cancel-pf.tikz} \qedhere\]
\end{proof}

Now we can also establish that $H(0)$ and $H(2)$ are as in Section~\ref{sec:labelledHboxes}.
\begin{multicols}{3}
\begin{lemma}\label{lem:zero-is-grey}~
	\ctikzfig{zero-is-grey}
\end{lemma}
\vfill\null
\columnbreak
\begin{lemma}\label{lem:two-as-diagram}~
	\ctikzfig{two-as-diagram}
\end{lemma}
\vfill\null
\columnbreak
\begin{lemma}\label{lem:two-scalar}~
	\ctikzfig{scalar-2}
\end{lemma}
\end{multicols}

\begin{proof}[Proof of Lemma~\ref{lem:zero-is-grey}]
	\[\input{zero-is-grey-pf.tikz} \qedhere\]
\end{proof}

\begin{proof}[Proof of Lemma~\ref{lem:two-as-diagram}]
	\[\input{two-as-diagram-pf.tikz} \qedhere\]
\end{proof}

\begin{proof}[Proof of Lemma~\ref{lem:two-scalar}]
Identical to the proof of Lemma~\ref{lem:scalar-2}.
\end{proof}

With these equivalences established, we will now use induction to show
that the labelled H-boxes in \ZH\ are equivalent (using the rules of \ZHR)
to the corresponding H-boxes in \ZHR.

\begin{proposition} \label{prop:ZH-and-ZHR-same-labels}
All the integer labelled H-boxes of the \ZH-calculus (Section~\ref{sec:labelledHboxes})
are provably equal in \ZHR\ to the corresponding labelled H-box in the ZH$_R$-calculus.
\end{proposition}

\begin{proof} \label{prf:prop:ZH-and-ZHR-same-labels}
We will denote a labelled \ZH\ H-box of Section~\ref{sec:labelledHboxes} by $H(n')$
to distinguish it from a generator of the ZH$_R$-calculus.
We have for every $n \in \{ 1, 2, 3, \dots \} \subset R$:
\ctikzfig{n-induction}
Here the induction step is denoted by (*).
Note that we first used the average rule as an axiom of \ZHR\ acting on the generators (written $\avgrule$),
and then as a derived rule from the original ZH-calculus (Proposition~\ref{prop:average-integer}).
For negative integers we simply note that the definition of negation \eqref{eq:def-negative-numbers}
corresponds precisely to multiplying by $-1$ in \MultPhaseRule.
\end{proof}

As in Section~\ref{sec:avg-intro-mult}, we can prove !-boxed versions of the axioms \MultPhaseRule, \AvgRule and \IntroRule. In particular, the !-boxed version of \MultPhaseRule and \UnitRule (i.e.\ Lemma~\ref{lem:unit-bb}) where the !-box is expanded $0$ times give:
\begin{equation}\label{eq:scalar-cancel}
\begin{tikzpicture}
	\begin{pgfonlayer}{nodelayer}
		\node [style=hadamard] (0) at (-1, -0) {$b$};
		\node [style=hadamard] (1) at (-2, -0) {$a$};
		\node [style=hadamard] (2) at (1, -0) {$ab$};
		\node [style=none] (3) at (0, -0) {$=$};
	\end{pgfonlayer}
\end{tikzpicture} \qquad\qquad\qquad\qquad\qquad\qquad
  \input{one-cancellation.tikz}
\end{equation}
These rules enable us to multiply scalars at will, and in particular eliminate scalars by multiplying by the inverse
(when the inverse exists in $R$).
Finally we will prove that the star generator is equal to the scalar $\half$ in \ZHR.

\begin{lemma} \label{lem:star-is-half}
If $\half \in R$, then $\dstar = \scalar{hadamard}{\half}$.
\end{lemma}
\begin{proof} \label{prfLem:star-is-half}
\[
\scalar{hadamard}{\half}
\namedeq{\ref{lem:scalarcancelstars}}
\dstar \scalar{white dot}{} \scalar{hadamard}{\half}
\namedeq{\ref{lem:two-scalar}}
\dstar \scalar{hadamard}{2}\scalar{hadamard}{\half}
\namedeq{\eqref{eq:scalar-cancel}}
\dstar \scalar{hadamard}{1}
\namedeq{\eqref{eq:scalar-cancel}}
\dstar \qedhere
\]
\end{proof}

With the correspondence to the standard ZH-calculus now firmly established, we can easily adapt our previous results to prove that \ZHR\ is complete.

\begin{theorem}\label{thm:ZHRHalf-complete}
Let $R$ be a commutative unital ring with $\half \in R$.
Then the ZH$_R$-calculus is complete,
i.e.\ for any ZH$_R$-diagrams $D_1$ and $D_2$, if $\llbracket D_1 \rrbracket = \llbracket D_2 \rrbracket$,
then $D_1$ can be transformed into $D_2$ using the rules of Figure~\ref{fig:ZH-rules}.
\end{theorem}
\begin{proof}
As the rules of Figure~\ref{fig:ZH-rules} are a superset of those of Figure~\ref{fig:phasefree-rules}, anything we have proven for ZH remains provable in \ZHR.
We have also shown in Proposition~\ref{prop:ZH-and-ZHR-same-labels} that all the labelled (\ZH) H-boxes used in the completeness proof
are equivalent to their \ZHR-calculus counterparts.
In particular, the reduction to normal form conditional on having proved \eqref{eq:mult-def}, \eqref{eq:intro-def} and \eqref{eq:average-def} of Section~\ref{sec:normal-form-conditional} remains valid.
But now, instead of having to prove the multiplication, introduction and average rules, we have them assumed as axioms.
Thus, we can transform any \ZHR-diagram into reduced normal form, and hence the calculus is complete.

\end{proof}

\subsection{The ZH-calculus over rings without a half} \label{sec:ring-without-half}

The \ZHR-calculus when $R$ does not contain a half is slightly trickier to deal with.
As before, we will assume that $R$ is a commutative unital ring.
But instead of assuming that $\half\in R$, we will merely assume that $2$ is not a zero-divisor.

With these conditions, $R$ embeds faithfully into $R[\half]$,
and so \ZHR\ as defined in Section~\ref{sec:ring-with-half} has a non-universal interpretation into $R[\half]$-\textbf{bit}.
We could then add the $\dstar$ generator
creating a universal interpretation into $R[\half]$-\textbf{bit}.
This is essentially the process one follows to go from $\ZH_\mathbb{Z}$ to our `phase-free' $\ZH$,
but also raises the question of `why not just start with \ZHRhalf\ instead?'
We present in this subsection a definition of \ZHR, with $\half \notin R$,
that is universal for $R$-\textbf{bit} and is complete using the same rules as \ZHR\ when $\half \in R$,
with the exception of an additional `scalar cancellation' meta-rule.

Since we want to have an interpretation into $R$-\textbf{bit}
we will no longer be able to use the $\dstar$ generator,
which in turn makes the definition of the derived generators \GreyDef and \NotDef invalid.
Instead we will promote the grey spiders from being \emph{derived} generators to, simply, generators.
Hence, our list of generators becomes:
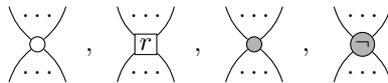
\[
\spider{white dot}{}\ ,\ \spider{hadamard}{r} \ ,\ \spider{gray dot}{} \ ,\ \spider{gray dot}{\neg}
\]

The ruleset of the \ZHR-calculus when $\half\not\in R$ consists of the rules of Figure~\ref{fig:ZH-rules}, where now the grey spiders are considered to be actual generators, and the usual symmetries of the generators.
In addition, as the grey spiders are now no longer defined in terms of the other generators, we introduce the following rules to relate the generators to one another,
which are scaled versions of the definitions of the grey spiders \eqref{eq:defx}:
\begin{equation}\label{eq:def2x}
  \input{X-spider-dfn-free-doubled.tikz} \tag{2X}
\end{equation}
\begin{equation}\label{eq:not2x}
\input{negate-dfn-free-doubled.tikz} \tag{2NOT}
\end{equation}

Finally, we will also need the following cancellation meta-rule to get completeness when $\half \notin R$:
\begin{equation}\label{eq:2cancel}
\whitedot D_1 = \whitedot D_2 \implies D_1 = D_2 \tag{2Cancel}
\end{equation}
In words, this rule says that if two diagrams are provably equal in \ZHR\ and they both contain a scalar white dot, then we can cancel the white dot on both sides and retain an equality.
(When $\dstar$ is in the calculus, this follows easily by an application of Lemma~\ref{lem:scalarcancelstars}.)
We call this implication a meta-rule because it is a statement about admissible deductive logic rather than diagram equality.
Note that since 2 is not a zero-divisor in $R$ this rule is sound.
A similar meta-rule was also used for proving completeness
of rational-angle fragments of the ZX calculus \cite{ZXNormalForm}.

The plan for showing completeness of \ZHR\ when $\half \notin R$
is to show that the \ZHRhalf-derivation from a \ZHR-diagram $D$ to its normal
form can be used as a \ZHR-derivation, up to rescaling.
While this idea is straightforward, formalising it takes some care.

\begin{remark} \label{rem:grey-spider-unpacking-explicit}
Every \ZHR\ diagram can be seen as a \ZHRhalf\ diagram by translating each diagram component to their obvious counterpart.
For grey spiders we need to be careful about this translation,
since they are generators in \ZHR, but derived generators in \ZHRhalf.
We will be using this embedding in the lemmas that follow,
noting that we will still explicitly `unpack' grey spiders
using the rules \NotDef, \GreyDef, \TwoNot and \TwoX when we need to (and not leave such unpacking implicit).
We are treating the $\dstar$ element as syntactic sugar to represent the 0-arity $H(\half)$
box when it arises.
\end{remark}

\begin{lemma} \label{lem:rescaled-derived-generators}
Suppose $D_1$ and $D_2$ are \ZHR-diagrams.
Let $*$ denote either of the \ZHRhalf-rules \NotDef or \GreyDef,
with $2*$ denoting the corresponding rescaled \ZHR\ version \TwoNot or \TwoX.
Then if the \ZHRhalf\ rule application
\[D_1 \namedeq{*} \dstar D_2\]
is a valid diagrammatic rewrite,
so is the \ZHR\ rule application
\[ \whitedot D_1 \namedeq{$2*$} D_2 \]
\end{lemma}

\begin{proof} \label{prf:lem:rescaled-derived-generators}
This follows immediately from the definitions.
\end{proof}

\begin{definition}
\label{def:ZHR-with-stars}
We say that a \ZHRhalf\ diagram is a \emph{$\ZHRstar$} diagram
(pronounced `ZHR with stars' diagram)
if it is of the form $D \tensor (\dstar)^{\tensor n}$, where $D$ is a \ZHR\ diagram.
\end{definition}

Recall that the rules of \ZHRhalf\ are the rules given in Figure~\ref{fig:ZH-rules}
and the rules relating the generators to derived generators: \GreyDef, \NotDef, and \ZDef.

\begin{lemma} \label{lem:leaving-ring-R-for-Rhalf}
The only \ZHRhalf\ rules that
can be applied to a \ZHRstars\ diagram to give
a diagram that is no longer a \ZHRstars\ diagram are
\MultPhaseRule, \AvgRule, and \HPhaseRule.
\end{lemma}
\begin{proof}
The \MultPhaseRule and \AvgRule rules (applied right-to-left)
can both result in H-box phases that are no longer in $R$.
Furthermore, recall from Remark~\ref{rem:grey-spider-unpacking-explicit} that the star element
is used here as short-hand for the 0-arity H-box labelled by $\half$.
Hence $\half$-labelled H-boxes of arity 0 are allowed in \ZHRstar\ diagrams, while $\half$-labelled H-boxes of higher arity are not allowed.
Yet the \HPhaseRule rule, when applied right-to-left to a $\half$-labelled H-box of arity 0,
yields a $\half$-labelled H-box with arity greater than 0.

All other rules preserve \ZHRstar\ diagrams: the only remaining rules that affect H-box labels are \UnitRule, which can only introduce the integer label 1, and \IntroRule, which does not change any labels and does not match H-boxes of arity 0.
The other rules do not involve H-box labels at all.
\end{proof}

\begin{lemma} \label{lem:all-star-boxes-are-arity-zero}
The $\ZHRhalf$ derivation from a \ZHR\ diagram $D$ to its normal form $N$ of Theorem~\ref{thm:ZHRHalf-complete}
is a \ZHRstar\ diagram at every step of the derivation.
\end{lemma}
\begin{proof} \label{prf:lem:all-h0-boxes-are-arity-zero}
The starting diagram $D$ is a \ZHR\ diagram (and therefore also a \ZHRstar\ diagram).
Each step of the derivation to its normal form $N$ is a rule application, and we will show that
every rule application used preserves the condition of being a \ZHRstars\ diagram.
Intuitively, this is because the process for rewriting to normal form was derived in the phase-free ZH-calculus, which differs from $\ZH_\mathbb{Z}^{\dstar}$ only in that grey spiders are syntactic sugar rather than generators.

More formally: by Lemma~\ref{lem:leaving-ring-R-for-Rhalf} the only rule applications that could result
in a diagram that is no longer in \ZHRstars\ form are \MultPhaseRule, \AvgRule, and \HPhaseRule.
As can be verified by going through the steps of the derivation in Section~\ref{sec:normal-forms}, the derivation from $D$ to $N$ only uses \MultPhaseRule and \AvgRule
with phases that are elements of $R$, and hence they preserve the \ZHRstars\ form.

Also note that at no point in the derivation is the rule \HPhaseRule
applied to a star, since in ZH that would result in an invalid $\half$-phased H-box.
\end{proof}

The previous lemma shows that the derivation from a \ZHR-diagram $D$ to its normal form in \ZHRhalf\ consists of diagrams that are close to being \ZHR-diagrams, except for the star generators.
We shall next show that the only rules that interact with $\!\dstar\!$ are \NotDef and \GreyDef,
which are precisely the rules we replace when using $\ZHR$ instead of $\ZHRhalf$.

\begin{lemma} \label{lem:only-use-stars-as-stars}
In the $\ZHRhalf$ derivation from a \ZHR\ diagram $D$ to its normal form $N$ of Theorem~\ref{thm:ZHRHalf-complete},
the element $\dstar$ is only ever involved in a rewrite when the rule being applied is \NotDef or \GreyDef.
\end{lemma}

\begin{proof} \label{prf:lem:only-use-stars-as-stars}
The only \ZHRhalf\ rules that interact with the $\dstar$ element are \NotDef, \GreyDef,
and \HPhaseRule (viewing the star as a 0-arity H-box).
As noted in the proof of Lemma~\ref{lem:all-star-boxes-are-arity-zero}
the \HPhaseRule rule is never applied to a star in the course of the derivation,
since the derivation was designed to be applied to ZH diagrams where such a rewrite would be invalid.
\end{proof}

\begin{lemma} \label{lem:ZHRhalf-nf-derivation-to-ZHR-nf-derivation}
If the \ZHR\ diagram $D$ has normal form $N$, then there is a $m\in \N$ such that the equation $D \tensor (\whitedot)^{\tensor m} = N \tensor (\whitedot)^{\tensor m}$
is derivable in \ZHR.
\end{lemma}
\begin{proof} \label{prf:lem:ZHRhalf-nf-derivation-to-ZHR-nf-derivation}
By Lemma~\ref{lem:all-star-boxes-are-arity-zero} there is a \ZHRhalf-derivation from $D$ to $N$
such that every diagram in the chain is a \ZHRstar\ diagram.
As noted in Remark~\ref{rem:grey-spider-unpacking-explicit},
we are keeping the unpacking of grey spiders explicit,
and so applications of the \ZHRhalf{}-rules \NotDef and \GreyDef are explicit steps in this derivation.

We then simultaneously compose every diagram
in the derivation by the same number of $\whitedot$ generators
such that by Lemma~\ref{lem:rescaled-derived-generators} we can replace all applications of \NotDef with \TwoNot
and all applications of \GreyDef with \TwoX,
and then cancel all remaining instances of a $\dstar$ with a $\whitedot$.
This leaves us with a sound \ZHR-derivation of \ZHR\ diagrams from $D \tensor (\whitedot)^{\tensor m}$ to $N \tensor (\whitedot)^{\tensor m}$ for some $m \geq 0$.
\qedhere

\end{proof}

\begin{theorem}
Let $R$ be a commutative unital ring where $2$ is a zero-divisor and $\half\not\in R$. Then the \ZHR-calculus is complete.
In other words, for any \ZHR-diagrams $D_1$ and $D_2$, if $\llbracket D_1 \rrbracket = \llbracket D_2 \rrbracket$
then $D_1$ can be transformed into $D_2$ using the rules of Figure~\ref{fig:ZH-rules}, the rules \TwoNot and \TwoX, and the meta-rule \TwoCancel.
\end{theorem}
\begin{proof}
The completeness of \ZHRhalf\ gives us the \ZHRhalf-derivations $D_1 = N$ and $N = D_2$.
By Lemma~\ref{lem:ZHRhalf-nf-derivation-to-ZHR-nf-derivation}
there are \ZHR-derivations $(\whitedot)^{\tensor a} \tensor D_1 = (\whitedot)^{\tensor a} \tensor N$
and $(\whitedot)^{\tensor b} \tensor N = (\whitedot)^{\tensor b} \tensor D_2$ for some $a,b\in \N$.
We can therefore derive $(\whitedot)^{\tensor a+b} \tensor D_1 = (\whitedot)^{\tensor a+b} \tensor N = (\whitedot)^{\tensor a+b} \tensor D_2$
in \ZHR\ by multiplying every diagram in the derivations by a suitable number of $\tensor$-products of \whitedot.

By the \TwoCancel meta-rule we can then assert that $D_1 = D_2$ is provable in ZH$_R$.
\end{proof}

We have now shown that \ZHR\ is universal, sound, and complete whenever $R$ is a commutative, unital ring
where $2$ is not a zero-divisor.
In the case $\half \in R$, the calculus and ruleset is very similar to our basic ZH.
When $\half \notin R$ we needed to tweak the generators to avoid using the $\dstar$ generator,
and we needed to introduce a meta-rule for cancelling $\whitedot$ generators.

\section{Modifications to the ZH-calculus}\label{sec:alternative-rules}

In this section, we will discuss a couple of variations on the ZH-calculus. In particular, we present an alternative to \OrthoRule in Section~\ref{sec:o-rule}, and an alternative to \AvgRule and \IntroRule in Section~\ref{sec:alternative-intro-avg}. We show that \HHRule can actually be proved from the other rules in Section~\ref{sec:hh-rule-necessity}. In Section~\ref{sec:tof-had} we present an alternative parameter-free ZH-calculus which is complete over a subring of matrices over $\mathbb{Z}[\frac{1}{\sqrt{2}}]$ which exactly corresponds to those that can be implemented by the Toffoli+Hadamard gate set.

\subsection{Replacing the (o) rule}\label{sec:o-rule}

Of all the basic rules of the ZH-calculus in Figure~\ref{fig:phasefree-rules} the ortho rule \OrthoRule seems a bit out of place. In this section we will see that we can replace it with two other rules, namely Lemmas~\ref{lem:copy-znot-h} and~\ref{lem:dedup}. These two rules can be rephrased as statements about the AND gate.
Namely, up to some simple rewriting, Lemma~\ref{lem:copy-znot-h} says that if we post-select the AND gate with $\bra{1}$ that this post-selection copies through:
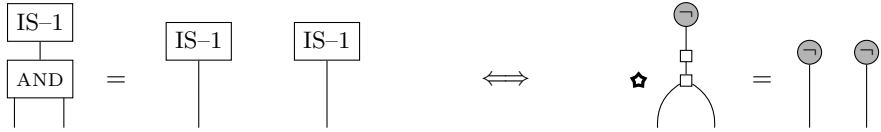
\ctikzfig{AND-postselect}
Lemma~\ref{lem:dedup} can be rephrased as stating that a COPY gate followed by an AND gate applied to its outputs is just the identity:
\ctikzfig{AND-COPY-ID}

Lemmas~\ref{lem:copy-znot-h} and~\ref{lem:dedup} show that these two rules can be proved by the basic ZH-calculus rules including \OrthoRule. In this section we will show the converse: that assuming just the rules of Figure~\ref{fig:phasefree-rules} except for \OrthoRule, together with the Lemmas~\ref{lem:copy-znot-h} and~\ref{lem:dedup} we can prove \OrthoRule.
Note that all the basic lemmas of Section~\ref{s:basic-derived} are proven without the usage of \OrthoRule (except for Lemma~\ref{lem:copy-znot-h}), and hence can be used in this section.

Before proving \OrthoRule, we need to prove a number of lemmas that are essentially statements in Boolean logic.

\begin{multicols}{2}
\begin{lemma}\label{lem:o-rule-alt1}$(x\oplus y)\wedge z = (x\wedge z)\oplus (y\wedge z)$:
	\ctikzfig{lem-o-rule-1}
\end{lemma}
\begin{lemma}\label{lem:o-rule-alt2}$(\neg x)\wedge y = (x\wedge y)\oplus y$:
	\ctikzfig{lem-o-rule-2}
\end{lemma}
\vfill\null
\columnbreak
\begin{lemma}\label{lem:o-rule-alt3}$(\neg x)\wedge x = 0$:
	\ctikzfig{lem-o-rule-3}
\end{lemma}
\begin{lemma}\label{lem:o-rule-alt4}$x\vee y = x\oplus (x\wedge y) \oplus y$:
	\ctikzfig{lem-o-rule-4}
\end{lemma}
\end{multicols}

\begin{proof}[Proof of Lemma~\ref{lem:o-rule-alt1}]
\[\input{lem-o-rule-1-pf.tikz} \qedhere\]
\end{proof}
\begin{proof}[Proof of Lemma~\ref{lem:o-rule-alt2}]~
\[\input{lem-o-rule-2-pf.tikz} \qedhere\]
\end{proof}
\begin{proof}[Proof of Lemma~\ref{lem:o-rule-alt3}]~
\[\input{lem-o-rule-3-pf.tikz} \qedhere\]
\end{proof}
\begin{proof}[Proof of Lemma~\ref{lem:o-rule-alt4}]~
\[\input{lem-o-rule-4-pf.tikz} \qedhere\]
\end{proof}

\begin{lemma}\label{lem:o-rule-alt5}We have $(x\wedge (\neg y))\vee (y \wedge z) = (x\wedge (\neg y))\oplus (y\wedge z)$:
	\ctikzfig{lem-o-rule-5}
\end{lemma}
\begin{proof}~
	\ctikzfig{lem-o-rule-5-pf}
\end{proof}

\begin{theorem}\label{thm:o-alt}
	In the presence of the other rules of the ZH-calculus, \OrthoRule is equivalent to the union of Lemmas~\ref{lem:copy-znot-h} and~\ref{lem:dedup}.
\end{theorem}
\begin{proof}
	The proof of Lemmas~\ref{lem:copy-znot-h} and~\ref{lem:dedup} shows the forward direction: that the standard ruleset of the ZH-calculus, including \OrthoRule, implies the two lemmas.
	For the converse direction, recall that the only one of the basic derived rules in Section~\ref{s:basic-derived} whose proof required \OrthoRule was Lemma~\ref{lem:copy-znot-h}, which is now one of our axioms.
	Note that we have proven Lemma~\ref{lem:o-rule-alt5} using just Lemma~\ref{lem:dedup}, ZH-calculus axioms other than \OrthoRule, and the basic derived rules.
	Proving \OrthoRule from Lemma~\ref{lem:o-rule-alt5}, other ZH axioms, and basic derived rules will therefore give the desired result.
	To make the application of intermediate rules clearer, we re-arrange the diagrams of \OrthoRule so the wire which is usually an output is now the middle input instead:
	\[\input{ortho-proof.tikz} \qedhere\]
\end{proof}

\subsection{The (hh) rule is not necessary}\label{sec:hh-rule-necessity}

It turns out that \HHRule can actually be derived from the other rules. Before proving this, since the proof of Lemma~\ref{lem:scalarcancelstars} uses \HHRule, we will need to find a different way to cancel scalars first. The following turns out to suffice:
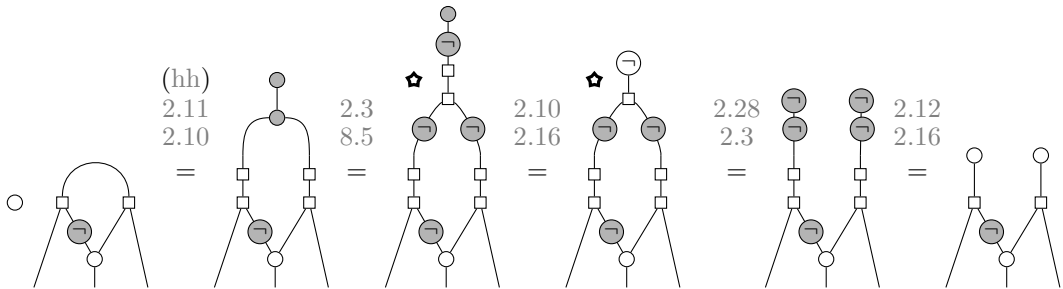
\begin{equation}\label{eq:hhh-scalars}
\input{hhh-scalars.tikz}
\end{equation}

We can now prove \HHRule:
\begin{equation}\label{eq:hh-proof}
  \input{hh-proof.tikz}
\end{equation}

The fact that we can derive \HHRule from the other rules raises the question whether any other rules are also not necessary. We hypothesise that they are in fact all needed.
It is at least the case that \SpiderRule and \HHRule are necessary: they are the only rules that relate higher-arity spiders, respectively H-boxes, to lower arity ones (this argument can be formalised by presenting alternative interpretations into linear maps that change what higher-arity spiders/H-boxes correspond to). \IDRule is also necessary as it is the only rule that relates a generator to the identity wire. At least one of \StrongCompRule and \HCompRule is necessary as they are the only ones that relate an empty diagram to a non-empty diagram (the $n=m=0$ case).
We do not know of any argument for the necessity of the other rules \MultRule and \OrthoRule, although it seems likely that they are both necessary.

\subsection{Merging the intro and average rule}\label{sec:alternative-intro-avg}

The intro rule \IntroRule and the average rule \AvgRule can be subsumed by a single larger rule. Before we present this rule, let us first present an alternative to \AvgRule that holds when the ring contains a half.
\begin{lemma}\label{lem:average-true-form}
	Let $R$ be a ring with a half. In the presence of the other rules in Figure~\ref{fig:ZH-rules}, \AvgRule is equivalent in \ZHR\ to the following rule:
	\begin{equation}\label{eq:average-rule-orig}
	\input{average-rule.tikz}
	\end{equation}
\end{lemma}
\begin{proof}
	As the LHS of \AvgRule agrees with the LHS of \eqref{eq:average-rule-orig}, it suffices to show that the RHS of both are equal.
	Note first that:
	\begin{equation}\label{eq:not-two-is-half}
		\input{not-two-is-half.tikz}
	\end{equation}
	Hence:
	\ctikzfig{average-two-forms}
	Note that this only shows that \AvgRule implies \eqref{eq:average-rule-orig}, as we proved Lemma~\ref{lem:two-cancel} using \AvgRule.

	For the converse direction, we have
	\ctikzfig{average-prf-rev}
	which implicitly uses a straightforward generalisation of Lemma~\ref{lem:copy-xnot-h} to H-boxes with arities other than 3 and with phase labels.
	None of the lemmas invoked in this proof, nor their antecedents, rely on \AvgRule.
\end{proof}

\begin{theorem}[\cite{renaud}]\label{thm:renaud}
	Let $R$ be a ring with a half. In the presence of the other rules of Figure~\ref{fig:ZH-rules}, the combination of the rules  \IntroRule and \AvgRule is equivalent to the following rule:
	\begin{equation}\label{eq:average-renaud}
		\input{average-renaud.tikz}
	\end{equation}
\end{theorem}
\begin{proof}
	As the \ZHR-calculus is complete with the rules presented in Figure~\ref{fig:ZH-rules}, and \eqref{eq:average-renaud} is easily shown to be sound for any $a$ and $b$, the equation \eqref{eq:average-renaud} can be proven using \IntroRule and \AvgRule for any particular choice of $a$ and $b$. Alternatively, by using the ZH-calculus over the polynomial ring in two variables $R[a,b]$, there is also a proof of it using the variables $a$ and $b$ directly.

	It hence remains to show that \eqref{eq:average-renaud} together with (Figure~\ref{fig:ZH-rules})${}\setminus\{\AvgRule,\IntroRule\}$ implies both \IntroRule and \AvgRule.
	We will first prove~\eqref{eq:average-rule-orig}.
  By plugging a white dot into the top wire of \eqref{eq:average-renaud} we get this, except for a scalar factor that will be shown to cancel afterwards:
  \begin{equation}\label{eq:renaud-to-average}
  \scalebox{0.9}{\input{renaud-to-average.tikz}}
  \end{equation}
	Bending the wires up, and taking $a:=\frac{a'+b'}{2}$ and $b:=\frac{a'-b'}{2}$ gives \eqref{eq:average-rule-orig} for $a'$ and $b'$ (up to scalar factors).
  Using this we note that:
	\begin{equation}\label{eq:zero-to-grey}
	 \input{zero-to-grey.tikz}
	\end{equation}
  We can use this to show that the scalar involving $H(0)$ cancels:
  \begin{equation}\label{eq:zero-scalar-cancel}
    \input{zero-scalar-cancel.tikz}
  \end{equation}
  Hence, \eqref{eq:renaud-to-average} and \eqref{eq:zero-to-grey} reduce to what we would expect. We still need to show that a scalar white dot is equal to $H(2)$.
  First we need to know how to decompose $H(\frac12)$:
  \begin{equation}\label{eq:half-to-hboxes}
    \input{half-to-hboxes.tikz}
  \end{equation}
  Now proving the equivalent of Lemma~\ref{lem:two-scalar} is straightforward:
  \ctikzfig{two-scalar-cancel-reprise}
  Combining~\eqref{eq:renaud-to-average} with this now gives us~\eqref{eq:average-rule-orig}.

	Finally, by taking $b=0$ in \eqref{eq:average-renaud}, we can derive \IntroRule:
	\[\scalebox{0.9}{\input{renaud-to-intro.tikz}}\]
	We can now apply the argument of Lemma~\ref{lem:average-true-form} to get \AvgRule.
\end{proof}

\subsection{ZH-calculus as a complete language for Toffoli+Hadamard circuits}\label{sec:tof-had}

The ZH-calculus presented in Section~\ref{sec:phase-free-ZH} is universal and complete for matrices over the ring $\mathbb{Z}[\half]$.
As shown in~\cite{Amy2020numbertheoretic}, the unitary matrices over $\mathbb{Z}[\half]$ correspond to circuits generated by Toffoli and $H\otimes H$, or in other words, circuits consisting of Toffoli gates and an even number of Hadamard gates.
The reason these matrices, and thus the ZH-calculus, are restricted to an even number of Hadamard gates is because of the normalisation of the individual binary H-boxes:
we have $\intf{\,\begin{tikzpicture}
	\begin{pgfonlayer}{nodelayer}
		\node [style=hadamard] (0) at (0, 0) {};
		\node [style=none] (1) at (0, 0.35) {};
		\node [style=none] (2) at (0, -0.35) {};
	\end{pgfonlayer}
	\begin{pgfonlayer}{edgelayer}
		\draw (1.center) to (2.center);
	\end{pgfonlayer}
\end{tikzpicture}
\,} = \left(\begin{smallmatrix}1&1\\1&-1\end{smallmatrix}\right)$, which differs from the unitary Hadamard gate by a factor of $\frac{1}{\sqrt{2}}$.
With all generators being interpreted as matrices over $\mathbb{Z}[\half]$, a diagram containing an odd number of binary H-boxes can hence not be normalised to be unitary.

We chose to define the interpretation of the \begin{tikzpicture}
	\begin{pgfonlayer}{nodelayer}
		\node [style=st] (3) at (0, 0) {};
	\end{pgfonlayer}
\end{tikzpicture}
 generator to be $\half$, because it led to the simplest possible calculus. By redefining the interpretation of this generator, we can make ZH-diagrams represent matrices $M=(\frac{1}{\sqrt{2}})^n A$ where $A$ is a matrix over $\mathbb{Z}$, which correspond precisely to the matrices generated by the Toffoli+Hadamard gate set~\cite{Amy2020numbertheoretic}.
We do this by redefining the \begin{tikzpicture}
	\begin{pgfonlayer}{nodelayer}
		\node [style=st] (3) at (0, 0) {};
	\end{pgfonlayer}
\end{tikzpicture}
 generator to represent $\frac{1}{\sqrt{2}}$ instead of $\frac{1}{2}$. By replacing every occurrence of \begin{tikzpicture}
	\begin{pgfonlayer}{nodelayer}
		\node [style=st] (3) at (0, 0) {};
	\end{pgfonlayer}
\end{tikzpicture}
 by two stars in the definitions of Section~\ref{sec:phase-free-ZH}, and in all the derivations of the preceding sections, all the derivations remain sound under this new interpretation.
In particular, the derivation that each diagram can be reduced to normal form remains true.
For almost any diagram we can also bring the diagram to reduced normal form, with one class of exceptions: if the diagram represents the zero matrix then the reduced normal form should contain no stars. However,
when we reinterpret the star (and thus double the number of occurrences in the rewrite rules) we see that all rewrite rules preserve the parity of the number of stars in the diagram. Hence, the following equality that says that $\frac{1}{\sqrt{2}}\cdot 0 = 0$ is not derivable with those rules:
\begin{equation}
  \text{Using $\intf{\begin{tikzpicture}
	\begin{pgfonlayer}{nodelayer}
		\node [style=st] (3) at (0, 0) {};
	\end{pgfonlayer}
\end{tikzpicture}
} = \frac{1}{\sqrt{2}}$} \qquad
\input{star-zero-rule.tikz}
\end{equation}
Adding this as another rewrite rule solves this edge case and we still get a complete calculus,
but now for matrices of the form $(\frac{1}{\sqrt{2}})^n A$ where $A$ only contains integers.

\section{Conclusion}\label{sec:conclusion}

We studied the ZH-calculus, a graphical language for reasoning about qubit quantum computing. We found a small set of rewrite rules motivated by basic identities between Boolean functions that allowed us to prove completeness of the fragment that corresponds to the Toffoli-Hadamard gateset.
We then extended the ZH-calculus so that it can represent matrices over arbitrary rings where $2$ is not a zero divisor. We found an extended ruleset that is complete for any such ring.
We have argued that our calculus, both the parameter-free one, as well as the one complete over any ring, is the simplest complete graphical calculus for an approximately universal fragment of quantum computing found so far.

Of the rules of the ZH-calculus as presented in Figure~\ref{fig:phasefree-rules}, we showed that \HHRule is actually not necessary, as it can be proven from the others. We believe that the rest are all necessary for completeness.
In Section~\ref{sec:o-rule} we showed that \OrthoRule can be replaced by two smaller and simpler to understand rules. It might perhaps be possible to show that only one of these suffices, which would give us a simpler axiomatisation.
For the \ZHR-calculus we required a meta-rule to cancel scalars when $\half\not\in R$. It is not clear at the moment whether the calculus is complete without this meta-rule.

The phase-free calculus with its translation to Boolean functions elucidates the relationship between classical computation and universal quantum computation using the Toffoli-Hadamard gate set.
Correspondingly, the ZH-calculus over arbitrary rings should be useful for analysing universal quantum computation with multiply-controlled gates as well as quantum computation or error-correction schemes involving hypergraph states.
The Fourier-transform relationship between the ZH-calculus over $\mathbb{C}$ and the ZX-calculus~\cite{GraphicalFourier2019} also points towards the usefulness of combining the ZX- and ZH-calculus, and translating between them.
Indeed,
\cite{east2020akltstates} 
already exhibits a combined calculus called the ZXH-calculus.
Whether on its own or in combination with other graphical languages, the ZH-calculus has applications in a broad range of areas of quantum computing.

\paragraph{Acknowledgements} The authors wish to thank Renaud Vilmart for pointing out \eqref{eq:average-renaud} as an alternative to the intro and average rule, and Patrick Roy for pointing out the derivation~\eqref{eq:hh-proof} that showed that \HHRule is superfluous. The authors would also like to extend their gratitude to the anonymous reviewers that read the paper closely and found several small errors and oversights that we could subsequently fix.
JvdW and AK are supported in part by AFOSR grant FA2386-18-1-4028. For the majority of the work in this article, JvdW is supported by a Rubicon fellowship financed by the Dutch Research Council (NWO).
Additionally, JvdW acknowledges that this project has received funding from the European Union's Horizon 2020 research and innovation programme under the Marie Sklodowska-Curie grant agreement No 101018390.

\bibliographystyle{plainnat}
\bibliography{extract}

\end{document}